\documentclass[11pt]{article}
\usepackage{fullpage}
\usepackage{authblk}
\usepackage{amssymb}
\usepackage{xcolor}
\usepackage{ifthen}
\usepackage{algorithm}
\usepackage{algpseudocode}
\usepackage{amsmath,multicol}
\usepackage{amsthm}
\usepackage{ioa_code}
\usepackage{graphicx}
\usepackage{subcaption}
\usepackage{diagbox}
\usepackage{enumitem}
\usepackage{booktabs}
\usepackage{subcaption}
\usepackage{cleveref}

\makeatletter
\def\mainlistofsymbols{
  \normalsize
  \vspace*{1.5 em}
  \@starttoc{los}
}

\def\partonelistofsymbols{
  \normalsize
  \vspace*{1.5 em}
  \@starttoc{p1los}
}

\def\parttwolistofsymbols{
  \normalsize
  \vspace*{1.5 em}
  \@starttoc{p2los}
}

\def\l@symbol#1#2{\addpenalty{-\@highpenalty} \vskip 4pt plus 2pt
{\@dottedtocline{0}{0em}{8em}{#1}{#2}}}
\makeatother

\newcommand{\newhiddensym}[2]{%
}

\newcommand{\algIOA}[2]{\ifmmode{\text{#1}_{#2}}\else{$\text{#1}_{#2}$}\fi}

\newcommand{\EX}{\ifmmode{\xi}\else{$\xi$}\fi}
\newcommand{\EXF}{\ifmmode{\phi}\else{$\phi$}\fi}

\newcommand{\inter}[1]{
	\ifmmode{\left(\bigcap_{\mathcal{Q}\in#1}\mathcal{Q}\right)}
	\else{$\left(\bigcap_{\mathcal{Q}\in#1}\mathcal{Q}\right)$}
	\fi
}

\newcommand{\op}{\pi}

\mathchardef\mhyphen="2D

\newcommand{\vid}[1]{\ifmmode{\nu_{#1}}\else{$\nu_{#1}$}\fi}

\newcommand{\seen}{\ifmmode{seen}\else{$seen$}\fi}

\newcommand{\maxts}[1]{\ifmmode{maxTS_{#1}}\else{$maxTS_{#1}$}\fi}
\newcommand{\maxtag}[1]{\ifmmode{maxTag_{#1}}\else{$maxTag_{#1}$}\fi}
\newcommand{\maxpair}[1]{\ifmmode{maxMPair_{#1}}\else{$maxMPair_{#1}$}\fi}
\newcommand{\mintag}[1]{\ifmmode{minTag_{#1}}\else{$minTag_{#1}$}\fi}
\newcommand{\maxps}{\ifmmode{maxPS}\else{$maxPS$}\fi}
\newcommand{\conftg}[1]{\ifmmode{confirmed_{#1}}\else{$confirmed_{#1}$}\fi}
\newcommand{\maxconftag}{\ifmmode{\ms{maxCT}}\else{$maxCT$}\fi}

\newcommand{\remove}[1]{} 

\newtheorem{lemma*}{Lemma}
\newtheorem{lemma}{Lemma}
\newtheorem{corollary}{Corollary}
\newtheorem*{claim}{Claim}
\newtheorem{theorem}{Theorem}
\newcommand{\states}[1]{{\textit{states}(#1)}}

\newcommand{\inactions}[1]{{\textit{in}(#1)}}
\newcommand{\outactions}[1]{{\textit{out}(#1)}}

\newcommand{\extactions}[1]{{\textit{ext}(#1)}}
\newcommand{\trans}[1]{{\textit{trans}(#1)}}

\newcommand{\writeop}[2]{{ \textit{write}( #1, #2) }}
\newcommand{\readop}[1]{{ \textit{read}(#1)}}
\newcommand{\Writetr}[1]{{\textit{W}(#1)}}
\newcommand{\Readtr}[1]{{\textit{R}(#1)}}
\newcommand{\rot}{READ transaction}

\newcommand{\rots}{READ transactions}
\newcommand{\wot}{WRITE transaction}
\newcommand{\wots}{WRITE transactions}
\newcommand{\wotsSNOW}{\textbf{W}RITE transactions}

\newcommand{\INV}[1]{{\textit{INV}(#1)}}
\newcommand{\RESP}[1]{{\textit{RESP}(#1)}}

\newcommand{\inv}[1]{{\textit{inv}(#1)}}
\newcommand{\resp}[1]{{\textit{resp}(#1)}}

\newcommand{\prefix}[1]{{ P(#1)}}
\newcommand{\frag}[2]{{ F_{#1}(#2)}}

\newcommand{\frage}[3]{{ #1   \ifthenelse{    \equal{#3}{} } {  \ifthenelse{ \equal{#2}{}}{}{(#2)}   }  { \ifthenelse{ \equal{#2}{}}{}{(#2)}^{(#3)}}     }}
\newcommand{\fragt}[3]{{ #1   \ifthenelse{    \equal{#3}{} } {  \ifthenelse{ \equal{#2}{}}{}{}   } 
		{ \ifthenelse{ \equal{#2}{}}{}{}^{(#3)}}     }}
\newcommand{\fragn}[3]{{ #1   \ifthenelse{    \equal{#3}{} } {  \ifthenelse{ \equal{#2}{}}{}{}   } 
		{ \ifthenelse{ \equal{#2}{}}{}{}}     }}

\newcommand{\suffix}[1]{{ S(#1)}}
\newcommand{\finiteprefix}[2]{{ \sigma_0, a_1, \ldots,  a_{#2}, \sigma_{#2} }}
\newcommand{\finiteprefixt}[2]{{ a_1, \ldots,  a_{#2} }}
\newcommand{\finiteprefixA}[2]{{ P_{#2}}}

\newcommand{\send}[2]{{\textit{send}(#1)_{#2}}}
\newcommand{\recv}[2]{{\textit{recv}(#1)_{#2}}}

\newcommand{\getTagArray}{{\it get-tag-array}}
\newcommand{\writeValue}{{\it write-value}}
\newcommand{\readValue}{{\it read-value}}
\newcommand{\readValuesAndTags}{{\it read-values-and-tags}}

\newcommand{\informReader}{{\it info-reader}}
\newcommand{\informSerializer}{{\it update-coor}}
\newcommand{\updateCoord}{{\it update-coor}}

\newcommand{\getTagArrayTag}{{\sc get-tag-arr}}
\newcommand{\writeValueTag}{{\sc write-val}}
\newcommand{\readValueTag}{{\sc read-val}}
\newcommand{\readValuesTag}{{\sc read-vals}}

\newcommand{\informSerializerTag}{{\sc update-coor}}
\newcommand{\informReaderTag}{{\sc info-reader}}
\newcommand{\updateCoordTag}{{\sc  update-coor}}
\newcommand{\ackTag}{{\sc ack}}

\definecolor{worange}{RGB}{245, 128, 37}

\definecolor{hpurple}{rgb}{0.6, 0.4, 0.8}

\newtheorem{definition}{Definition}[section]

\begin{document}
\title{SNOW Revisited: Understanding When Ideal READ Transactions Are Possible}

\author[1]{Kishori M. Konwar}
\author[2]{Wyatt Lloyd}
\author[2]{Haonan Lu}
\author[3]{Nancy Lynch}
\affil[1]{RLE, MIT}
\affil[2]{Department of Computer Science, Princeton University}
\affil[3]{CSAIL, MIT}
\affil[ ]{{kishori@mit.edu, 
		wlloyd@cs.princeton.edu, 
		haonanl@cs.princeton.edu,
		lynch@csail.mit.edu
	}}


\maketitle

\begin{abstract}
\rots{} that read data distributed across servers dominate the workloads of real-world distributed storage systems.
The SNOW Theorem~\cite{SNOW2016} stated that ideal \rots{} that have optimal latency and the strongest guarantees---i.e., ``SNOW'' \rots{}---are impossible in one specific setting that requires three or more clients: at least two readers and one writer. However, it left many open questions.

We close all of these open questions with new impossibility results and new algorithms.
First, we prove rigorously the result from~\cite{SNOW2016} saying that it is impossible to have a \rot{}s system that satisfies SNOW properties with three or more clients.
The insight we gained from this proof led to teasing out the implicit assumptions that are required to state the results and also, resolving the open question regarding the possibility of SNOW with two clients.
We show that it is possible to design an algorithm, where SNOW is possible in a multi-writer, single-reader (MWSR) setting when a client can send messages to other clients;
on the other hand, we prove it is impossible to implement SNOW in a multi-writer, single-reader (MWSR) setting--which is more general than the two-client setting--when client-to-client communication is disallowed.
We also correct the previous claim in~\cite{SNOW2016} that incorrectly identified one existing system, Eiger~\cite{Lloyd:nsdi2013}, as supporting the strongest guarantees (SW) and whose read-only transactions had bounded latency. Thus, there were no previous algorithms that provided the strongest guarantees and had bounded latency. Finally, we introduce the first two algorithms to provide the strongest guarantees with bounded latency.

\end{abstract}


\section{Introduction}
\label{sec:intro}
Today's web services are built on \textit{distributed storage systems} that provide fault tolerant and scalable access to data.
Distributed storage systems scale their capacity and throughput by \textit{sharding} (i.e., partitioning) data across many machines within a datacenter, i.e., each machine stores a subset of the data.
They also \textit{geo-replicate} the data across several geographically dispersed datacenters to tolerate failures and to increase their proximity to users.

Distributed storage systems abstract away the complexities of sharding and replication from application code by providing guarantees for accesses to data.
These \textit{guarantees} include consistency and transactions.
\textit{Consistency} controls the values of data that accesses may observe
and \textit{transactions} dictate what accesses may be grouped together.
Stronger guarantees provide an abstraction closer to a single-threaded environment, greatly simplifying application code.
Ensuring the guarantees hold, however, often comes with worse performance.
Therefore, the tradeoff between performance and guarantees lies at the heart of designing such systems.


The performance-guarantee tradeoffs that result from replication have been well-studied with several well-known impossibility results~\cite{AW94, Fischer:pds1983,Gilbert:sigact2002, lipton88pram}.
For instance, the CAP Theorem~\cite{Gilbert:sigact2002} proves that system designers must choose either availability during network partitions (performance) or strong consistency across replicas (guarantee).
However, little prior work exists on what performance-guarantee tradeoffs result from sharding.

Understanding the performance-guarantee tradeoff due to sharding is important because user requests are typically handled across many shards but within a single nearby datacenter (replica).
This is particularly true for the reads needed to handle a user request, which are what dominate real-world workloads:
Facebook reported 500 reads for every write in their TAO system~\cite{TAO2013} and
Google reported three orders of magnitude more reads than general transactions for their F1 database that runs on their Spanner system~\cite{Corbett:osdi2012}.
In this work, we focus on clarifying the performance-guarantee tradeoff for reads that results from sharding.
Distributed storage systems group read requests (that each individually accesses a separate shard) into \textit{\rots{}} that together return a consistent, cross-shard view of the system.
Whether a view is consistent is determined by the consistency model a system provides.
The ideal \rots{} would have the strongest guarantees:
They would provide strict serializability~\cite{Papadimitriou79}, the strongest consistency model,
and they could be used in a system that also includes \textit{\wots{}} that group write requests (each to a separate shard) together.
The alternative to the latter property are \rots{} that can only be used in systems that have non-transactional, \textit{simple writes}.

The ideal \rots{} would also provide the best performance.
In particular, they would provide the lowest possible latency because the prevalence of reads makes them dominate the user response times that are aggressively optimized by web services~\cite{latency:shopzilla, latency:amazon,
  latency:search}.
The \textit{optimal latency} for a \rot{} is to match the latency of non-transactional, \textit{simple reads}:
complete in a single round trip of non-blocking parallel requests to the shards that return only the requested data~\cite{SNOW2016}.

\subsection{Previous Results and Open Questions}
The SNOW Theorem was the first result in the sharding dimension that is relevant to \rots{}~\cite{SNOW2016}.
The SNOW Theorem is an impossibility result that proves no \rot{} can provide \textbf{S}trict serializability with \textbf{N}on-blocking client-server communication that completes with \textbf{O}ne response, with only one version of the data, per read in a system with concurrent \wotsSNOW{} (\S\ref{sec:snow}).
It shows there is a fundamental tradeoff between the latency and guarantees of \rots{} that system designers must grapple with, they must pick either the strongest guarantees (S and W) or optimal latency (N and O).

SNOW is trivially possible in systems with a single client or a single server because the single entity naturally serializes all transactions.
The SNOW Theorem shows SNOW is impossible in systems with at least three clients and at least two servers.
It explicitly leaves open the question of the possibility of SNOW in a system with two clients.
In addition, the model used in the prior work implicitly leaves open several questions.
It assumed the three clients were a single writer and multiple readers (SWMR).
This leaves open the possibility of SNOW with multiple writers and a single reader (MWSR).
The SNOW Theorem also implicitly assumed that clients do not directly exchange messages and that write operations in such a system must eventually complete.
This also leaves open the question of whether allowing or disallowing client-to-client (C2C) communication has any impact on the feasibility of \rot{}s with SNOW properties.

In this work, the new impossibility results 
are philosophically similar to other impossibility results---such as FLP~\cite{Fischer:pds1983} and CAP~\cite{Brewer:pdc2000, Gilbert:sigact2002}---in that they help system designers avoid wasting effort in trying to achieve the impossible. That is, the SNOW Theorem identifies a boundary in the design space of \rots{}, beyond which no algorithms can possibly exist. By revisiting SNOW, our work makes this boundary more precise.


\subsection{Our Contributions}
Our work builds on the SNOW Theorem to clarify this fundamental tradeoff by providing a thorough proof for the SNOW Theorem~\cite{SNOW2016} and also, answer the open questions mentioned above. 
First, we formally state the SNOW properties of executions using the I/O automata framework~\cite{Lynch1996}  with the additional requirement, for the W property, that any WRITE must eventually complete (\S\ref{sec:prelim}). Next, we identify and prove some basic results, required in the proofs, for transforming one valid execution to another possible and safe execution in a \rot{} system (\S\ref{sec:summary}). 
Then we present a new, rigorous proof of the impossibility of SNOW with three or more clients even when client-to-client (C2C) communication is allowed (\S\ref{sec:formal_proof}).
Next, we show that the feasibility of implementing an algorithm with SNOW in MWSR (which also includes the two-client system model) depends on whether C2C communication is allowed:
when it is not allowed, SNOW is impossible (\S\ref{subsec:no_snow_no_c2c}); and when it is allowed, SNOW is possible (\S\ref{app:algorithm-a}) for any MWSR setting.

Prior to this work, the Eiger~\cite{Lloyd:nsdi2013} algorithm was previously believed to be the only algorithm that provided a bounded number of non-blocking rounds~\cite{Lloyd:nsdi2013} and guaranteed strict serializability. 
Next, we show this claim is not true by showing that not all execution of Eiger is strictly serializable (\S\ref{sec:eiger}).

Next, after realizing the limits posed by the SNOW Theorem, we ask ourselves whether it is possible to construct \rot{} algorithms with no C2C communication, as in most practical systems, where one of the SNOW properties is relaxed. One obvious candidate property is the ``O" property, where one of the two restrictions (i.e., ``one-round" of communication and ``one-version" of data) can be relaxed. We provide two algorithms for the \emph{multiple-writers multi-reader} (MWMR) setting: the first algorithm $B$ guarantees SNW and the ``one-version" property and completes \rot{}s in two rounds (\S\ref{mwmr_snow_one_version}); the second algorithm, $C$, guarantees SNW and the ``one-round" property but returns up to as many versions of the data as there are concurrent \wot{}s
(\S\ref{mwmr_snow_one_round}). Thereby, making these \rot{}s algorithms with a bounded number of non-blocking rounds and guarantees strict serializability. 
Due to space limitations, we omit most of the proofs and present them in an extended version in arXiv~\cite{konwar2018snow}.

\begin{figure}[t]
  \centering
  \begin{subfigure}[b]{0.49\columnwidth} {
      \begin{center}
\begin{tabular}{@{}c c c@{}}
  \midrule
                   \multicolumn{3}{c}{\phantom{ABCDef}\textbf{Client-to-Client?}}\\
  \textbf{Setting} & \textbf{Yes}        & \textbf{No}\\
  \hline\\[-1.8ex]
  2 clients        & $\checkmark$   & $\times$ \\
  MWSR             & $\checkmark$   & $\times$ \\
  $\ge 3$ clients  & $\times$       & ($\times$)\\
  \midrule
\end{tabular}
      \end{center}
    }
    \caption{Is SNOW possible?}
    \label{fig:contr_table_a}
  \end{subfigure}
  \hfill
  \begin{subfigure}[b]{0.49\columnwidth} {
      \begin{center}
\begin{tabular}{@{}c c c c@{}}
  \midrule          
  & \multicolumn{3}{c}{\phantom{a}\textbf{Rounds}}\\
  \textbf{Versions} & \textbf{1}    & \textbf{2}    & $\boldsymbol{\infty}$ \\
  \hline\\[-1.8ex]
                1   & $(\times)$    & $\checkmark$  & $(\checkmark)$ \\
              $|W|$ & $\checkmark$  &  & \\
  \midrule\\
\end{tabular}
      \end{center}
    }
    \caption{Bounded SNW algorithms.}
    \label{fig:contr_table_b}
  \end{subfigure} 
  \caption{A summary of our new results.  Previous results are marked in parentheses.
    $\times$ indicates we have proved that such \rots{} are impossible.
    $\checkmark$ indicates we have described such a new \rot{} algorithm. 
    $|W|$ is the number of concurrent \wots{}.}
  \label{tbl:result_map}
  \vspace{-1.5em}
\end{figure}


\section{Transactions Processing System}
\label{sec:prelim}
Web services typically have two tiers of machines within a datacenter: a stateless
frontend tier and a stateful storage tier.
The frontends handle user requests by executing application logic that generates
sub-requests to read/write data in the storage tier that shards (or splits) data across many machines.
We refer to the front-ends as the \textit{clients}, the storage machines as
the \textit{servers} and the stored data items as \textit{objects}, to match common terminology.
While web services are typically geo-replicated, we focus on sharding within a datacenter because the reads that dominate their workloads are handled within a single datacenter.

We consider a transaction processing system that comprises a set of read/write objects $\mathcal{O}$,
 where each object 
$o \in \mathcal{O}$ is maintained by a separate server process,
and also another set of processes, we refer to as \emph{clients}, that can initiate transactions, after the previous ones, if any, have completed.
 The system allows two types of transaction: {\sc READ} transaction, a group of read requests for the values stored in some subset of objects in $\mathcal{O}$; and {\sc WRITE} transaction, a group of write requests intending to update the values stored in some subset of objects  $\mathcal{O}$.
 A read-client executes only READ transactions, while write-client executes only WRITE transactions; no client executes both types of transaction.

A typical READ transaction, we denote as $\textit{R}(o_{i_1}, o_{i_2}, \ldots, o_{i_q})$ or in short by $R$, consists a set of individual read requests
  $read(o_{i_1})$,  $read(o_{i_2})$ and  $read(o_{i_q})$ to read values in objects $o_{i_1}, o_{i_2}, \ldots, o_{i_q}$, respectively.
   $read(o)$ denotes a read that intends to read the value of object $o$.
 A typical WRITE, denoted as   
 $\textit{W}((o_{i_1}, v_{i_1}), (o_{i_2}, v_{i_2}), \ldots, (o_{i_p}, v_{i_p}))$ or in short as $W$, consists of a set of which requests to update the values of objects $o_{i_1}, o_{i_2}, \ldots, o_{i_p}$ with  $v_{i_1}, v_{i_2}, \ldots, v_{i_p}$, that are values from the  domans $V_{i_1}, V_{i_2}, \ldots, V_{i_p}$,  respectively.
  
 A read (or write) client initiates a  READ (or WRITE) transaction with an invocation step  $\INV{R}$ (or $\INV{W}$), then it carries out the read or write operations in the transaction; and eventually completes  the transaction with a $\RESP{R}$ (or $\RESP{W}$).
 After the completion of the reads or writes in a transaction the client responds, in the case of $R$, with the values of objects;  and, in the case of $W$, an \emph{ok} status, to the external client.

We assume that the network channels are reliable but asynchronous, i.e., any message sent by a process will eventually arrive at its destination uncorrupted. We assume  local computations are  asynchronous, i.e., local computations at various processes proceed at arbitrary and unpredictable speeds. When a client receives a transaction request, usually from an external client, such as an user's device,  it  executes the transaction, denote by $R$ or $W$,  and finally, responds to the external client with the results.

We model a distributed algorithm using the I/O automata modeling framework (see ~\cite{Lynch1996} for a detailed account).
In the rest of this paper, for any execution of an automaton $\mathcal{A}$, 
$\finiteprefix{k-1}{k}\ldots$, where $\sigma$'s and $a$'s are states and actions,
we use the notation $\finiteprefixt{k-1}{k}\ldots$ which shows only the actions to simplify notation. 
We use the notation $\textit{prefix}(\alpha, a)$ to refer to the finite prefix of any execution $\alpha$ ending with action  $a$ such that $a$ occurs within $\alpha$.
In our model, an individual read, such as $read(o)$, in some read transaction $R$ initiated by some read client $r$ consists of the following sequence of actions: after $\INV{R}$ at $r$, $r$ sends a message $m$ (requesting the value stored in $o$) to a server $s$ via action $\send{m}{r, s}$. When  $s$ receives $m$ via  action $\recv{m}{r, s}$ it sends the value $v$ (stored in $o$) to $r$  via action $\send{v}{s, r}$. Then $read(o)$ completes as soon as $r$ receives $v_j$ via action  $\recv{v}{s, r}$; $R$ completes with $\RESP{R}$ after all the  reads in it are complete.

\subsection{SNOW Properties} 
In this subsection, we define the SNOW properties for a transaction 
processing system.  Namely, we require that any fair execution of the
system satisfies the following
 four properties: $(i)$ \emph{Strict serializability} (S), which means there is a total ordering  of the transactions such that  all transactions in the resulting execution  appear to be processed by  a single machine one at a time; $(ii)$ \emph{Non-blocking operations} (N), which  means that the servers respond immediately to the read requests of a READ transaction without waiting for any input from other processes; $(iii)$ \emph{One response per read} (O), which requires that any read operation  consists of one round trip of communication with a server,  and also, that  the server responds with a message that contains exactly  one version of the  object value; and $(iv)$ \emph{WRITE transactions that conflict} (W) implies the  existence of concurrent WRITE transactions that update the data objects  while READ transactions are in progress reading the same objects.  Below we describe the individual properties of the SNOW properties in more detail.

{\bf Strict serializability (S).} By \emph{strict serializability} (for a formal definition please see~\cite{HW90}), we mean each WRITE or READ transaction appears to the clients to be executed atomically, at some point in an execution  between the invocation and response events.

Next, we describe the non-blocking and one-response 
properties.  Both are defined as properties of read operations to an 
individual object. For the purpose of elucidation, we consider an execution $\alpha$ of a transaction processing system $T$ that has a set of objects $\mathcal{O}$, where there is a READ transaction  $\textit{R}(o_{i_1}, o_{i_2}, \ldots, o_{i_q})$, in short $R$, invoked at some reader $r$, such that $R$ contains a read $read(o_j)$ for some $o_j \in \mathcal{O}$ maintained at server $s_j$.
%

{\bf Non-blocking reads (N).} The \emph{non-blocking} property means that if $r$, a reader-client,  sends any message to any $s_i$ ($s_i$ manages object $o_i$) during the transaction then  $s_i$ can  respond to  $r$ without waiting for any external input event, such as the arrival of messages, any mutex operations, time, etc. This property ensures that {\sc READ} transactions are delayed only due to delay in message delivery between $r$ and $s_i$. We define this property formally as follows.

\begin{definition}[Non-blocking read (N)]
 Suppose in  $\alpha$, following the action $\INV{R}$,  the actions  $\recv{m_j^{r}}{r, s_j}$ and   $\send{v_j}{ s_j, r}$
   corresponding to $read(o_j)$, occurs at $s_j$.  Then there  exists an execution $\alpha'$ of $T$ such that
  \begin{enumerate}
  \item[ $(i)$ ] The execution fragments 
 $\textit{prefix}( \alpha, \recv{m_j^{r}}{r, s_j})$ and $\textit{prefix}( \alpha', \recv{m_j^{r}}{r, s_j})$ are identical, where
 $\textit{prefix}(\alpha, a)$ is the prefix of $\alpha$ ending with $a$.
 \item [ $(ii)$] In $\alpha'$ the action  $\send{v_j}{ s_j, r}$ at $s_j$ occurs after $\recv{m_j^{r}}{r, s_j}$ without any input action in between.
\end{enumerate}
\end{definition}
  
{\bf One-response per read (O).} The \emph{one-response} property requires that each read operation, $read(o_i)$,  during any
READ transaction, completes successfully in one round of client-to-server communication and the \emph{one-version}  states that exactly one version of the value is sent by server $s_i$, that manages $o_i$, to $r$. \emph{One-round} consists of 
a read  request from the client initiating the read operation to the server  and  the response containing value sent by the server.

\begin{definition}[One response per read (O)]
 Suppose in  $\alpha$, the action $\INV{R}$ occur at $r$  then in $\alpha$ there exists  exactly a pair of actions $ \recv{m_j^{r}}{r, s_j})$ and   $\send{v_j}{ s_j, r}$, corresponding to $R$,  occur at $s_j$, such that 
 $v_j$ is the object value of $o_j$.
\end{definition}
If the reads,  of some READ transaction, of a transaction processing system respect the \emph{non-blocking} and \emph{one-response} properties then each read includes one-round trip from client to server,  where the server returns only the requested value as soon as it receives the request. It is worth noting that the READ transaction can complete only after all the $read(\cdot)$s in it complete.

\begin{definition}[Conflicting writes (W)]
 Suppose in  $\alpha$, the action $\INV{W}$, the actions  occur
  at a write client $w$ then there is an action $\RESP{W}$ in $\alpha$ that appears after $\INV{W}$.
\end{definition}

{\bf WRITE transactions that conflict (W).} The \emph{conflicting writes}
 property states that READ transactions complete even in the presence of concurrent WRITE
  transactions, where the write operations  might update some objects that are also being read by read operations in READ transaction. 
  This shows that READ transactions can be invoked at any point, even in the presence of ongoing WRITE
   transactions.  Note that the liveness of any WRITE transaction is not implied by any of 
   the SNOW properties;  however, for useful practical systems   the  WRITE transactions must eventually complete. 
    Therefore, we assume that every WRITE
    transaction eventually completes via the $RESP$ event, and think of this constraint as a part of the W property.  

{\bf The SNOW Theorem.}
\label{sec:snow}
Consider a transaction processing system with an asynchronous network where a set $\mathcal{O}$ of objects are maintained by individual server processes,  with at least one write client and at least two read clients. Then the SNOW Theorem~\cite{SNOW2016} can be stated as follows.


\emph{``For any transaction processing system in an asynchronous setting, with at least one writer and two reader clients, and at least two sharded objects, it is impossible to have an algorithm such that all of its executions guarantee the SNOW properties.''} 




\section{Technical Preliminaries}
\label{sec:summary}
In this section, we present some useful preliminary results and ideas that we will later use to prove the impossibility results.
We assume a simple system with two servers, $s_x$ and $s_y$, denote the stored values as $x$ and $y$, respectively, and either two or three clients.
One of the clients is a writer $w$, which initiates only \wots{}.
One or two of the clients are readers, $r_1$ and $r_2$, which initiate only \rots{}.

{\bf Client-to-client (C2C) communication.}
We consider two types of settings pertinent to communication among the clients: $(i)$ allow C2C communication, where a client can send a message to any other client, and $(ii)$ disallow C2C communications, where a client cannot send any message directly to another client in the system.


Servers $s_x$ and $s_y$ store values for objects $o_x$ and $o_y$, respectively.
the initial values of $o_x$ and $o_y$ are $x_0$ and $y_0$, respectively. Because there is one  object on each server  the server and object identifiers are often used interchangeably to remove redundancy. For instance, we simply say that $s_x$ returns $x_0$ to the client that initiated the transaction, which means that $s_x$ returns the value $x_0$ of object $o_x$ at the end of the READ transaction.  

Our proofs often use a special type of execution fragment, named non-blocking fragments, that represent the \rot{} algorithm is non-blocking and returns one version of each object. The one-round property is captured by allowing only one non-blocking fragment on each server for a \rot{}. Our proof strategy plays non-blocking fragments against the requirements of strict serializability and write isolation under the freedom of network asynchrony. We explain non-blocking fragments and helper notations in the context of execution $\alpha$, of the system described as above, as follows:


%

%


\begin{enumerate}[leftmargin=*]
%
%
\item \emph{Non-blocking fragments.}
 For a \rot{} $R_i$ by reader $r_i$, $i\in \{1, 2\}$, suppose there is a execution fragment that starts with $recv(m_j^r)_{r_i, s_j}$ and ends with
 $send(v_j)_{s_j, r_i}$, both of which  occur  at $s_j$. Moreover, suppose there is  
 no other input action at $s_j$ in this fragment. Then we call this execution fragment 
   a \emph{non-blocking fragment} for $R_i$ at  $s_j$ and denote by $\frage{F_{i,j}}{\alpha}{v_j}$, $j \in \{x, y\}$(Fig.~\ref{fig:exe3_fragments}).
   When the context is clear, we omit the first subscript of $F$. For instance, for a \rot{} $R$, $\frage{F_{x}}{\alpha}{x_0}$ denotes the non-blocking fragment of $R$ on $s_x$.  

%
%
\item Suppose \rot{} $R_i$ completes in $\alpha$. Consider the execution fragment in $\alpha$  between the event 
$\INV{R_i}$  and  whichever of the events  $\send{m_y^{r_i}}{r_i, s_y}$ and  $\send{m_x^{r_i}}{r_i, s_x}$ that occurs later. If all the 
actions in this fragment occur at  $r_i$, then we denote this fragment as $\frage{I_i}{\alpha}{}$ (Fig.~\ref{fig:exe3_fragments}).  


\item Suppose \rot{} $R_i$  completes in $\alpha$. Consider the execution fragments in $\alpha$ that
occurs between the later of the events   $\recv{x}{ s_x, r_i}$ or  $\recv{y}{s_y, r_i}$, i.e., at the point in $\alpha$ when $r_i$ receives responses from both servers,  and the event $\RESP{R_i}$. If all the 
actions in this fragment occur at $r_i$, then we denote this fragment by $\frage{E_i}{\alpha}{x, y}$, where $R_i$ returns the values $(x, y)$  (Fig.~\ref{fig:exe3_fragments}) to the external client.  

\item We use $R(\alpha)$ and $W(\alpha)$ to denote the READ and WRITE transactions in the context of $\alpha$. When the context is clear, we simply use $R$ and $W$.
\item We use the subscript of a returned value to denote the version identifier, which uniquely identifies a version from a totally ordered set. For instance, $x_0$ is the $0^{\textit{th}}$ version of $x$ (the initial value of object $o_x$) on server $s_x$.

\end{enumerate}

\remove{
 \begin{figure}[t]
      \centering
         \includegraphics[width=4in]{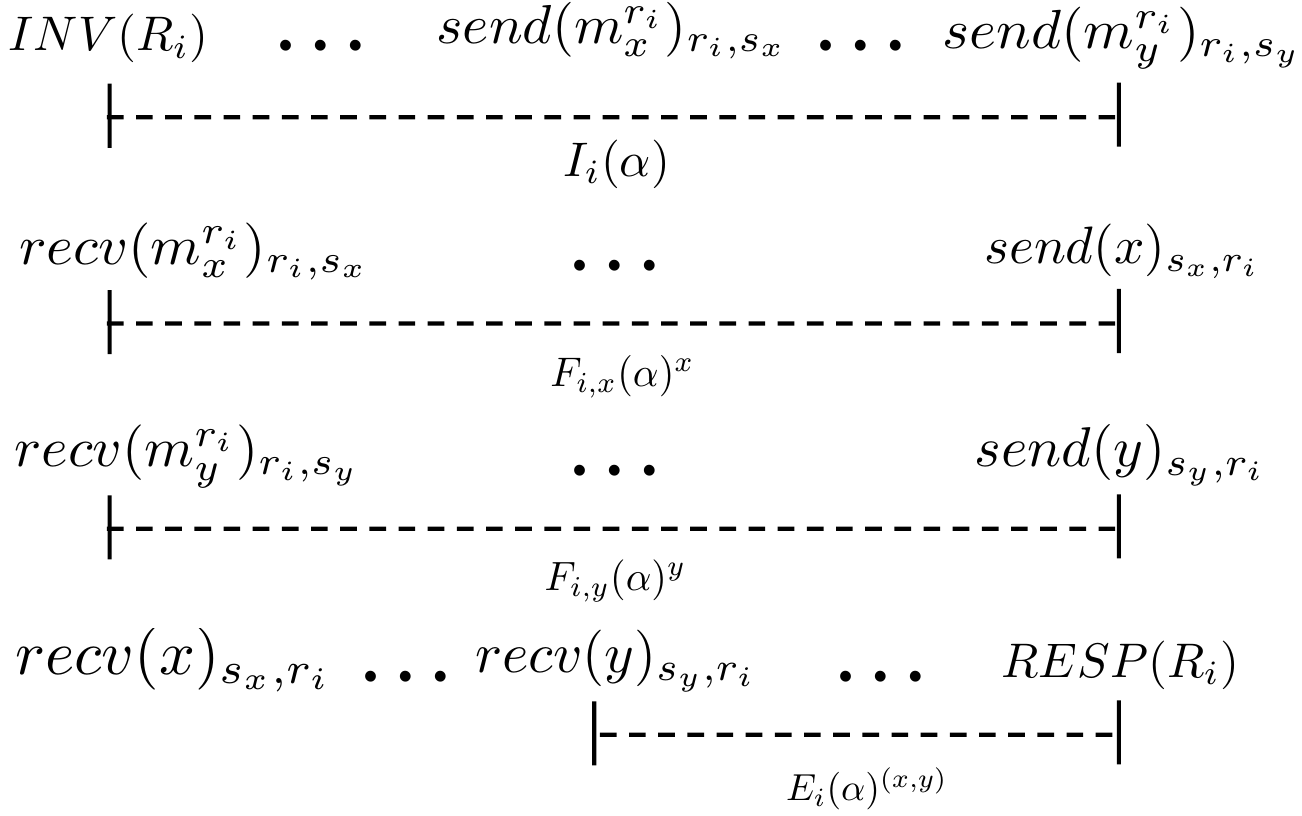}
         \caption{
\small{The relevant actions in the execution fragments $I_i(\alpha)$, $F_{i,x}(\alpha)^{x}$,    $F_{i,y}(\alpha)^{y}$ and   $E_{i}(\alpha)^{(x,y)}$ for any \rot{} $R_i$, $i \in \{1, 2\}$
of a fair execution $\alpha$ of $\mathcal{A}$.}}
         \label{fig:exe3_fragments}
 \end{figure}
}

In our proofs, we frequently use arguments that rely on the existence of  non-blocking fragments and the constraints of strict serializability and write isolation. Below we state a few useful lemmas regarding the executions, of some algorithm $\mathcal{A}$, where all \rots{} are assumed to have all SNOW properties; we will use these lemmas in later sections. Due to space constraints, we explain these lemmas at a high level. 

The following lemma  states that a \rot{} has to return the same version from both servers in order to satisfy strict serializability and write isolation.

\begin{lemma}\label{lem:exec3_consistent} 
Suppose $\alpha$ is any execution of $\mathcal{A}$ such that \rot{} $R$ is in $\alpha$. Suppose the  execution fragment $\frage{I}{\alpha}{} \circ \frage{F_{x}}{\alpha}{x_{t}} \circ \frage{F_{y}}{\alpha}{y_{s}} \circ \frage{E}{\alpha}{x_{t'}, y_{s'}}$ in $\alpha$, corresponds to $R$,  where   $x_t, x_{t'} \in V_1$ and  $y_s, y_{s'} \in V_2$, 
then $s=s'=t=t'$.  
\end{lemma}


\begin{proof}
Suppose $R$ is invoked at reader $r$.  Then, via the action $send(x_t)_{s_x, r}$, in 
execution fragment  $\frage{F_{1}}{\alpha}{x_{t}}$, server $s_x$ sends the value $x_t$ to $r$,   which is received at $r$ through the action 
$recv(x_{t'})_{s_x, r}$ in $\frage{E}{\alpha}{x_{t'}, y_{s'}}$.
 By the assumptions of the reliable channel automata in our model,  we have $x_{t}=x_{t'}$, i.e., $t=t'$. Similar argument for 
$\frage{F_{2}}{\alpha}{y_{s}}$  and $\frage{E}{\alpha}{x_{t'}, y_{s'}}$ leads us to conclude $s=s'$. Next, $R$ responds with $(x_{t'}, y_{s'})$, which implies  by the S property for executions of
 $\mathcal{A}$ that  $x_{t'}$ and $y_{s'}$ must correspond to the same version, i.e., $s'=t'$. 
\end{proof}


\remove{
\begin{corollary}\label{lem:exec3_consistent_corollary} 
Suppose $\alpha$ is any execution of $\mathcal{A}$ such that a \rot{} $R$ is in $\alpha$. Suppose the  execution fragment $\frage{I}{\alpha}{} \circ  X_1 \circ \frage{F_{1}}{\alpha}{x^{t}} \circ X_2 \circ  \frage{F_{2}}{\alpha}{y_{s}} \circ X_3 \circ \frage{E}{\alpha}{x_{t'}, y^{s'}}$ in $\alpha$,  corresponds to $R$,  where   $x_t, x_{t'} \in V_1$ and  $y^s, y^{s'} \in V_2$, $X_1, X_2, X_3$ are some execution fragments that do not contain any action at $r$, $s_x$ or $s_y$, and $s, s', t, t' $ are version identifiers  then $(i)$ $s = s'$  and $ t=t'$ and $(ii)$ $s'=t'$.  
\end{corollary}

\begin{proof} 
The constraint $(i)$ in the statement can be derived  from the fact that $\mathcal{A}$ satisfies the O property,  which implies that any version returned by $R$ for object value $o_1$ (or $o_2$) and this  must be the only version sent by the server $s_x$ (or $s_y$) to $r$. Constraint 
$(ii)$ from the S property of $\mathcal{A}$.
\end{proof}
}

The following lemma states that we can create a new execution $\alpha'$ that is indistinguishable to $\alpha$ by swapping two adjoining fragments, which happen on two distinct automata in $\alpha$ if either $(a)$ both fragments have no input actions or 
$(b)$ one of the fragments have no external (input or output) actions. Our proofs leverage this lemma to create new executions by swapping such fragments and finally derive an execution that violates strict serializability.  
  \begin{lemma}[Commuting fragments] \label{lem:exec3_commute} 
  Let $\alpha$ be an execution of $\mathcal{A}$. Suppose $\frage{G_1}{\alpha}{}$ and $\frage{G_2}{\alpha}{}$ are any  execution fragments in $\alpha$ such that
  all actions in each fragment occur only at one automaton and  either $(a)$ none of the fragments contain input actions, or $(b)$ 
  at least one of the fragments have no external actions. Suppose $\frage{G_1}{\alpha}{}$ 
  and $\frage{G_2}{\alpha}{}$  occur at two distinct automata and the execution
  fragment $\frage{G_1}{\alpha}{}\circ \frage{G_2}{\alpha}{}$ occurs in $\alpha$.
  Then there exists an execution $\alpha'$ of $\mathcal{A}$, where  the execution fragment  $\frage{G_2}{\alpha}{}\circ \frage{G_1}{\alpha}{}$ appears in $\alpha'$, such that 
  $(i)$  $\frage{G_1}{\alpha}{} \sim \frage{G_1}{\alpha'}{}$ and  $\frage{G_2}{\alpha}{} \sim \frage{G_2}{\alpha'}{}$
  $(ii)$ 
the prefix  in $\alpha$  before  $\frage{G_1}{\alpha}{} \circ \frage{G_2}{\alpha}{}$
 is identical to the prefix  in $\alpha'$ before  $\frage{G_1}{\alpha'}{} \circ \frage{G_2}{\alpha'}{}$; and $(ii)$ 
the suffix in $\alpha$  after   
 $\frage{G_1}{\alpha}{} \circ \frage{G_2}{\alpha}{}$
is identical to the suffix in $\alpha'$ after the execution fragment 
 $\frage{G_2}{\alpha'}{} \circ \frage{G_1}{\alpha'}{}$. 
  \end{lemma}

  \begin{proof}
  This is clear because the adversary can move the actions in $G_2$ to occur before $G_1$ at their respective automata, and  because 
   either $(a)$ none of the fragments have any input action or $(b)$ at least one of them has no external actions, and hence 
  the actions in one of these fragments cannot affect the actions in the other fragment.
  \end{proof}

The following lemma states that if there are two fair executions of $\mathcal{A}$ with \rot{} $R$ in each of them, and suppose at any server the non-blocking fragments of $R$ are identical (in terms of the sequence of states and actions), then $R$ returns the similar values in both executions.
\begin{lemma}[Indistinguishability]  \label{lem:exec3_equiv}
Let $\alpha$ and $\beta$ be executions of $\mathcal{A}$ and let $R$ be any \rot{}. Then 
$(i)$  if $\frage{F_{x}}{\alpha}{} \stackrel{s_x}{\sim}\frage{F_{x}}{\beta}{}$ then both 
$R(\alpha)$ and $R(\beta)$ respond with the same value $x$ at $s_x$; and 
$(ii)$ if  $\frage{F_{y}}{\alpha}{} \stackrel{s_y}{\sim}\frage{F_{y}}{\beta}{}$ then both 
$R(\alpha)$ and $R(\beta)$ respond with the same value $y$ at $s_y$.
\end{lemma}

\begin{proof}
Suppose $R$ is invoked at some reader $r$. Let $j \in \{1, 2\}$ and  suppose the fragments  $\frage{F_{j}}{\alpha}{}$ and  $\frage{F_{j}}{\beta}{}$ 
appears in $\alpha$ and $\beta$ respectively, where in $\frage{F_{j}}{\alpha}{}$ server $s_j$
sends  $v_j \in V_j$ to $r$. Then $R(\alpha)$ must return $v_j$ for object $o_j$ by the O property  of $\mathcal{A}$. Then since $\frage{F_{j}}{\alpha}{} \stackrel{s_j}{\sim} \frage{F_{j}}{\beta}{}$ then
in  $\frage{F_{j}}{\beta}{}$  the server $s_j$ 
must also  send $v_j$ to $r$,  therefore,  both $R(\alpha)$ and $R(\beta)$ must return value $v_j$ for $o_j$.
%
\end{proof}

 The following lemma shows that for any finite execution of $\mathcal{A}$ that ends with the invocation of  \rot{} $R_1$, it is always possible to have an extended execution of $\mathcal{A}$ where the fragments $I$, $\frage{F_{x}}{}{}$, $\frage{F_{y}}{}{}$ and $E$ appear consecutively due to the asynchronous network.
\begin{lemma}\label{lem:read_format}
If any finite execution of  $\mathcal{A}$ ends with $\INV{R}$, for a \rot{} $R_1$ then there exists an extension $\alpha$ which is a fair execution of $\mathcal{A}$ and  is  of the form  
   $\frage{P}{\alpha}{} \circ\frage{I}{\alpha}{} \circ \frage{F_{1,x}}{\alpha}{x} \circ \frage{F_{1,y}}{\alpha}{y} \circ \frage{E}{\alpha}{x, y} \circ \frage{S}{\alpha}{}$, where $\frage{P}{\alpha}{}$ is the prefix and $\frage{S}{\alpha}{}$ denotes the rest of the execution. 
 \end{lemma}

 \begin{proof}
 Consider a finite execution of $\mathcal{A}$  that end with  $INV(R)$, which  occurs at some reader  $r$, 
then the  adversary  induces the execution fragment $\frage{I}{\alpha}{}$  by delaying all actions, except the internal and output actions at $r$, between 
the actions $\INV{R}$ and the later of the actions $send(m_1^{r})_{r, s_x}$ and $send(m_2^{r})_{r, s_y}$. Next, the adversary  
delivers $m_1^{r}$ at $s_x$ (via the action $recv(m_1^{r})_{r, s_x}$) and delays all actions, other than internal and output actions at $s_x$, until $s_x$ responds with $x$, via  $send(x)_{s_x, r}$; 
we identify this  execution fragment
 as $\frage{F_{1}}{\alpha}{x}$. Subsequently, in a similar manner, 
the  adversary  delivers the message  $m_2^{r}$ and delays appropriate actions to induce the execution fragment 
$\frage{F_{2}}{\alpha}{y}$. Finally, the adversary  delivers the values $x$ and $y$ to $r$ (via the events $recv(x)_{s_x, r}$ and $recv(y)_{s_y, r}$), 
and delays all actions at other automata until $R$ completes  with action $\RESP{R}$ by returning $(x, y)$.  As a result, we  arrive at a fair execution of $\mathcal{A}$ of the form
 $\frage{I}{\alpha}{} \circ \frage{F_{1}}{\alpha}{x} \circ \frage{F_{2}}{\alpha}{y} \circ \frage{E}{\alpha}{x, y} \circ \frage{S}{\alpha}{}$.
\end{proof}

\remove{
 \begin{figure}[t]
      \centering
        \includegraphics[scale=0.4]{figures/executions_3_1-crop.pdf} 
         \caption{
        \small{Executions of $\mathcal{A}$ with three clients and operations $W$, $R_1$ and $R_2$ leading to the contradiction of S in $\alpha_{10}$.
        Arrows show the transposition of execution fragments from the previous execution.
         }} 
         \label{fig:executions3}
 \end{figure}
}

\begin{figure}[t]
	\hspace*{-0.0cm}

		\centering
	     \includegraphics[width=5.0in]{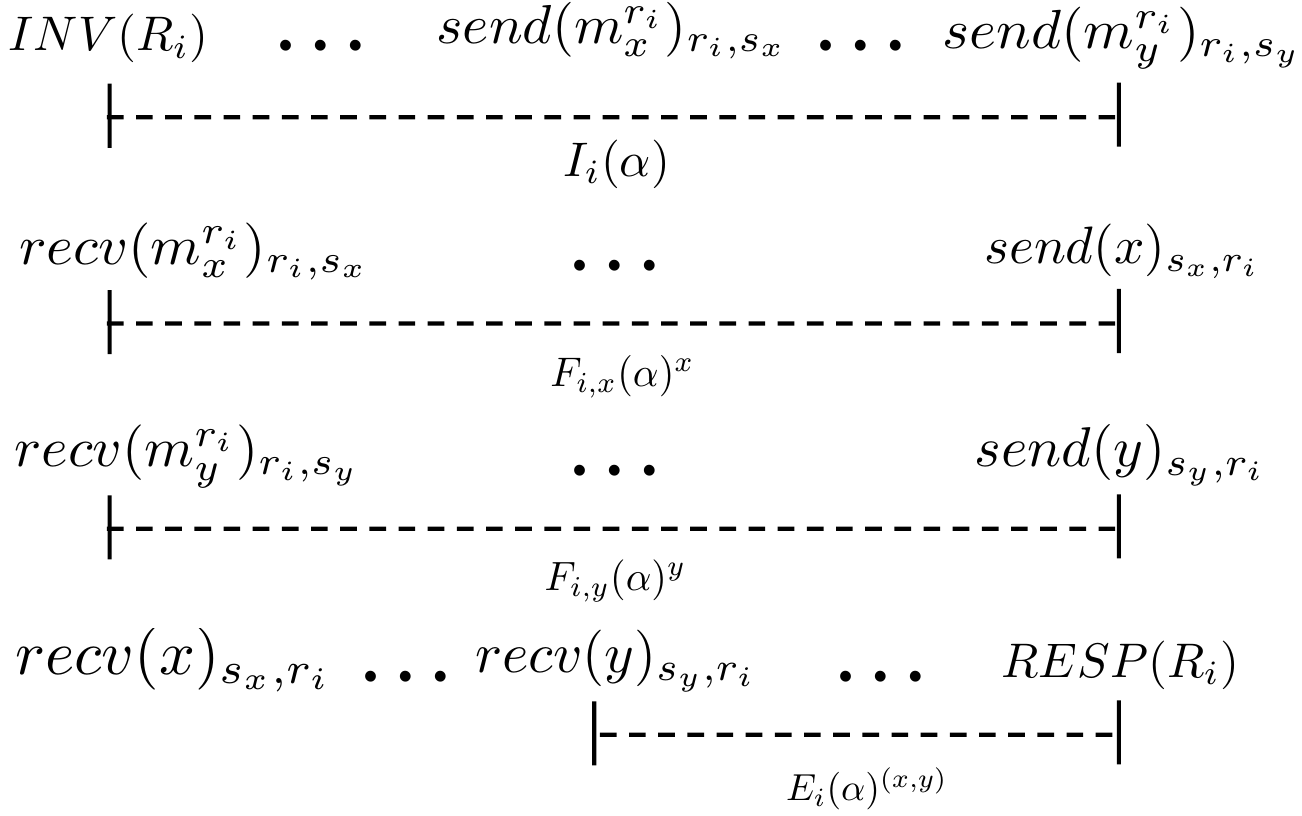}
	     \caption{\small{
	      The relevant actions in the execution fragments $I_i(\alpha)$, $F_{i,x}(\alpha)^{x}$,    $F_{i,y}(\alpha)^{y}$ and   $E_{i}(\alpha)^{(x,y)}$ for any \rot{} $R_i$, $i \in \{1, 2\}$
			of a fair execution $\alpha$ of $\mathcal{A}$.
	     }}
	     \label{fig:exe3_fragments}
\end{figure}

\begin{figure}[t]
		\centering
		\includegraphics[width=5.0in]{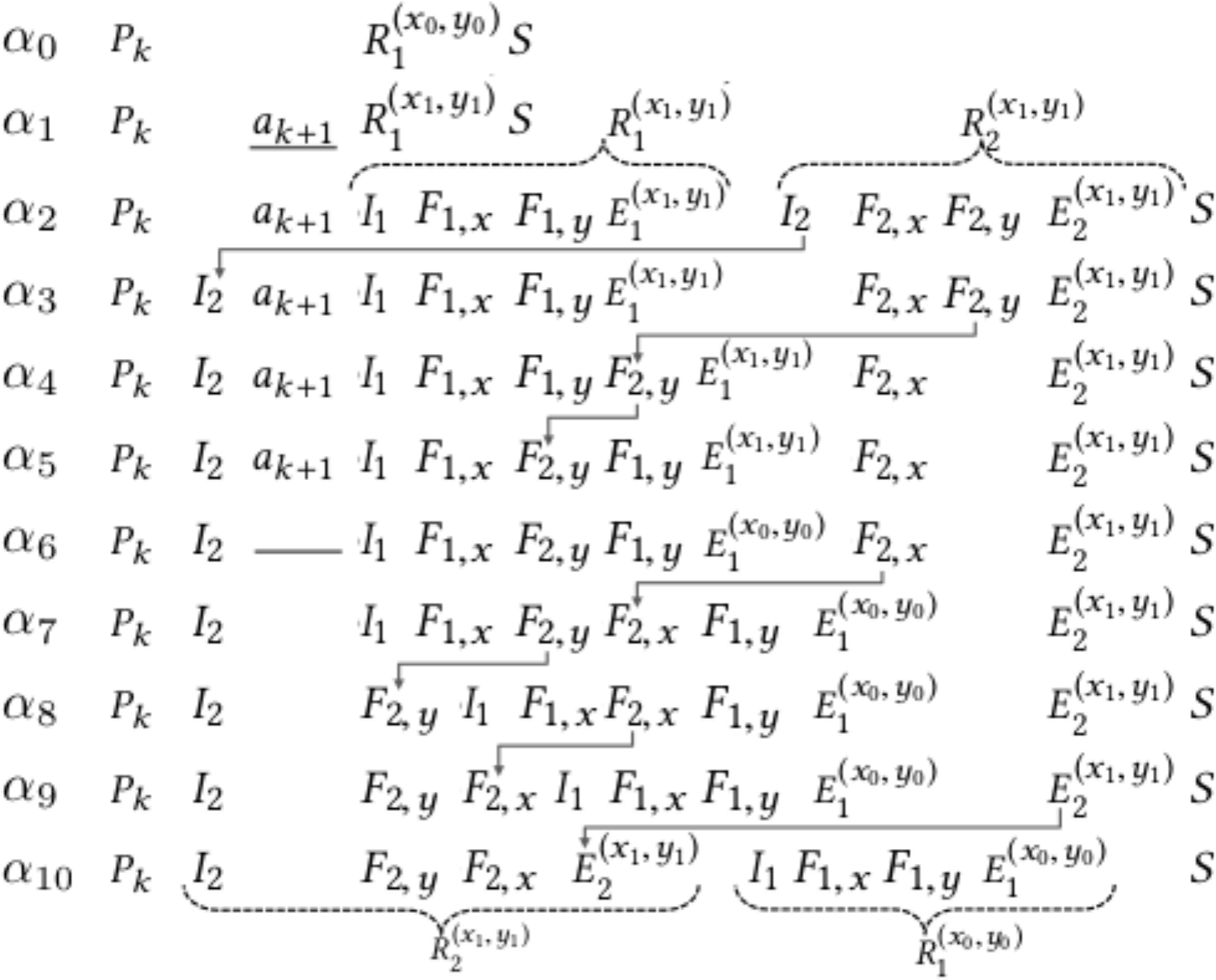}		
		\caption{\small{Executions of $\mathcal{A}$ with three clients and operations $W$, $R_1$ and $R_2$ leading to the contradiction of S in $\alpha_{10}$.
			Arrows show the transposition of execution fragments from the previous execution.
		}}
           	\label{fig:executions3}	
\end{figure}

\section{No SNOW with Three Clients and C2C}
\label{sec:formal_proof}
\sloppy This section provides the sketch of a formal proof of the SNOW Theorem with 3 clients, i.e., SNOW is impossible in a system with 3 or more clients even when client-to-client communication is allowed. The main result of this section is captured by the following theorem. 
\begin{theorem}\label{thm:snow3}
	The SNOW properties cannot be implemented in a system with two readers and one writer, for two servers even in the presence of client-to-client communication.
\end{theorem}

Our proof strategy is to assume the existence of an algorithm $\mathcal{A}$ that satisfies all SNOW properties and create an execution $\alpha$ of $\mathcal{A}$ that contradicts the S property. 
We begin with an execution of $\mathcal{A}$ that contains \rots{} $R_1$ and $R_2$, which both read $s_x$ and $s_y$, and \wot{} $W$ that writes $(x_1, y_1)$ to $s_x$ and $s_y$ respectively (both servers have initial values $x_0$, $y_0$). $R_1$ begins after $W$ completes, and $R_2$ begins after $R_1$ completes. By the S property both  $R_1$   and $R_2$ should return $(x_1, y_1)$.
Then we create a sequence of 
executions of $\mathcal{A}$  (Fig.~\ref{fig:executions3}), where we interchange the fragments until we  finally reach
 an execution in which $R_2$ completes before $R_1$ begins, but $R_2$ returns 
$(x_1, y_1)$ and $R_1$ returns $(x_0, y_0)$ which contradicts the  S property.

The following lemma shows that in an execution of $\mathcal{A}$ with a \wot{} $W$ and a \rot{} $R_1$, there exists a point  in the execution
such that if $R_1$ is invoked before that point then $R_1$ returns $(x_0, y_0)$ and if $R_1$ invoked after 
that point then $R_1$ returns $(x_1, y_1$).

\begin{lemma}[Existence of $\alpha_0$ and $\alpha_1$] \label{lem:exec3_alpha1} 
There exist  executions $\alpha_0$ and $\alpha_1$ of $\mathcal{A}$ that contain transactions $W$ and  $R_1$ 
that satisfy the following properties where $k$ is some positive integer and $\finiteprefixt{k-1}{k}$ is a prefix of $\finiteprefixt{k}{k+1}$:
$(i)$ $\alpha_0$ can be written as 
$\finiteprefixt{k-1}{k} \circ 
\frage{R_1}{\alpha_0}{x_0, y_0}
 \circ \frage{S}{\alpha_0}{}$ ;
$(ii)$   $\alpha_1$ can be written as 
$\finiteprefixt{k}{k+1} \circ 
\frage{R_1}{\alpha_1}{x_1, y_1}  \circ \frage{S}{\alpha_1}{}$; and  
$(iii)$ $a_{k+1}$ in $\alpha_1$  occurs at  $r_1$. 

\end{lemma}

\begin{proof}
Now we describe the construction of a sequence  of finite executions of $\mathcal{A}$, $\{\gamma_k\}_{k=0}^{\infty}$  such that each $\gamma_k$ contains $W$ and $R_1$.
Consider an execution $\alpha$ of $\mathcal{A}$ that contains $W$. 
Suppose $R_1$ is invoked at $r_1$ after the execution fragment 
$\finiteprefixt{k}{k+1}$,  a
prefix of $\alpha$. Allowed by network asynchrony, let $\INV{R}$ be followed by only internal and external actions at $r_1$ until both
 $\send{m_x^{r_1}}{ r_1, s_x}$ and $\send{m_y^{r_1}}{ r_1, s_y}$ occur, thereby creating an execution fragment of the form
$\finiteprefixt{k}{k+1}\circ I_1(\alpha)$. We denote $\finiteprefixt{k}{k+1}$ by $\finiteprefixA{k}{k+1}$.

Next, consider the network delivers the message $m_x^{r_1}$ at $s_x$, and delays 
all actions at other automata and also any  input action at $s_x$ until $s_x$ sends $x$ to $r_1$. Therefore, we achieve the execution fragment   
$\finiteprefixA{k}{k+1}\circ \frage{ I_{1,x}}{\alpha}{} \circ \frage{F_{1,x}}{\alpha}{}$  of  ${\mathcal A}$. 
Next, the network delivers $m_y^{r_1}$ at $s_y$ and delays all actions at other automata and input actions at $s_y$ until $s_y$ sends $y$ to $r_1$. 
Then the network delivers $x$ and $y$ at $r_1$ but it delays  actions at other automata and any other input action at $r_1$ until $\RESP{R_1}$ occurs. Now we have an execution fragment of $\mathcal{A}$, which can be written as 
 $\finiteprefixA{k}{k+1}\circ \frage{I_1}{\alpha}{} \circ \frage{F_{1,x}}{\alpha}{x} \circ \frage{F_{1,y}}{\alpha}{y} \circ \frage{E_1}{\alpha}{x,y}$, where $R_1$ responds with $(x, y)$ such that $(x, y) \in \{ (x_0, y_0), (x_1, y_1) \}$.  We denote this finite execution prefix as $\gamma_k$. Therefore, there exists a sequence of such finite executions  
 $\{\gamma_k\}_{k=0}^{\infty}$.

Because $R_1$ precedes  $W$, by the S property $R_1$ must respond with $(x_0, y_0)$ in $\gamma_0$. If $k$ is large enough such that $a_{k}$ occurs in $\alpha$ after the completion of $W$ then by  the S property,  $R_1$ must return  $(x_1, y_1)$ in 
$\gamma_{k+1}$.  Therefore, there exists a minimum $k$  where in $\gamma_k$
\rot{} $R_1$ returns  $(x_0, y_0)$ and 
 in $\gamma_{k+1}$,
$R_1$ returns  $(x_1, y_1)$. We denote this minimum by $k^*$.
Note that $\gamma_{k^*}$ corresponds to $\alpha_0$ and $\gamma_{k^*+1}$ corresponds to $\alpha_1$ in $(i)$ and $(ii)$ respectively. 

Now, we prove case $(iii)$ by eliminating the possibility of  $a_{k^*+1}$ occurring at $s_x$, $s_y$, $w$ or $r_2$.
%
The S property requires that $R_1$ must retrieve the same version from both $s_x$ and $s_y$, which implies that $s_x$ and  $s_y$ must send values of the same version.
%
%
Observe that $R_1$ returns the $0^{\textit{th}}$ version in $\alpha_0$ and the $1^{\textit{st}}$ version in $\alpha_1$,
 while the prefixes $\fragt{P_{k^*}}{\alpha_0}{}$ and 
 $\fragt{P_{k^*+1}}{\alpha_1}{}$ differ by a single action $a_{k^*+1}$. Importantly, just one action at any of  $s_x$, $s_y$, $r_2$
or $w$ is not enough for $s_x$ and $s_y$ to coordinate the same version to send. Therefore, $a_{k^*+1}$ must occur at $r_1$, which can 
possibly help coordinate by sending some information via $m_x$ and $m_y$ sent to $s_x$ and $s_y$ respectively.

\emph{ \underline{Case $a_{k^*+1}$ occurs at $s_x$}:} Consider the  prefix of execution $\alpha_0$ up to $a_{k^*}$.  Suppose the network invokes 
$R_1$ immediately after action $a_{k^*}$ via $\INV{R_1}$. By Lemma~\ref{lem:read_format} there exists  
an execution $\alpha'$ that contains an  execution fragment of the form
$\fragt{P_{k^*}}{\alpha'}{}
  \circ\frage{I_1}{\alpha'}{} \circ  \frage{F_{1,x}}{\alpha'}{x} 
  \circ\frage{F_{1,y}}{\alpha'}{y}
  \circ \frage{E}{\alpha'}{x, y}$.
Then, $\frage{I_1}{\alpha_1}{} \stackrel{r_1}{\sim} \frage{I_1}{\alpha'}{}$ and
  $\frage{F_{1,y}}{\alpha_1}{} \stackrel{s_y}{\sim} \frage{F_{1,y}}{\alpha'}{}$
  because in both executions the actions of $\frage{I_1}{}{}$ occur entirely at $r_1$
  and those of $\frage{F_{1,y}}{}{}$ occur entirely at $s_y$,
  and thus they are unaffected by the addition of the single action $a_{k^*+1}$ at $s_x$.  
As a result,  $\frage{F_{1,y}}{\alpha'}{}$ must send the same value $y_1$ to $r_1$ as 
 in  $\frage{F_{1,y}}{\alpha_1}{}$.
%
Then in $\alpha'$, $R_1(\alpha')$ returns $y_1$ by Lemma~\ref{lem:exec3_equiv}, and thus $R_1(\alpha')$ returns $(x_1, y_1)$ by the 
S property. However, this contradicts the 
fact that in $\gamma_{k^*}$  $R_1$ responds with $(x_0, y_0)$.

\emph{ \underline{Case $a_{k^*+1}$ occurs at $s_y$}:} A  contradiction can be shown by following a line of reasoning similar to the  preceding case.

\emph{\underline{Case $a_{k^*+1}$ occurs at $w$}:}
This can be argued in a similar manner as the previous case with the trivial fact that $\frage{F_{1,x}}{\alpha_1}{} \stackrel{s_x}{\sim} \frage{F_{1,x}}{\alpha'}{}$ and $\frage{F_{1,y}}{\alpha_1}{} \stackrel{s_y}{\sim} \frage{F_{1,y}}{\alpha'}{}$.

\emph{\underline{Case $a_{k^*+1}$ occurs at $r_2$}:} 
A contradiction can be derived using a line of  reasoning as in the previous case.

So we conclude that $a_{k^*+1}$ must occur at $r_1$ in $\alpha_1$.
\end{proof}

In the remainder of the section, we suppress the explicit reference to the execution. For instance, we use $\fragt{I_i}{\alpha}{}$, $\fragt{F_{i,x}}{\alpha}{x}$, 
$\fragt{F_{i,y}}{\alpha}{y}$, $\fragt{E_i}{\alpha}{x, y}$ and $ \fragt{S}{\alpha}{}$, where we drop $\alpha$, 
instead of $\frage{I_i}{\alpha}{}$, $\frage{F_{i,x}}{\alpha}{x}$, 
$\frage{F_{i,y}}{\alpha}{y}$, $\frage{E_i}{\alpha}{x, y} $ and $ \frage{S}{\alpha}{}$.
If a 
\rot{} $R_i$ has an  execution fragment of the form 
$\fragt{I_i}{\alpha}{} \circ \fragt{F_{i,x}}{\alpha}{x} \circ 
\fragt{F_{i,y}}{\alpha}{y} \circ \fragt{E_i}{\alpha}{x, y} $
we denote it  as  $\fragt{R_i}{}{x, y}$. 
In the rest of the section, $\alpha_0$, $\alpha_1$, and the value of $k$ are the same as in the discussion above. We denote  the execution fragments  $\finiteprefixt{k}{k}$ and $\finiteprefixt{k}{k+1}$ as 
$\finiteprefixA{k}{k}$ and $\finiteprefixA{k}{k+1}$ respectively.
%
Our proof proceeds by stating a sequence set of lemmas 
(Fig.~\ref{fig:executions3}). The first lemma states there exists an execution in which two consecutive \rots{} follow a \wot{}, and both \rots{}
return the new values by the \wot{}.

\begin{lemma}[Existence of $\alpha_2$] \label{lem:exec3_alpha2} 
\sloppy There exists an execution $\alpha_2$  of $\mathcal{A}$ that contains 
$W$, $R_1$, and $R_2$, and can be written in the form 
$\finiteprefixA{k}{k+1} \circ 
\frage{R_1}{}{x_1, y_1}
 \circ
 \frage{R_2}{}{x_1, y_1}
\circ \fragt{S}{\alpha_2}{}$, where both $R_1$ and $R_2$ return $(x_1, y_1)$.
\end{lemma}


\begin{proof}
We can construct an execution  $\alpha_2$ of $\mathcal{A}$ as follows. Consider the prefix 
$\finiteprefixt{k}{k+1} \circ 
\frage{R_1}{\alpha_1}{x_1, y_1}$
of the execution $\alpha_1$, from Lemma~\ref{lem:exec3_alpha1}. At the end of this prefix, the network invokes $R_2$. Now, by 
Lemma~\ref{lem:read_format}, due to $\INV{R_2}$
there is an extension of the prefix of the form 
$\finiteprefixt{k}{k+1} \circ  \frage{R_1}{\alpha_1}{x_1, y_1} \circ 
\frage{I}{\alpha}{} \circ \frage{F_{1}}{\alpha}{x} \circ \frage{F_{2}}{\alpha}{y} \circ \frage{E}{\alpha}{x, y}$. 
By the S property, we have $x = x_1$ and $y= y_1$. Therefore, $\alpha_2$ (Fig.~\ref{fig:executions3}) can be written in the form
$\finiteprefixA{k}{k+1} \circ 
\frage{R_1}{}{x_1, y_1}
 \circ
 \frage{R_2}{}{x_1, y_1}
\circ \fragt{S}{\alpha_2}{}$, where $\fragt{S}{\alpha_2}{}$ is the rest of the execution.
\end{proof}

Based on the previous execution, the following lemma proves that there is an execution of 
$\mathcal{A}$ where $I_2$ occurs earlier than the action $a_{k+1}$ and the invocation of  $R_1$.
\begin{lemma}[Existence of $\alpha_3$] \label{lem:exec3_alpha3} 
\sloppy There exists  execution $\alpha_3$  of $\mathcal{A}$ that contains transactions $W$, $R_1$ and $R_2$, and can be written in the form  
$\finiteprefixA{k-1}{k} \circ  \fragt{I_2}{\alpha_3}{} \circ  a_{k+1} \circ 
\frage{R_1}{}{x_1, y_1}
\circ \fragn{F_{2,x}}{\alpha_3}{x_1} \circ \fragn{F_{2,y}}{\alpha_3}{y_1} \circ \fragn{E_2}{\alpha_3}{x_1, y_1}
\circ \fragt{S}{\alpha_3}{}$,  where both $R_1$ and $R_2$ return $(x_1, y_1)$.
\end{lemma}


\begin{proof}
Consider the execution $\alpha_2$ as in  Lemma~\ref{lem:exec3_alpha2}. 
In the execution fragment $ \fragt{I_1}{\alpha_2}{} \circ  \fragt{F_{1,x}}{\alpha_2}{x_1} \circ \fragt{F_{1,y}}{\alpha_2}{y_1} 
\circ \fragt{E_1}{\alpha_2}{x_1, y_1}$  in $\alpha_2$,  none of the actions occur at $r_2$ and  by Lemma~\ref{lem:exec3_alpha1}, $a_{k+1}$ occurs at $r_1$, 
also the actions in $\fragt{I_2}{}{}$ occur only at $r_2$.
Starting with $\alpha_2$, and by repeatedly using Lemma~\ref{lem:exec3_commute}, we create a sequence of four executions of $\mathcal{A}$ by repeatedly swapping 
 $\fragt{I_2}{\alpha_2}{}$ with the execution fragments 
$\fragt{E_1}{\alpha_2}{x_1, y_1}$, $\fragt{F_{1,y}}{\alpha_2}{y_1}$, $ \fragt{F_{1,x}}{\alpha_2}{x_1}$
and $\fragt{I_1}{\alpha_2}{}$, which appears in   $ \fragt{I_1}{\alpha_2}{} \circ  \fragt{F_{1,x}}{\alpha_2}{x_1} \circ \fragt{F_{1,y}}{\alpha_2}{y_1} 
\circ \fragt{E_1}{\alpha_2}{x_1, y_1}
\circ \fragt{I_2}{\alpha_2}{}$, 
where the following sequence of execution fragments
 $ \fragt{I_1}{\alpha_2}{} \circ  \fragt{F_{1,x}}{\alpha_2}{x_1} \circ \fragt{F_{1,y}}{\alpha_2}{y_1}  \circ \fragt{I_2}{\alpha_2}{} \circ \fragt{E_1}{\alpha_2}{x_1, y_1}$ (by commuting  $\fragt{I_2}{\alpha_2}{}$  and $\fragt{E_1}{\alpha_2}{x_1, y_1}$); ~
 $ \fragt{I_1}{\alpha_2}{} \circ  \fragt{F_{1,x}}{\alpha_2}{x_1} \circ \fragt{I_2}{\alpha_2}{}  \circ \fragt{F_{1,y}}{\alpha_2}{y_1} \circ \fragt{E_1}{\alpha_2}{x_1, y_1}$ (by commuting $\fragt{I_2}{\alpha_2}{}$  and $\fragt{F_{1,y}}{\alpha_2}{y_1}$); ~
 $ \fragt{I_1}{\alpha_2}{} \circ  \fragt{I_2}{\alpha_2}{}  \circ \fragt{F_{1,x}}{\alpha_2}{x_1} \circ \fragt{F_{1,y}}{\alpha_2}{y_1} \circ \fragt{E_1}{\alpha_2}{x_1, y_1}$ (by commuting   $\fragt{I_2}{\alpha_2}{}$  and $\fragt{F_{1,x}}{\alpha_2}{x_1}$)
appear. Finally, we have an  execution $\alpha'$ of the form 
$ \finiteprefixA{k}{k+1} \circ  \fragt{I_2}{\alpha''}{} \circ 
\frage{R_1}{}{x_1, y_1}
\circ \fragt{F_{2,x}}{\alpha'}{x_1} \circ \fragt{F_{2,y}}{\alpha'}{y_1} \circ \fragt{E_2}{\alpha_3'}{x_1, y_1}
\circ \fragt{S}{\alpha'}{}$ (by commuting $\fragt{I_2}{\alpha_2}{}$ and $\fragt{I_1}{\alpha_2}{}$)
Next, from $\alpha'$, by using Lemma~\ref{lem:exec3_commute}  and swapping $a_{k+1}$ with $\fragt{I_2}{\alpha_2}{}$ 
we have shown the existence  of  an execution 
 $\alpha_3$.
\end{proof}

In the following lemma, we show that we can create an execution $\alpha_4$ of $\mathcal{A}$, where 
$\fragt{F_{2,y}}{\alpha_4}{}$ occurs immediately before  $ \fragt{E_1}{\alpha_4}{x_1, y_1}$, while $\fragt{R_1}{\alpha_4}{}$ and 
$\fragt{R_2}{\alpha_4}{}$ both return 
$(x_1, y_1)$.
\begin{lemma}[Existence of $\alpha_4$] \label{lem:exec3_alpha4} 
\sloppy There exists  execution $\alpha_4$  of $\mathcal{A}$ that contains transactions $W$, $R_1$ and $R_2$ and  can be written in the form
$ \finiteprefixA{k-1}{k} \circ \fragt{I_2}{\alpha_4}{}  \circ a_{k+1} $$ \circ \fragt{I_1}{\alpha_4}{}
\circ \fragn{F_{1,x}}{\alpha_4}{x_1} \circ   \fragn{F_{1,y}}{\alpha_4}{y_1} 
\circ \fragn{F_{2,y}}{\alpha_4}{y_1} 
\circ \fragn{E_1}{\alpha_4}{x_1, y_1}
$$\circ \fragn{F_{2,x}}{\alpha_4}{x_1} \circ  \fragn{E_2}{\alpha_4}{x_1, y_1}\circ \fragt{S}{\alpha_4}{}$, where both $R_1$ and $R_2$ return $(x_1, y_1)$.
\end{lemma}


\begin{proof}
We start with an execution  $\alpha_3$, as in  Lemma~\ref{lem:exec3_alpha3}, and apply Lemma~\ref{lem:exec3_commute} twice.  

First, by  Lemma~\ref{lem:exec3_commute}, we know there exists an execution $\alpha'$ of $\mathcal{A}$
where  $\fragt{F_{2,x}}{\alpha_3}{}$ (identify as $G_1$) and 
$\fragt{F_{2,y}}{\alpha_3}{}$ (identify as $G_2$)  are interchanged since actions of 
$\fragt{F_{2,x}}{\alpha_3}{}$ occurs solely at $s_x$ and those of $\fragt{F_{2,y}}{\alpha_3}{}$ at $s_y$, 
and  $\fragt{F_{2,x}}{\alpha_3}{}$ and $\fragt{F_{2,y}}{\alpha_3}{}$ return $x_1$ and $y_1$, respectively, to $r_2$.

Next, by Lemma~\ref{lem:exec3_commute} there is 
execution  of $\mathcal{A}$,  say $\alpha_4$ where
the fragments 
$\fragt{E_1}{\alpha_3}{}$ (identify as $G_1$)  and 
$\fragt{F_{2,y}}{\alpha_3}{}$ (identify as $G_2$) are interchanged, with respect to $\alpha'$, because  the actions in  $\fragt{E_1}{\alpha'}{}$ occur at $r_1$ and 
those of $\fragt{F_{2,y}}{\alpha_3}{}$ at $s_y$. Furthermore, $\alpha_4$
  can be written in the form
 $  \finiteprefixA{k-1}{k}  \circ  \fragt{I_2}{\alpha''}{} \circ  a_{k+1}  \circ $$
  \fragt{I_1}{\alpha''}{} \circ \fragt{F_{1,x}}{\alpha''}{x_1}
   \circ \fragt{F_{1,y}}{\alpha''}{y_1} 
  \circ \fragt{F_{2,y}}{\alpha''}{y_1} 
\circ \fragt{E_1}{\alpha''}{x_1, y_1}
\circ \fragt{F_{2,x}}{\alpha''}{x_1} \circ \fragt{E_2}{\alpha''}{x_1, y_1}
\circ \fragt{S}{\alpha''}{}$.
\end{proof}

 Next, we create an execution  $\alpha_5$ where 
    $\fragt{F_{2,y}}{\alpha_5}{}$ occurs before  $ \fragt{F_{1,y}}{\alpha_5}{}$. 

\begin{lemma}[Existence of $\alpha_5$] \label{lem:exec3_alpha4b} 
\sloppy There exists  execution $\alpha_5$  of $\mathcal{A}$ that contains transactions $W$, $R_1$ and $R_2$ and  can be written in the form
$ \finiteprefixA{k-1}{k} \circ \fragt{I_2}{\alpha_5}{}  \circ a_{k+1} $$ \circ \fragt{I_1}{\alpha_5}{}
\circ \fragn{F_{1,x}}{\alpha_5}{x_1} \circ \fragn{F_{2,y}}{\alpha_5}{y_1} \circ  \fragn{F_{1,y}}{\alpha_5}{y_1} \circ \fragn{E_1}{\alpha_5}{x_1, y_1}
$$\circ \fragn{F_{2,x}}{\alpha_5}{x_1} \circ  \fragn{E_2}{\alpha_5}{x_1, y_1}\circ \fragt{S}{\alpha_5}{}$, 
where both $R_1$ and $R_2$ return $(x_1, y_1)$.
\end{lemma}

\begin{proof}
Given all actions  in $ \fragt{F_{1,y}}{\alpha_4}{} $ and 
  $\fragt{F_{2,y}}{\alpha_4}{}$ occur at $s_y$ in $\alpha_4$, consider the prefix of $\alpha_4$ that ends with $\fragn{F_{1,x}}{\alpha''}{x_1}$. We extend this prefix as follows.
In this prefix,  the actions $send(m_y^{r_2})_{r_2, s_y}$ and 
 $send(m_y^{r_1})_{r_1, s_y}$ do not have their corresponding $recv$ actions.  Suppose  the network 
delivers $m_y^{r_2}$ at $s_y$ (via the action $recv(m_y^{r_2})_{r_2, s_y}$) and delays all actions, other than internal and output actions at $s_y$, until $s_y$ responds with $y$, via  action $send(y)_{s_y, r_2}$. This extended execution fragment is of the form $\fragt{F_{2,y}}{\alpha_4}{}$. Similarly, the network further extends the execution by placing  the 
action $recv(m_y^{r_1})_{r_1, s_y}$ at $s_y$  and creates the execution fragment of the form $\fragt{F_{1,y}}{\alpha_4}{}$. Note that, so far, the actions due to the above extensions are entirely at $s_y$. Suppose the network makes the execution 
fragments $\fragt{E_1}{\alpha''}{}$ happen next by delivering values sent during $\fragt{F_{1,x}}{\alpha_4}{}$ and $\fragt{F_{1,y}}{\alpha_4}{}$ via the actions  $recv(x)_{s_x, r_1}$ and $recv(y)_{s_y, r_1}$ respectively at $r_1$. Then $\fragt{F_{2,x}}{\alpha''}{}$ occurs next, such that this fragment contains exactly the same sequence of actions as  in the corresponding execution fragment in  $\alpha_4$.
This is possible because they are not influenced by  any output action in $\fragt{F_{2,y}}{\alpha_4}{}$  or $\fragt{F_{1,y}}{\alpha_4}{}$. 
Suppose the network places the execution fragment  $\fragt{E_2}{\alpha''}{}$ next.
Let us denote the execution that is an extension of this finite execution so far as $\alpha_5$, which is of the form $ \finiteprefixA{k-1}{k} \circ \fragt{I_2}{\alpha_5}{}  \circ a_{k+1} $$ \circ \fragt{I_1}{\alpha_5}{}
\circ \fragt{F_{1,x}}{\alpha_5}{} \circ \fragt{F_{2,y}}{\alpha_5}{} \circ  \fragt{F_{1,y}}{\alpha_5}{} \circ \fragt{E_1}{\alpha_5}{}
$$\circ \fragt{F_{2,x}}{\alpha_5}{} \circ  \fragt{E_2}{\alpha_5}{}\circ \fragt{S}{\alpha_5}{}$.
 Now we need to argue about the values returned by the reads.
 
Note that both $\alpha_4$ and $\alpha_5$ have the same execution fragment  $\frage{F_{1,x}}{\alpha_4}{}$. Therefore,   $\frage{F_{1,x}}{\alpha_4}{} \stackrel{s_x}{\sim} \frage{F_{1,x}}{\alpha_5}{}$, and thus $s_x$ also returns $x_1$ in $\fragt{F_{1,x}}{\alpha_4}{}$ in $\alpha_5$. Next  by Lemma~\ref{lem:exec3_equiv} for $R_1$,  $s_y$ returns $y_1$ in $\fragt{F_{1,y}}{\alpha_5}{}$  and  hence by the S property, $R_1(\alpha_5)$ returns $(x_1, y_1)$, i.e.,   that $r_1$ returns the new version of object values. Therefore, $\frage{F_{1,x}}{\alpha_4}{}$,  $\fragt{F_{1,y}}{\alpha_5}{}$ and $\fragt{E_1}{\alpha_5}{}$ are of the form $ \fragt{F_{1,x}}{\alpha_5}{x_1}$, $\fragt{F_{1,y}}{\alpha_5}{y_1}$ and  $\fragt{E_1}{\alpha_5}{x_1, y_1}$, respectively.

Note that by construction of $\alpha_5$ above,  the execution fragment  $\fragt{F_{2,x}}{\alpha_4}{}$ in both $\alpha_4$ and $\alpha_5$ is the same, therefore,   $\frage{F_{2,x}}{\alpha_4}{} \stackrel{s_x}{\sim} \frage{F_{2,x}}{\alpha_5}{}$. Hence as in 
$\alpha_4$,  $s_x$  returns $x_1$  in the execution fragment  $\frage{F_{2, 1}}{\alpha_5}{}$  in   $\alpha_5$,  i.e.,  of the form $\frage{F_{2, 1}}{\alpha_5}{x_1}$.
Since $s_x$ returns $x_1$ in $\fragt{F_{2, 1}}{\alpha_5}{}$ in $\alpha_5$,  by   
Lemma~\ref{lem:exec3_equiv}  and   the S property,
  $R_2$ returns $(x_1, y_1)$  and hence $\fragt{E_2}{\alpha_5}{}$ is of the form $\fragt{E_2}{\alpha_5}{x_1, y_1}$.

From the above argument we know that  $\alpha_5$ is of the form $ \finiteprefixA{k-1}{k} \circ \fragt{I_2}{\alpha_5}{}  \circ a_{k+1} $$ \circ \fragt{I_1}{\alpha_5}{}
\circ \fragt{F_{1,x}}{\alpha_5}{x_1} \circ \fragt{F_{2,y}}{\alpha_5}{y_1} \circ  \fragt{F_{1,y}}{\alpha_5}{y_1} \circ \fragt{E_1}{\alpha_5}{x_1, y_1}
$$\circ \fragt{F_{2,x}}{\alpha_5}{x_1} \circ  \fragt{E_2}{\alpha_5}{x_1, y_1}\circ \fragt{S}{\alpha_5}{}$.
\end{proof}

In the next lemma, we show the existence of an execution of $\mathcal{A}$ where $R_1$ returns $(x_0, y_0)$  and $I_2$ occurs immediately after $a_{k}$ and 
 $R_2$ responds with $(x_1, y_1)$. 
\begin{lemma}[Existence of $\alpha_6$] \label{lem:exec3_alpha5} 
\sloppy There exists  execution $\alpha_6$  of $\mathcal{A}$ that contains transactions $W$, $R_1$ and $R_2$ and can be
written in the form 
$ \finiteprefixA{k-1}{k} \circ\fragt{I_2}{\alpha_6}{} \circ  \fragt{I_1}{\alpha_6}{} \circ\fragn{F_{1,x}}{\alpha_6}{x_0}
\circ\fragn{F_{2,y}}{\alpha_6}{y_1}\circ \fragn{F_{1,y}}{\alpha_6}{y_0} \circ\fragn{E_1}{\alpha_6}{x_0, y_0}
\circ\fragn{F_{2,x}}{\alpha_6}{x_1} \circ \fragn{E_2}{\alpha_6}{x_1, y_1}\circ \fragt{S}{\alpha_6}{}$, 
where  $R_1$ returns $(x_0, y_0)$  and $R_2$ returns $(x_1, y_1)$.
\end{lemma}

\begin{proof}
The crucial part of this proof is to carefully use the result of  Lemma~\ref{lem:exec3_alpha1} so that $R_1$ returns $(x_0, y_0)$, instead of $(x_1, y_1)$.
Note that the same prefix $\finiteprefixA{k-1}{k}$ appears in $\alpha_5$ of Lemma~\ref{lem:exec3_alpha4b} as well as in $\alpha_0$ and $\alpha_1$ of Lemma~\ref{lem:exec3_alpha1}, where
 $k$ is defined as in Lemma~\ref{lem:exec3_alpha1}.

By Lemma~\ref{lem:exec3_alpha1}, action $a_{k+1}$ occurs at $r_1$.
 In the execution fragment 
 $a_{k+1} $$ \circ \fragt{I_1}{\alpha_5}{}
\circ \fragt{F_{1,x}}{\alpha_5}{x_1} \circ \fragt{F_{2,y}}{\alpha_5}{y_1}$ of $\alpha_5$, the actions in
 $a_{k+1} $$ \circ \fragt{I_1}{\alpha_4}{}$ occur at $r_1$; actions in 
$\fragt{F_{1,x}}{\alpha_4}{x_1}$ occur at $s_x$; and  actions in $\fragt{F_{2,y}}{\alpha_4}{y_1}$ occur at $s_y$. 
Now consider the  prefix of execution $\alpha_4$ ending with $I_2$ and the network invokes $R_1$ immediately after $I_2$  (instead of 
after $a_{k+1}$) and  extends it by the execution fragment  $ \fragt{I_1}{\epsilon}{} \circ \fragt{F_{1,x}}{\epsilon}{} \circ  \fragt{F_{2,y}}{\epsilon}{}$ to create a new finite execution $\epsilon$, which is of the form $\finiteprefixA{k-1}{k}  \circ \fragt{I_2}{\epsilon}{} \circ \fragt{I_1}{\epsilon}{} \circ \fragt{F_{1,x}}{\epsilon}{} \circ  \fragt{F_{2,y}}{\epsilon}{}$. As a result, $a_{k+1}$ may not be in $\epsilon$ because we introduce changes before $a_{k+1}$ occurs.

Note  that if  in the prefix
 $\finiteprefixA{k-1}{k}  \circ \frage{I_2}{\epsilon}{} \circ \frage{I_1}{\epsilon}{} \circ \frage{F_{1,x}}{\epsilon}{} \circ  \frage{F_{2,y}}{\epsilon}{}$  of $\epsilon$ 
we ignore the actions in $\frage{I_2}{\epsilon}{}$
 then the remaining execution is the same as the prefix  
 $\finiteprefixA{k-1}{k}  \circ  \frage{I_1}{\alpha_0}{} \circ \frage{F_{1,x}}{\alpha_0}{} \circ  \frage{F_{2,y}}{\alpha_0}{}$ of $\alpha_0$ in 
Lemma~\ref{lem:exec3_alpha1}. Here we explicitly use the notations $\epsilon$ and $\alpha_0$ to avoid confusion. Since the actions in 
$\frage{I_2}{\epsilon}{}$ have no impact on the actions in 
 $\frage{I_1}{\epsilon}{} \circ \frage{F_{1,x}}{\epsilon}{} \circ  \frage{F_{2,y}}{\epsilon}{}$,
we have $\frage{F_{1,x}}{\epsilon}{} \stackrel{s_x}{\sim} \frage{F_{1,x}}{\alpha_0}{}$. Therefore, by Lemma~\ref{lem:exec3_consistent} 
  $\frage{F_{1,x}}{\epsilon}{}$ returns $x_0$ as in  $\frage{F_{1,x}}{\alpha_0}{}$, i.e., $s_x$ returns $x_0$ in $\fragn{F_{1,x}}{\epsilon}{x_0}$. 
Now by Lemma~\ref{lem:exec3_equiv} we  conclude that 
for  any  extension of $\epsilon$, say $\gamma$, \rot{} 
 $R_1(\gamma)$ returns $x_0$ at $s_x$ and by the S property  $R_1(\gamma)$ returns  $(x_0, y_0)$. Also, since 
$\frage{F_{2,y}}{\alpha_5}{} \stackrel{s_y}{\sim} \frage{F_{2,y}}{\epsilon}{} \stackrel{s_y}{\sim} \frage{F_{2,y}}{\gamma}{}$  by Lemma~\ref{lem:exec3_equiv} and the S property,  $R_2(\gamma)$ must return $(x_1, y_1)$.  Therefore, $\gamma$ has an extension of $\alpha_6$ (Fig.~\ref{fig:executions3}) which is of  the form $ \finiteprefixA{k-1}{k}  \circ\fragt{I_2}{\alpha_6}{} \circ  \fragt{I_1}{\alpha_6}{} \circ\fragt{F_{1,x}}{\alpha_6}{x_0}
\circ\fragt{F_{2,y}}{\alpha_6}{y_1}\circ \fragt{F_{1,y}}{\alpha_6}{y_0} \circ\fragt{E_1}{\alpha_6}{x_0, y_0}
\circ\fragt{F_{2,x}}{\alpha_6}{x_1} \circ \fragt{E_2}{\alpha_6}{x_1, y_1}\circ \fragt{S}{\alpha_6}{}$ as in the statement of the lemma.
\end{proof}

The  following lemma shows that there exists  an execution $\alpha_7$ for $\mathcal{A}$ where $\fragt{F_{2,x}}{\alpha_7}{}$ 
appears  before $\fragn{F_{1,y}}{\alpha_7}{y_0} \circ\fragn{E_1}{\alpha_7}{x_0, y_0}$, where $R_1$ returns $(x_0, y_0)$ and $R_2$ returns $(x_1, y_1)$.  The lemma can be proven by  starting from $\alpha_6$ in Lemma~\ref{lem:exec3_alpha5} and  moving the execution fragments  of $R_2$ earlier, a little at a time, until finally we have $R_2$ finishing before $R_1$ starts.  This simply uses commutativity since the actions in the swapped execution fragments occur at different automata.


\begin{lemma}[Existence of $\alpha_7$] \label{lem:exec3_alpha6} 
\sloppy There exists  execution $\alpha_7$  of $\mathcal{A}$ that contains transactions $W$, $R_1$ and $R_2$, and   can be written in the form
$\finiteprefixA{k-1}{k} \circ \fragt{I_2}{\alpha_7}{}  \circ \fragt{I_1}{\alpha_7}{} \circ \fragn{F_{1,x}}{\alpha_7}{x_0} 
\circ \fragn{F_{2,y}}{\alpha_7}{y_1}\circ \fragn{F_{2,x}}{\alpha_7}{x_1} 
 \circ \fragn{F_{1,y}}{\alpha_7}{y_0} \circ\fragn{E_1}{\alpha_7}{x_0, y_0}\circ
 \fragn{E_2}{\alpha_7}{x_1, y_1}\circ \fragt{S}{\alpha_7}{}$ where $R_1$ returns $(x_0, y_0)$ and $R_2$ returns $(x_1, y_1)$.
\end{lemma}

\begin{proof}
\sloppy This result is proved by applying the result of Lemma~\ref{lem:exec3_commute}  to the execution created in Lemma~\ref{lem:exec3_alpha5}. 
Suppose,  $\alpha_6$ (Fig.~\ref{fig:executions3}) is a  execution as  in Lemma~\ref{lem:exec3_alpha5}, 
where in  the   execution fragment 
$\fragt{E_1}{\alpha_6}{x_0, y_0}\circ\fragt{F_{2,x}}{\alpha_6}{}$  we identify 
 $\fragt{E_1}{\alpha_6}{x_0, y_0}$  as $G_1$ and 
$\fragt{F_{2,x}}{\alpha_6}{y_0}$  as $G_2$. The actions of $G_1$ and $G_2$ occur at two distinct automata, therefore,  we  can use the result of Lemma~\ref{lem:exec3_commute}, to argue that there exists an execution $\alpha'$ of $\mathcal{A}$  that  contains the  execution fragment 
$\fragt{F_{2,x}}{\alpha_7}{}  \circ\fragt{E_1}{\alpha_7}{x_0, y_0}$, and $\alpha_6$ and $\alpha'$ are identical in the prefixes and suffixes corresponding to $G_1$ and $G_2$.

 Now,  $\alpha'$ contains  $\fragt{F_{1,y}}{\alpha''}{} \circ \fragt{F_{2,x}}{\alpha'}{}$, where  
  the actions in  $\fragt{F_{1,y}}{\alpha''}{}$ (identified as $G_1$) and $\fragt{F_{2,x}}{\alpha''}{}$ (identify as $G_2$)  occur  at distinct automata. Hence, by  Lemma~\ref{lem:exec3_commute} there exists  an execution $\alpha_7$ of the form
$\finiteprefixA{k-1}{k}  \circ \fragt{I_2}{\alpha_7}{}  \circ \fragt{I_1}{\alpha_7}{} \circ \fragt{F_{1,x}}{\alpha_7}{x_0} 
\circ \fragt{F_{2,y}}{\alpha_7}{y_1}\circ \fragt{F_{2,x}}{\alpha_7}{x_1} 
 \circ \fragt{F_{1,y}}{\alpha_7}{y_0} \circ\fragt{E_1}{\alpha_7}{x_0, y_0}\circ
 \fragt{E_2}{\alpha_7}{x_1, y_1}\circ \fragt{S}{\alpha_7}{}$.
\end{proof}

The following lemma leverages Lemma~\ref{lem:exec3_commute} to show the existence of an execution $\alpha_8$ of $\mathcal{A}$ where $\fragt{F_{2,y}}{\alpha_8}{}$ appears before $\fragt{I_1}{\alpha_8}{} \circ \fragt{F_{1,x}}{\alpha_8}{}$, and $R_1$ returns $(x_0, y_0)$ while $R_2$ returns $(x_1, y_1)$.

\begin{lemma} [Existence of $\alpha_8$]  \label{lem:exec3_alpha7} 
\sloppy There exists  execution $\alpha_8$  of $\mathcal{A}$ that contains transactions $W$, $R_1$ and $R_2$
and   can be written in the form 
$\finiteprefixA{k-1}{k}  \circ  \fragt{I_2}{\alpha_8}{} \circ \fragn{F_{2,y}}{\alpha_8}{y_1} \circ 
$$\fragt{I_1}{\alpha_8}{} \circ \fragn{F_{1,x}}{\alpha_8}{x_0} 
\circ  \fragn{F_{2,x}}{\alpha_8}{x_1} 
$$ \circ \fragn{F_{1,y}}{\alpha_8}{y_0} \circ \fragn{E_1}{\alpha_8}{x_0, y_0}
\circ \fragn{E_2}{\alpha_8}{x_1, y_1}\circ \fragt{S}{\alpha_8}{}$, where $R_1$ returns $(x_0, y_0)$ and $R_2$ returns $(x_1, y_1)$.
\end{lemma}

\begin{proof}
 Consider the execution $\alpha_7$ of $\mathcal{A}$ as in Lemma~\ref{lem:exec3_alpha6}. In the context of 
  of  Lemma~\ref{lem:exec3_commute}, in $\alpha_7$ (Fig.~\ref{fig:executions3}) the actions in  
 $ \fragt{F_{1,x}}{\alpha_8}{}$ (identify as $G_1$)  occur at $s_x$  and  those in $\fragt{F_{2,y}}{\alpha_8}{}$ (identify as $G_2$) at  $s_y$. Then by  
 Lemma~\ref{lem:exec3_commute} there exists an execution $\alpha'$ of $\mathcal{A}$,  of the form 
 $ \finiteprefixA{k-1}{k}  \circ \fragt{I_2}{\alpha'}{}  \circ \fragt{I_1}{\alpha''}{} 
 \circ \fragn{F_{2,y}}{\alpha'}{y_1} 
\circ \fragn{F_{1,x}}{\alpha'}{x_0} 
\circ \fragn{F_{2,x}}{\alpha'}{x_1} 
 \circ \fragn{F_{1,y}}{\alpha'}{y_0} \circ\fragn{E_1}{\alpha'}{x_0, y_0}\circ
 \fragn{E_2}{\alpha''}{x_1, y_1}\circ \fragt{S}{\alpha'}{}$, where $ \fragt{F_{2,y}}{\alpha_8}{}$ and 
  $ \fragt{F_{1,x}}{\alpha_8}{}$  are interchanged.
 
 Since  actions in  $ \fragt{F_{2,y}}{\alpha_8}{}$ (identify as $G_1$)  occur at $s_y$ and those in 
  $ \fragt{I_{1}}{\alpha_8}{}$ (identify as $G_1$)  occur at $r_1$ then by Lemma~\ref{lem:exec3_commute}
  there is a  execution of $\mathcal{A}$, $\alpha_8$ where $\fragt{F_{2,y}}{\alpha_8}{}$ appear before $\fragt{I_1}{\alpha_8}{}$, i.e., of the form  $\finiteprefixA{k-1}{k}  \circ \fragt{I_2}{\alpha_8}{} \circ \fragt{F_{2,y}}{\alpha_8}{}  \circ
 \fragt{I_1}{\alpha_7}{} \circ \fragt{F_{1,x}}{\alpha_8}{} 
\circ \fragt{F_{2,x}}{\alpha_8}{} 
 \circ \fragt{F_{1,y}}{\alpha_8}{} \circ\fragt{E_1}{\alpha_8}{} \circ
 \fragt{E_2}{\alpha_8}{}\circ \fragt{S}{\alpha_8}{}$, where $ \fragt{F_{2,y}}{\alpha_8}{}$ and 
  $ \fragt{I_{1}}{\alpha_8}{}$  are interchanged.
 
 By $(ii)$ of Lemma~\ref{lem:exec3_equiv} we have 
$\frage{F_{2,x}}{\alpha'}{} \stackrel{s_x}{\sim} \frage{F_{2,x}}{\alpha_8}{}$ hence $\fragt{F_{2,x}}{\alpha_8}{}$ sends $x_1$ and 
$\fragt{F_{1,x}}{\alpha_8}{}$ and $\fragt{F_{1,y}}{\alpha_8}{}$ sends $x_0$ and $y_0$, respectively. So considering these returned values we have $\alpha_8$ 
(Fig.~\ref{fig:executions3}) in the form as stated in the lemma.
\end{proof}


The following lemma shows the existence of an execution  $\alpha_9$, of $\mathcal{A}$, where 
$\fragt{F_{2,x}}{}{}$ appears before $\fragt{F_{1,x}}{}{}$.
\begin{lemma} [Existence of $\alpha_9$]  \label{lem:exec3_alpha8} 
\sloppy There exists  execution $\alpha_9$  of $\mathcal{A}$ that contains transactions $W$, $R_1$ and $R_2$
and   can be written in the form 
$\finiteprefixA{k-1}{k}  \circ  \fragt{I_2}{\alpha_8}{} \circ \fragt{F_{2,y}}{\alpha_8}{y_1} \circ 
$$\fragt{I_1}{\alpha_8}{} \circ \fragt{F_{2,x}}{\alpha_8}{x_1} 
\circ  \fragt{F_{1,x}}{\alpha_8}{x_0} 
$$ \circ \fragt{F_{1,y}}{\alpha_8}{y_0} \circ \fragt{E_1}{\alpha_8}{x_0, y_0}
\circ \fragt{E_2}{\alpha_8}{x_1, y_1}\circ \fragt{S}{\alpha_8}{}$
where $R_1$ returns $(x_0, y_0)$ and $R_2$ returns $(x_1, y_1)$.
\end{lemma}

\begin{proof}
In $\alpha_8$ from Lemma~\ref{lem:exec3_alpha7}, all the actions in  
 $ \fragn{I_1}{\alpha_8}{}$  occur at $r_1$; those in   $ \fragn{F_{1,x}}{\alpha_8}{x_1}$  occur at $s_x$; and  the  
actions in $ \fragn{F_{2,x}}{\alpha_8}{y_1}$ occur only  at $s_x$. 
Note that actions of  both execution fragments $ \fragt{F_{2,x}}{\alpha_8}{}$  and   
$\fragt{F_{1,x}}{\alpha_8}{}$ occur at $r_1$.  
Consider the prefix of $\alpha_8$ that ends with $\fragt{I_1}{}{}$ then suppose the network  extends this prefix by adding an execution fragment of the form
$ \fragt{F_{2,x}}{\alpha_8}{} \circ  \fragt{F_{1,x}}{\alpha_8}{}$ as follows.  
First note that  the actions $send(m_x^{r_2})_{r_2, s_x}$ and 
 $send(m_x^{r_1})_{r_1, s_x}$ appears in the prefix but do not have corresponding $recv$ actions.
The  network places action $recv(m_x^{r_2})_{r_2, s_x}$, and allows an execution fragment of the form 
$ \fragt{F_{2,x}}{\alpha_8}{}$ to appear. Now, immediately after this the network further extends it with 
an execution fragment of the form  $\fragt{F_{1,x}}{\alpha_8}{}$  by placing action $recv(m_x^{r_1})_{r_1, s_x}$. 
 Next the fragment $\fragn{F_{1,y}}{\alpha_8}{y_0}$ is added  and is the same as $ \frage{F_{1,y}}{\alpha_8}{}$.
This last step can be argued by the fact that none of the actions in  $ \fragn{F_{1,y}}{\alpha_8}{y_0}$ can be affected by 
any of the output actions at $ \fragt{F_{2,x}}{\alpha_8}{}$ and  $\fragt{F_{1,x}}{\alpha_8}{}$. 
Following this the network allows the rest of the execution by adding an execution fragment of the form $\fragn{E_1}{\alpha_8}{x_0, y_0}
\circ \fragn{E_2}{\alpha_8}{x_1, y_1}\circ \fragn{S}{\alpha_8}{}$. The resulting execution is of the form 
$\finiteprefixA{k-1}{k}  \circ  \fragt{I_2}{\alpha_8}{} \circ \fragt{F_{2,y}}{\alpha_8}{y_1} \circ 
$$\fragt{I_1}{\alpha_8}{} \circ \fragn{F_{2,x}}{\alpha_8}{x_1} 
\circ  \fragn{F_{1,x}}{\alpha_8}{x_0} 
$$ \circ \fragt{F_{1,y}}{\alpha_8}{y_0} \circ \fragn{E_1}{\alpha_8}{x_0, y_0}
\circ \fragn{E_2}{\alpha_8}{x_1, y_1}\circ \fragt{S}{\alpha_8}{}$, where we retained the values wherever it is known, and we denote this  execution by $\alpha_9$.

Now, we argue about the return values in $\alpha_9$. 
Applying Lemma~\ref{lem:exec3_equiv} to $R_2$ and $\fragn{F_{2,y}}{\alpha_8}{y_1}$ implies that $R_2$ 
returns $(x_1, y_1 )$. Similarly, 
applying Lemma~\ref{lem:exec3_equiv} to $R_1$ and $\fragn{F_{1,y}}{\alpha_8}{y_0}$ implies that $R_1$ 
must return $(x_0, y_0 )$ in $\alpha_9$.
\end{proof}


Now by constructing a sequence of executions, $\alpha_3$ through $\alpha_{10}$ (Fig.~\ref{fig:executions3}, lemmas and proofs omitted due to lack of space) realize the existence of an execution $\alpha_{10}$  of $\mathcal{A}$ where the execution fragments corresponding to $R_2$ appears 
before $R_1$, where  $R_1$ returns $(x_0, y_0)$ and $R_2$ completes by returning $(x_1, y_1)$
but $R_1$ is in real time after $R_2$, therefore, violates the S property.
\begin{lemma} [Existence of $\alpha_{10}$]  \label{lem:exec3_alpha11} 
\sloppy There exists an execution $\alpha_{10}$  of $\mathcal{A}$ that contains transactions $W$, $R_1$ and $R_2$
and   can be written in the form 
$\finiteprefixA{k-1}{k}   \circ
\frage{R_2}{}{x_1, y_1} \circ
 \frage{R_1}{}{x_0, y_0}
 \circ \fragt{S}{\alpha_{10}}{}$.
where $R_1$ returns $(x_0, y_0)$ and $R_2$ returns $(x_1, y_1)$.
\end{lemma}
\begin{proof}
Now, by applying Lemma~\ref{lem:exec3_commute} to $\alpha_{9}$, we can swap 
 $\fragt{F_{2,x}}{}{}$  and $ \fragt{I_1}{}{}$ to create an execution $\alpha_{10}$  (Fig.~\ref{fig:executions3}) of $\mathcal{A}$, which is
of the form 
$ \finiteprefixA{k-1}{k}   \circ 
\fragt{I_2}{\alpha_9}{} \circ \fragt{F_{2,y}}{\alpha_9}{y_1} 
\circ   \fragt{F_{2,x}}{\alpha_9}{x_1}  \circ 
\frage{R_1}{}{x_0, y_0} \circ \fragt{E_2}{\alpha_8}{x_1, y_1}
 \circ \fragt{S}{\alpha_9}{}$, where the returned values are determined by  Lemma~\ref{lem:exec3_consistent}.
 
Note that none of the actions in 
$\fragt{I_1}{\alpha_9}{} \circ \fragt{F_{1,x}}{\alpha_9}{x_0} 
 \circ \fragt{F_{1,y}}{\alpha_9}{y_0} \circ \fragt{E_1}{\alpha_9}{x_0, y_0}$
occur at $r_2$ and all actions in $\fragt{E_2}{\alpha_9}{x_1, y_1} $ occur at $r_2$. Therefore, by 
applying Lemma~\ref{lem:exec3_commute}, we can consecutively swap $\fragt{E_2}{}{}$ with 
$\fragt{E_1}{}{}$, $\fragt{F_{1,y}}{}{}$,  $\fragt{I_1}{}{}$,  and $\fragt{F_{1,x}}{}{}$. Therefore, we create a sequence of 
four executions of $\mathcal{A}$ to arrive at execution $\alpha_{10}$ (Fig.~\ref{fig:executions3}) of the form
$\finiteprefixA{k-1}{k}   \circ
\frage{R_2}{}{x_1, y_1}
 \circ
 \frage{R_1}{}{x_0, y_0}
 \circ \fragt{S}{\alpha_{10}}{}$.
\end{proof}

Using the above results,  we prove Theorem ~\ref{thm:snow3} for 3-clients by showing the existence of $\alpha_{10}$,  where $R_2$ completes before $R_1$ is invoked and  $R_2$ returns $(x_1, y_1)$ whereas $R_1$ returns $(x_0, y_0)$, which  violates the S property. 

\sloppy 

\setcounter{algorithm}{3}

\section{Two Client Open Question}
\label{sec:2c2s}
This section closes the open question of whether SNOW properties can be implemented with two clients. We first prove that SNOW remains impossible in a 2-client 2-server system if the clients cannot directly send messages to each other.
 However, in the presence of client to client communication, it is possible to have all SNOW properties with two clients and at least two servers.
 
 \subsection{No  SNOW Without C2C  Messages}
 \label{subsec:no_snow_no_c2c}

 In this section, we prove the following results that states it is impossible guarantee the SNOW properties in a transaction processing system with two clients, without client-to-client communication.
  \begin{theorem}\label{thm:two-snow}
  	The SNOW properties cannot be implemented in a system with two clients and two servers, where the clients do not communicate with each other.
  \end{theorem}

 We use system model with two servers $s_x$ and $s_y$ with two clients, a reader $r_1$ that issues only 
\rots{} and a writer $w$ that issues only \wots{}. A \wot{} $W$ writes $(x_1, y_1)$ to $s_x$ and $s_y$, and a \rot{} $R$ reads both servers. We assume that there is a bi-directional communication channel between any pair of client and server and any pair of servers. There is no communication channel between clients. We assume that each transaction can be identified by a unique number, e.g., transaction identifier.

Our strategy is still proof by contradiction: We assume there exists some algorithm $\mathcal{A}$ that satisfies all SNOW properties, and then we show the existence of a sequence of executions of $\mathcal{A}$, eventually leading to an execution that contradicts  the S property.
 First,  we show the existence of an execution $\alpha$ of $\mathcal{A}$ where 
$R_1$ is invoked after $W$ completes, where the  send actions $send(m_x^{r_1})_{r_1, s_x}$ and $send(m_y^{r_1})_{r_1, s_y}$  at the  $r_1$  occur consecutively 
in  $P(\alpha)$, which is a prefix of $\alpha$. 
 Then we show that $\alpha$  can be written in the form  
$\prefix{\alpha} \circ \frag{1,x}{\alpha}$ (Fig.~\ref{fig:execution1} $(a)$, Lemma~\ref{lem:exec_alpha}). 
 We then prove the existence of  another  execution $\beta$, which can be written in the form $\prefix{\beta} \circ \frag{1,x}{\beta} \circ \frag{1,y}{\beta}$ by extending $\alpha$ with an execution fragment $\frag{1,y}{\beta}$, such that $\frag{1,x}{\beta} \stackrel{s_x}{\sim} \frag{1,x}{\alpha}$
  (Fig.~\ref{fig:execution1} $(b)$; Lemma~\ref{lem:exec_beta}).
Note that in  any arbitrary extension of $\beta$, $R_1$ eventually returns $(x_1, y_1)$.
Next, we show the existence of an execution $\gamma$ of the form  $\prefix{\gamma} \circ \frag{1,x}{\gamma} \circ \frag{1,y}{\gamma}$,  
 where the send actions $send(m_x^{r_1})_{r_1, s_x}$ and $send(m_y^{r_1})_{r_1, s_y}$  at  $r_1$ occur before $W$ is invoked (Fig.~\ref{fig:execution1} $(c)$; $\gamma$), but 
$\frag{1,x}{\gamma}$  and $\frag{1,y}{\gamma}$
occur after $\RESP{W}$ as in $\beta$.
Based on $\gamma$, we show the existence of an execution $\delta$ 
of the form  $\prefix{\eta} \circ \frag{1,x}{\eta} \circ \frag{1,y}{\eta}\circ \suffix{\eta}$, where 
  $R_1$ responds with $(x_1, y_1)$.
%
   %
    Finally, starting with  $\eta$, we create a sequence of executions $\delta (\equiv \eta)$, $\delta^{(1)}$, $\cdots$ $\delta^{(f)}$, of $\mathcal{A}$, 
     where in each of them $R_1$ responds with $(x_1, y_1)$  (Fig.~\ref{fig:execution1} $(e)$ and $(g)$; Lemma~\ref{thm:two-snow}). Additionally,  for any $\delta^{(i)}$, the fragments   $\frag{1,x}{\delta^{(i)}}$ and
    $\frag{1,y}{\delta^{(i)}}$  appear before $\delta^{(i-1)}$.
    Based on $\delta^{(f)}$, we prove the existence of an execution $\phi$, where $R_1$ returns 
    $(x_1, y_1)$ even before $W$ begins, which violates the $S$ property.

  Now, we explain the relevant lemmas used in the main proof. 
  The first lemma states that there exists a finite execution of ${\mathcal A}$, where $R_1$ begins after  
  $W$ completes, and the actions $send(m_x^{r_1})_{r_1, s_x}$ and $send(m_y^{r_1})_{r_1, s_y}$  occur 
before either of the servers $s_x$ and  $s_y$ 
  receives $m_x^{r_1}$ or $m_y^{r_1}$ from $r_1$;  also,  server $s_x$ responds to $r_1$ in a non-blocking manner (execution $\alpha$ in Fig.~\ref{fig:execution1}).  
\begin{lemma}\label{lem:exec_alpha} 
There exists a finite execution $\alpha$ of $\mathcal{A}$ that contains transactions $R_1(\alpha)$  
 and $W(\alpha)$ where $\INV{R}$ appears after $\RESP{W}$, and the following conditions hold:
\begin{enumerate}
\item[$(i)$] The actions $send(m_x^{r_1})_{r_1, s_x}$ and  $send(m_y^{r_1})_{r_1, s_y}$ appear consecutively in $trace(\alpha)|r_1$; and 
\item[$(ii)$]   $\alpha$ contains the execution fragment  $\frag{1,x}{\alpha}$.
\end{enumerate} 
\end{lemma}

\begin{proof}
Consider a finite  execution fragment  of $\mathcal{A}$ with a completed transaction $W$, where  after $W$ 
completes  the adversary invokes $R$, i.e., $\INV{R_1}$ occurs.
 Note that each of the  read operations $op_1^{r_1}$ and $op_2^{r_1}$,  in $R_1$,  can be invoked by the adversary at any point in the execution.  
 Following $\INV{R}$,  the adversary introduces the invocation action $\inv{op_1^r}$;  by the O property of the read operations of $\mathcal{A}$  
the action $send(m_x^{r_1})_{r_1, s_x}$ eventually occurs. Next, the adversary introduces 
 $\inv{op_2^{r_1}}$ and also, delays the arrival of  $m_x^{r_1}$ until  action $send(m_y^{r_1})_{r_1, s_y}$ eventually occurs, 
 which  must occur in accordance with  the property O  of read operations. Let us call this finite execution $\alpha^0$. 
 
 Next, suppose at the end of $\alpha^0$  the adversary delivers the message $m_x^{r_1}$, which has been delayed so far, 
 via the action  $recv(m_x^{r_1})_{r, s_x}$, 
 at  $s_x$,   but  it delays any  other input actions at $s_x$. Note that by the $N$ property of read operations $s_x$ eventually responds with  $send(x)_{s_x, r_1}$, with one value $x$ by o property,  where $x = x_1$ by the S property, since $R$ begins after $W$ completes. 
 Let us call this execution $\alpha$. Note that $\alpha$ satisfies conditions $(i)$ and $(ii)$ by the design of the execution.
\end{proof}

 %

   The following lemma states that there is an execution $\beta$ where $R$ begins after  $W$ completes, and the two send events at $r_1$ occurs before  $\frag{1,x}{\beta}$, which thus occurs before $\frag{1,y}{\beta}$ ($\beta$ in Fig.~\ref{fig:execution1}).

\begin{lemma}\label{lem:exec_beta} 
There exists an  execution $\beta$ of $\mathcal{A}$ that contains transactions $R_1$ and $W$ where $\INV{R_1}$ appears after $\RESP{W}$ and the following conditions hold:
\begin{enumerate}
\item[$(i)$] The actions $send(m_x^{r_1})_{r_1, s_x}$ and  $send(m_y^{r_1})_{r_1, s_y}$ appear consecutively in $trace(\beta)|r_1$; and
\item[$(ii)$]   $\beta$ contains the execution fragment  $\frag{1,x}{\beta} \circ \frag{1,y}{\beta}$.
\end{enumerate} 
\end{lemma}

\begin{proof}
Consider the  execution $\alpha$ of $\mathcal{A}$ as constructed in Lemma~\ref{lem:exec_alpha}.  
At the end of the execution fragment $\alpha$,  the adversary  delivers the previously delayed message
 $m_y^r$, which is sent via the action $send(m_y^r)_{r, s_y}$,  by  introducing the action $recv(m_y^r)_{r, s_y}$. The adversary then delays any other  input action in $\mathcal{A}$. By the N property,  
server $s_y$ must respond to $r$, with some value $y$,  and hence 
the output action $send(y)_{s_y, r}$ must eventually occur at $s_y$.  
Let us call this finite execution as $\beta$. Note that $\beta$ satisfies the  properties $(i)$ and $(ii)$ in the statement of the lemma.
\end{proof}

 The following result shows  that starting with $\beta$  there is an execution 
 $\gamma$, where  $R_1$ is invoked before $W$ is invoked while
  $send(m_x^{r_1})_{r_1, s_x}$ and  $send(m_y^{r_1})_{r_1, s_y}$  occur before $\INV{W}$ (Fig.~\ref{fig:execution1} $(c)$) and $m_x^{r_1}$ and $m_y^{r_1}$ from $r_1$ reach $s_x$ and $s_y$ after the action  $\RESP{W}$.

\begin{lemma}\label{lem:exec_gamma} 
There exists an execution $\gamma$ that contains $R_1$ and $W$, where the action $\INV{R_1}$ appears before $\INV{W}$ and $\RESP{R_1}$ appears after  $\RESP{W}$, and the following conditions hold for $\gamma$:
\begin{enumerate}
\item[$(i)$] The actions $send(m_x^{r_1})_{r_1, s_x}$ and  $send(m_y^{r_1})_{r_1, s_y}$  appear before $\INV{W}$ and they appear consecutively in $trace(\gamma)|r_1$;
\item[$(ii)$]   $\gamma$ contains the execution fragment  $\frag{1,x}{\gamma} \circ \frag{1,y}{\gamma}$; and 
\item[$(iii)$] action $\RESP{W}$ occurs before $\frag{1,x}{\gamma}$.
\end{enumerate} 
\end{lemma}

\begin{proof}
Consider the execution $\beta$ of $\mathcal{A}$  as in Lemma~\ref{lem:exec_beta}. Note that $\beta$
is an execution of the composed automaton $\mathcal{A}$ ($\equiv S_1 \times r$).  In $\beta$,  the actions 
 $send(m_x^r)_{r, s_x}$ and   $send(m_y^r)_{r, s_y}$ occur at $r$; and following that,  the actions 
 $recv(m_x^r)_{r, s_x}$, $recv(m_y^r)_{r, s_y}$, $send(x)_{s_x, r}$ and  $send(y)_{s_y, r}$ occur at $S_1$. 
 Consider the executions $\alpha_r \equiv \beta|r$ and  $\alpha_{S_1} \equiv \beta|S_1$. Let $s_{\beta}$ denote  $trace(\beta)$.
 
 In $s_{\beta}$, $send(m_x^r)_{r, s_x}$,  $send(m_y^r)_{r, s_y}$ appear after $\RESP{W}$, as in $trace(\beta)$. Let $s'_{\beta}$ be the sequence of
 external actions of $S_2$ which we construct from $s_{\beta}$ by moving $send(m_x^r)_{r, s_x}$,  $send(m_y^r)_{r, s_y}$ before $\INV{W}$, which  is also 
 an external action of $\mathcal{A}$, and leaving the rest of the actions in $s_{\beta}$ as it is.

 In $\beta$,  $\INV{R}$, $recv(x)_{ s_x, r}$ and  $recv(y)_{s_y, r}$ 
  are the only input actions at $r$, therefore,   $s'_{\beta} |r= trace(\alpha_r)$. 
  On the other hand,  $recv(m_x^r)_{r, s_x}$, $recv(m_y^r)_{r, s_y}$ are the only input actions at $s_x$ ,
   therefore, $s'_{\beta} |S_1 = trace(\alpha_{S_1})$.  Now, by Theorem ~\ref{thm:paste}, there exists an execution 
 $\gamma$ of $\mathcal{A}$ such that, $s'_{\beta} = trace(\gamma)$ 
 and $\alpha_r = \gamma|r$ and $\alpha_{S_1} = \gamma|S_1$. Therefore, 
 in $\gamma$, $send(m_x^r)_{r, s_x}$,  $send(m_y^r)_{r, s_y}$ appear before
  $\INV{W}$ (condition $(i)$) and since $s'_{\beta} = trace(\gamma)$ condition $(ii)$  holds. Conditions $(iii)$ holds trivially.
\end{proof}

In the following lemma we show  that in any   execution, of ${\mathcal A}$,  that is  an extension of either execution $\beta$ or execution  $\gamma$, as in  the preceding lemmas, $R_1$ eventually returns $(x_1, y_1)$.

\begin{lemma}\label{lem:exec_xi}
Let $\xi$  be an execution of $\mathcal{A}$ that is an   extension of the either  execution $\beta$ from Lemma~\ref{lem:exec_beta} or execution  $\gamma$ from Lemma~\ref{lem:exec_gamma},  then  $R(\xi)$  responds with $(x_1, y_1)$.
\end{lemma}

\begin{proof}
Note that in executions $\beta$ and $\gamma$,  the traces $trace(\beta)|s_x$  is a prefix of $trace(\gamma)|s_x$, since  $\beta$ ends with
$\frag{1,x}{\beta}$ and $\gamma$ ends with $\frag{1,y}{\gamma}$
Therefore, in both $\beta$ and $\gamma$, the  respective $send(x)_{s_x, r}$ actions have the same value for their  $x$'s. 
Now, in  any extended  execution $\eta$ of $\mathcal{A}$, which starts 
with $\beta$ or $\gamma$, by the properties $N$ and $O$ the transaction $R$ completes;  and  by the property $S$,   $R$ returns  $(x_1, y_1)$.
Therefore,  $R_1(\xi)$  returns $(x_1, y_1)$.
\end{proof}

 In the following lemma, we show there exists an    execution $\eta$  of  
$\mathcal{A}$  of the form 
$\prefix{\eta}\circ \frag{1,x}{\eta} \circ \frag{1,y}{\eta} \circ \suffix{\eta}$
where  $\RESP{R}$ appears in  $\suffix{\eta}$  (Fig.~\ref{fig:execution1} $(d)$) and 
$R(\eta)$ returns $(x_1, y_1)$. 

\begin{lemma}\label{lem:exec_delta} 
There exists an execution $\eta$ of $\mathcal{A}$ that contains transactions $R$ and $W$ where $\INV{R}$ appears before $\INV{W}$;  $\RESP{R_1}$ appears after  $\RESP{W}$ and the following conditions hold for $\eta$:
\begin{enumerate}
\item[ $(i)$] $\eta$ can be written in the form $\prefix{\eta}\circ \frag{1,x}{\eta} \circ \frag{1,y}{\eta} \circ \suffix{\eta}$, for some 
   $\prefix{\eta}$ and $\suffix{\eta}$;
\item[$(ii)$] The actions $send(m_x^r)_{r_1, s_x}$ and  $send(m_y^{r_1})_{r_1, s_y}$  
appear before  $\INV{W}$ and they appear consecutively in $trace(\eta)|r_1$;
\item[$(iii)$] action $\RESP{W}$ occurs before $\frag{1,x}{\eta}$; and 
 \item[$(iv)$] $R_1(\eta)$ returns $(x_1, y_1)$. 
\end{enumerate} 
\end{lemma}

\begin{proof}
Let $\gamma$ be an execution of $\mathcal{A}$, as described in Lemma~\ref{lem:exec_gamma}. 
 Let $\gamma^0$ be the execution fragment of $\gamma$
up to the action $send(y)_{s_y, r}$. Now, by Theorem~\ref{thm:extension} (1), there exists an execution $\gamma^0 \circ \mu$, of 
$\mathcal{A}$, where $\mu$ denotes the extended portion of the execution.

 Clearly, by the N and O properties, 
the actions  $\resp{op_1^r}$ and $\resp{op_2^r}$ must eventually occur in $\gamma^0 \circ \mu$.
 Now,  identify $\eta$ as $\gamma^0 \circ \mu$,  where $\prefix{\eta} \circ \frag{1,x}{\eta} \circ \frag{1,y}{\eta}$  is $\gamma^0$,   and $\mu$ is $\suffix{\eta}$, thereby, proving condition $(i)$.

Note the condition $(ii)$ is satisfied by $\eta$ because $\RESP{W}$ appears in $\prefix{\eta}$, therefore, the execution $\gamma$ is equivalent to the execution fragment of $\prefix{\eta}$ up to the event $\INV{W}$, and  also, 
$\gamma$ satisfies condition $(ii)$ as stated in Lemma~\ref{lem:exec_gamma}.

Condition $(iii)$ is true because $\frag{1,x}{\eta}$ begins with action $recv(m_x^r)_{r, s_x}$, which occurs after $\RESP{W}$.
Condition $(iv)$ is satisfied by $\eta$ because $\eta$ is an extension of $\gamma$ and due to  the result of Lemma~\ref{lem:exec_xi}.
\end{proof}

\begin{figure*}[t]
	\hspace*{-0.7cm}
	\begin{subfigure}{0.49\textwidth}
		\centering
	    \includegraphics[width=1.2\linewidth]{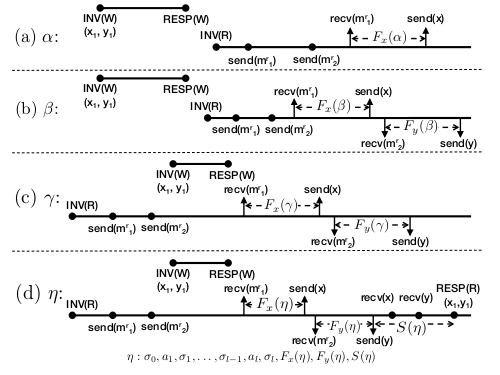}
	\end{subfigure}
     	\hspace*{-0.7cm}
	\begin{subfigure}{0.49\textwidth}
		\centering
	     \includegraphics[width=1.2\linewidth]{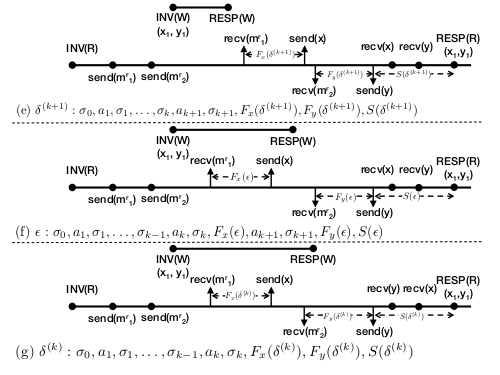}
	\end{subfigure} 
	         \caption{
	         	\small{Schematic representation of executions $\alpha$, $\beta$, $\gamma$ and $\eta$,
	         		of ${\mathcal A}$ with transactions $R$ and $W$. The executions evolve from left to right. 
	         		The dots denote external events at clients. 
	         		The up-arrow marks denote external actions at $s_x$. The down-arrow marks denote external actions at $s_y$.}} \label{fig:execution1}
\end{figure*}

\remove{
\begin{figure}[t]
\centering
\begin{subfigure}{1.0\columnwidth}
\centering
  \includegraphics[width=1.0\linewidth]{figures/fig4-1.pdf}
\end{subfigure}  
\begin{subfigure}{1.0\columnwidth}
\centering
  \includegraphics[width=1.0\linewidth]{figures/fig4-2.pdf}
\end{subfigure}  
\begin{subfigure}{1.0\columnwidth}
\centering
  \includegraphics[width=1.0\linewidth]{figures/fig4-3.pdf}
\end{subfigure}

\caption{\small{try graphs}}
\label{fig:latency}
\end{figure}
}

  The following proof proves Theorem~\ref{thm:two-snow}. In the proof we start with an execution $\eta$ and create a sequence of executions of 
  $\mathcal{A}$, where each one is of the form $\prefix{\cdot} \circ \frag{1,x}{\cdot}\circ \frag{1,y}{\cdot}\circ \suffix{\cdot}$, with progressively shorter $\prefix{\cdot}$ until we have a final execution that contradicts the $S$ property.
\remove{
 \begin{theorem}\label{thm:two-snow}
 The SNOW properties cannot be implemented in a system with two clients and two servers, where the clients do not communicate with each other.
\end{theorem}
}

\begin{proof}
\sloppy Consider an execution $\delta^{(\ell)}$  of $\mathcal{A}$ as in Lemma~\ref{lem:exec_delta}, and let $\prefix{\delta^{(\ell)}}$  be the execution fragment
$\finiteprefix{\ell-1}{\ell}$.
By  Lemma~\ref{lem:exec_delta}, $\RESP{R_1(\delta^{(\ell)})}$  returns
 $(x_1, y_1)$ and $\delta^{(\ell)}$ is also of the form    
$\finiteprefix{\ell-1}{\ell} \circ \frag{1,x}{\delta^{(\ell)}} \circ \frag{1,y}{\delta^{(\ell)}} \circ \suffix{\delta^{(\ell)}}$.     

Next, we inductively prove the existence of a finite sequence of executions of $\mathcal{A}$---i.e., by proving the existence of a new execution based on the existence of a previous one---as $\delta^{(\ell)}, \delta^{(\ell-1)}, \cdots \delta^{(i)}, \delta^{(i-1)},  \cdots \delta^{(f)}$, 
for some positive integer $f$, with the following properties: $(a)$ Each of the execution in 
the sequence can be written in the form  $\finiteprefix{i-1}{i} \circ \frag{1,x}{\delta^{(i)}} \circ \frag{1,y}{\delta^{(i)}} \circ \suffix{\delta^{(i)}}$ (or $\prefix{\cdot}\circ \frag{1,x}{\cdot}\circ \frag{1,y}{\cdot}\circ\suffix{\cdot}$); $(b)$ for each $i$, $f \leq i  < \ell$, we have  $\prefix{\delta^{(i)}}$ to be a prefix of $\prefix{\delta^{(i+1)}}$; and $(c)$ $R_1(\delta^{(f)})$  returns  $(x_0, y_0)$, and for any $i$,  $f <  i \leq \ell$,  we have   $R_1(\delta^{(i)})$  that returns  $(x_1, y_1)$.  Note that there is a final execution of the form $\delta^{(f)}$  because of the initial values of $x_0$ and $y_0$  and the \wot{}  $W$. 



 
Clearly, there exists an integer $k$,  $ f \leq  k < \ell$, such that $\frage{R}{\delta^{(k)}}{}$ returns $(x_0, y_0)$ and $\frage{R_1}{\delta^{(k+1)}}{}$ returns $(x_1, y_1)$. Now we start with execution 
$\delta^{(k+1)}$ and construct an execution $\delta^{(k)}$ as described in the rest of the proof. The following argument will show that  $\frage{R}{\delta^{(k)}}{}$ must also return $(x_1, y_1)$, which contradicts the assumption of having the $S$ property.

Consider the execution $\delta^{(k+1)}$ of the form 
$\finiteprefix{k}{k +1} \circ \frag{1,x}{\delta^{(k+1)}} \circ \frag{1,y}{\delta^{(k+1)}} \circ \suffix{\delta^{(k+1)}}$. 
 The action  $a_{k+1}$ can  occur at any of the automata $r$, $w$, $s_x$, and $s_y$. Therefore, we consider the following four possible cases.

\emph{ \underline{Case (i) $a_{k+1}$ occurs at $w$:}} The execution fragments
 $\frag{1,x}{\delta^{(k+1)}}$ and $\frag{1,y}{\delta^{(k+1)}}$ do not contain any input action  at $s_x$ or $s_y$. $a_{k+1}$ does not occur at $s_x$ or $s_y$. 
 Therefore,
 the asynchronous network can
  delay the occurrence of $a_{k+1}$ at $w$ to create the 
   finite execution   
$\finiteprefix{k-1}{k } \circ \frag{1,x}{\delta^{(k+1)}} \circ \frag{1,y}{\delta^{(k+1)}}$,  
 of. Also, there exists an execution $\delta^{(k)}$ that
 is an extension of the above finite execution, where $R_1$ completes in $\delta^{(k)}$.
  Clearly, $\delta^{(k)}$ can be written as  
 $\finiteprefix{k-1}{k } \circ \frag{1,x}{\delta^{(k)}} \circ \frag{1,y}{\delta^{(k)}} \circ \suffix{\delta^{(k)}}$, where $ \suffix{\delta^{(k)}}$ is the tail of the execution resulting from the extension. Moreover,  
  $\frag{1,x}{\delta^{(k+1)}}$   is indistinguishable  from   
  $\frag{1,x}{\delta^{(k)}}$ at $s_x$, i.e., 
      $\frag{1,x}{\delta^{(k+1)}} \stackrel{s_x}{\sim} \frag{1,x}{\delta^{(k)}}$. 
   Therefore, $send(x)_{s_x, r_1}$ has the same object value 
   $x$ in both fragments, which means $R_1$ returns $x_1$, and thus $R_1(\delta^{(k)})$ must return $(x_1, y_1)$ by the property $S$.

\emph{ \underline{Case (ii) $a_{k+1}$ occurs at $r$:}} Similar to Case $(i)$.

\emph{ \underline{Case (iii) $a_{k+1}$ occurs at $s_x$:}} Observe that the two execution fragments  
   $a_{k+1} \sigma_{k+1} \circ \frag{1,x}{\delta^{(k+1)}}$ and $\frag{1,y}{\delta^{(k+1)}}$ occur at separate  automata, i.e., at $s_x$ and $s_y$ respectively. 
    Also, the  execution fragments  $a_{k+1} \sigma_{k+1} \circ \frag{1,x}{\delta^{(k+1)}}$ and $\frag{1,y}{\delta^{(k+1)}}$ do not contain any  input actions at $s_x$ or $s_y$. 
    Therefore, 
    we can create an execution $\epsilon$, which  can be expressed as 
    $\finiteprefix{k-1}{k }  \circ \frag{1,y}{\epsilon}\circ a_{k+1}, \sigma_{k+1} \circ \frag{1,x}{\epsilon} \circ \suffix{\epsilon}$. Clearly,   
     $\frag{1,x}{\epsilon} \stackrel{s_x}{\sim} \frag{1,x}{\delta^{(k+1)}}$ and  
     $\frag{1,y}{\epsilon} \stackrel{s_y}{\sim} \frag{1,y}{\delta^{(k+1)}}$.  
     %
     %
   %
  Because $send(x)_{s_x, r_1}$ occurs in both $\frag{1,x}{\epsilon}$ and 
     $\frag{1,x}{\delta^{(k+1)}}$, it sends the same value $x_1$ to $r_1$. Therefore, $R$ returns $(x_1, y_1)$. 
   
   Now let us denote the execution fragment  $\finiteprefix{k-1}{k }  \circ \frag{1,y}{\epsilon}$ by $\epsilon'$, which  is simply a finite prefix of $\epsilon$.
     Allowed by the asynchronous network, we append $recv(m_x^{r_1})_{r_1, s_x}$ to $\epsilon'$, and create a finite execution $\epsilon''$ as 
        $\finiteprefix{k-1}{k }\circ \frag{1,y}{\epsilon''}, recv(m_x^{r_1})_{r_1, s_x}$, and delay any input action at $s_y$.    
        
        Let $\epsilon'''$ be an extension of $\epsilon''$. Clearly,  
             $\frag{1,y}{\epsilon'''} \stackrel{s_y}{\sim} \frag{1,y}{\epsilon''}$, and thus $send(y )_{s_y, r_1}$ sends the same value $y_1$ in $\epsilon''$ and $\epsilon'''$. 
             By the $N$ property, $send(x)_{s_x, r_1}$ eventually occurs, and by the $O$ property $x$ is send to $r_1$. Therefore, $R_1$ completes in 
$\epsilon'''$, which implies that 
             $R(\epsilon''')$ must return $(x_1, y_1)$.
             
             Note that the execution fragment of $\epsilon'''$ has no input actions of $s_x$ between $ recv(m_x^{r_1})_{r_1, s_x}$ and $send(x)_{s_x, r_1}$, which can be identified as $\frag{1,x}{\epsilon'''}$. Therefore, $\epsilon'''$ can be written as      $\finiteprefix{k-1}{k }  \circ \frag{1,y}{\epsilon'''} \circ \frag{1,x}{\epsilon'''} \circ \suffix{\epsilon'''}$.
   
    Next, since $\frag{1,x}{\epsilon'''}$ and $\frag{1,y}{\epsilon'''}$ contain actions of different automata, 
    we can create an execution prefix  $\epsilon^{(iv)}$ as 
    $\finiteprefix{k-1}{k }  \circ \frag{1,x}{\epsilon^{(iv)}} \circ \frag{1,y}{\epsilon^{(iv)}}$, 
     where  
     $\frag{1,x}{\epsilon^{(iv)}}$ appears before $\frag{1,y}{\epsilon^{(iv)}}$. Next, 
     we create an execution $\delta^{(k)}$ as an extension of 
       $\epsilon^{(iv)}$, which can be written   as $\finiteprefix{k-1}{k } \circ \frag{1,x}{\delta^{(k)}} \circ \frag{1,y}{\delta^{(k)}}\circ \suffix{\delta^{(k)}}$. 
       Then by the $O$ and $N$ properties, $R_1$ completes in $\delta^{(k)}$.  
       Since $\frag{1,x}{\epsilon^{(iv)}} \stackrel{s_x}{\sim} \frag{1,x}{\epsilon'''}$, $send(x)_{s_x, r_1}$ returns $x_1$ in $\epsilon^{(iv)}$. 
    Similarly, because 
      $\frag{1,x}{\delta^{(k)}} \stackrel{s_x}{\sim} \frag{1,x}{\epsilon^{(iv)}}$, $send(x)_{s_x, r_1}$ returns $x_1$ in $\delta^{(k)}$. So,
     $R_1(\delta^{(k)})$ returns $(x_1, y_1)$.

\emph{ \underline{Case (iv) $a_{k+1}$ occurs at $s_y$:}} Because $a_{k+1}$ occurs at server $s_y$ and $\frag{1,x}{\delta^{(k+1)}}$ occurs at server $s_x$ (different automata),  
we can create a new execution $\epsilon$ of $\mathcal{A}$ (Fig.~\ref{fig:execution1} $(f)$)  as   
             $\finiteprefix{k-1}{k } \circ \frag{1,x}{\epsilon} \circ a_{k+1}, \sigma_{k+1} \circ 
        \frag{1,y}{\epsilon}\circ \suffix{\epsilon}$, such that $\frag{1,x}{\epsilon} \stackrel{s_x}{\sim}\frag{1,x}{\delta^{(k+1)}}$, where  
         $a_{k+1}, \sigma_{k+1} $
       occurs after $\frag{1,x}{\delta^{(k+1)}}$ and 
         $R(\epsilon)$ returns $(x_1, y_1)$.
        
     Now, consider the finite execution   $\finiteprefix{k-1}{k }\circ \frag{1,x}{\epsilon}$ at the end of which we append $recv(m_y^{r_1})_{r_1, s_y}$ to create a finite execution of $\mathcal{A}$ as 
        $\finiteprefix{k-1}{k }\circ \frag{1,x}{\epsilon}, recv(m_y^{r_1})_{r_1, s_y}$. 
        Then, there exists 
        an execution $\epsilon'$ of $\mathcal{A}$, where 
              the network delays the input actions at $s_y$. By the $N$ and $O$ properties, $send(y)_{s_y, r_1}$ occurs in $\epsilon'$. Clearly, since   $\frag{1,x}{\epsilon'} \stackrel{s_x}{\sim} \frag{1,x}{\epsilon}$, $send(x)_{s_x, r_1}$ sends the same value in $\epsilon$ and $\epsilon'$. Therefore,  $R_1$ returns $(x_1, y)$.
              
              In $\epsilon'$, we denote the fragment that begins with $recv(m_y^{r_1})_{r_1, s_y}$ and ends with 
              $send(y)_{s_y, r_1}$ by $\frag{1,y}{\epsilon'}$ to have $\epsilon'$ as  
              $\finiteprefix{k-1}{k } \circ \frag{1,x}{\epsilon'}  \circ 
        \frag{1,y}{\epsilon'}$. 
        Then, there exists 
        an execution $\delta^{(k)}$ of $\mathcal{A}$, which is an extension of $\epsilon'$. Clearly, $\delta^{(k)}$ can be written as 
     $\finiteprefix{k-1}{k }  \circ \frag{1,x}{\delta^{(k)}} \circ \frag{1,y}{\delta^{(k)}} \circ \suffix{\delta^{(k)}}$, where $\suffix{\delta^{(k)}}$ is the tail of the extended execution.  Clearly,  $\frag{1,x}{\delta^{(k)}} \stackrel{s_x}{\sim} \frag{1,x}{\epsilon'}$. Therefore,  $R_1(\delta^{(k)})$ returns $x_1$ in $\delta^{(k)}$, which implies $R_1(\delta^{(k)})$ must return $(x_1, y_1)$.    
\end{proof}

\remove{
 
 \begin{figure}[!ht]
      \centering
         \caption{
\small{The architecture of a typical web service with clients, servers, and 
         the communication channels, between every pair of processes, inside a datacenter is modeled as a collection of I/O automata. Note that, unlike the architecture in Fig.~\ref{fig:architecture2a}, in this setup there are communication channels between every pair of clients.}} 
         \label{fig:architecture2c}
 \end{figure}

}

\subsection{SNOW with C2C Communication}
\label{app:algorithm-a}
In this section,  we show that SNOW is possible in the   \emph{multiple-writers single-reader} 
(MWSR) setting 
when client-to-client communication is allowed. In particular, we present an algorithm $A$, which has all SNOW properties in such setting.
We consider a system that has $\ell \geq 1$ writers with ids $w_1, 
w_2 \cdots w_{\ell} \in \mathcal{W}$ 
, one reader $r$, and  $k \geq 1$ servers with ids $s_1, s_2\cdots s_k \in \mathcal{S}$. 
Client-to-client communication is allowed. 
%
%
The pseudocode for algorithm $A$ is presented in Pseudocode~\ref{fig:algo_a}. 
%
We use keys to uniquely identify a \wot{}.  A key $\kappa \in \mathcal{K}$ is defined as a pair $(z, w)$, 
where $z \in \mathbb{N}$, and $w \in \mathcal{W}$ is the id of a writer. $\mathcal{K}$ denotes the set of all possible keys. 
Also, with each transaction we associate a tag $t \in \mathbb{N}$. 


\textit{\textbf{State variables:}} 
$(i)$ Each  \emph{writer $w$} stores a counter $z$ corresponding to the
number of \wots{}  it  has  invoked so far, initially $0$.
$(ii)$ The  \emph{reader} $r$ has an 
ordered list of elements, $List$, as $(\kappa, (b_1, \cdots, b_k))$,  where 
$\kappa  \in \mathcal{K}$  and 
$(b_1, \cdots b_k) \in  \{0, 1\}^k$. Initially,  
$List= [ ({\kappa}^0, (1, \cdots 1) ]$, where ${\kappa}^0  \equiv (0, w_0)$, 
and $w_0$ is any
place holder identifier for writer id. 
$(iii)$ Each   \emph{server} $s_i \in \mathcal{S}$  stores a set variable $Vals$ 
with elements 
of key-value pairs $({\kappa}, v_i) \in \mathcal{K} \times \mathcal{V}_i$. Initially,
$Vals= \{ ({\kappa}^0, v_i^0)\}$. 

\textit{\textbf{Writer steps:}} Any writer client, $w \in \mathcal{W}$, may invoke a \wot{} $\Writetr{ (o_{i_1}, v_{i_1}), (o_{i_2}, v_{i_2}), \cdots, (o_{i_p}, v_{i_p}) }$, comprising a set of write operations,
where  $I = \{i_1, i_2, \cdots, i_p\}$ is some subset of $p$ indices of $[k]$. We define the set  $S_I\triangleq \{s_{i_1}, s_{i_2}, \cdots, s_{i_p}\}$.      
This procedure consists of two consecutive phases: {\writeValue} and {\informReader}.  In the {\writeValue} phase,  $w$ creates a key ${\kappa}$ as  $ {\kappa}  \equiv (z + 1, w)$; and also increments the local counter $z$ by one.   Then it sends $(${\writeValueTag}$, ({\kappa}, v_{i}))$ to each server $s_i$ in $S_I$, and awaits {\ackTag}s  
from each server  in $S_I$.  After receiving all {\ackTag}s,    $w$ initiates the {\informReader} phase during which  it sends 
(\informReaderTag, $({\kappa}, (b_{1}, \cdots b_{k})$) to $r$, where for any $i \in [k]$, $b_i$ is a boolean variable, such that $b_i=1$ if $s_i \in S_I$, otherwise $b_i=0$. 
Essentially, such a $(k+1)$-tuple
identifies the set of objects that are updated during that \wot{}, i.e., if $b_i=1$ then object 
$o_i$ was updated 
during the execution of the  \wot{}, otherwise $b_i=0$.  
After $w$ receives    {\ackTag} 
from $r$ it completes the \wot{}. 

\textit{\textbf{Reader steps:}}  
We use the same notations for $I$ and $S_I$ as above for the set of indices and corresponding servers, possibly 
different across transactions.
The procedure  \Readtr{$ o_{i_1},  o_{i_2}, \cdots, o_{i_p}$}, 
for any  \rot{}, 
is  initiated at  reader  $r$, where   $o_{i_1},  o_{i_2}, \cdots, o_{i_p}$  denotes the  subset  of  
objects $r$ 
intends to read. This procedure
consists of only one phase,  {\readValue},  of communication 
between the reader and the servers in $S_I$.   Here $r$ sends  the message
(\readValueTag, ${\kappa}_i$) to each server $s_i \in S_I$, where 
the ${\kappa}_i$ is the key in  the tuple $({\kappa}_{i}, (b_{1}, \cdots, b_{k}))$  in  $List$ located at  index $j^*$ such that $b_i =1$ such that 
$i \in I$. 
After
receiving the values $v_{i_1}$, $v_{i_2}, \cdots v_{i_p}$ from all  servers in $\mathcal{S_I}$,  where $S_I \triangleq \{s_{i_1},  s_{i_2}, \cdots, s_{i_p}\}$, the transaction completes by 
returning $(v_{i_1}, \cdots v_{i_p})$.

On receiving a message  
(\informReaderTag, $({\kappa}, (b_{1}, \cdots, b_{k})$) from any writer $w$,  $r$ appends  
$({\kappa}, (b_{1}, \cdots, b_{k})$ to its  $List$,  and responds to $w$ with  
{\ackTag} and $t_w = |List|$, i.e., number of elements in $List$.
The order of the  elements in  $List$ corresponds to  the order  
the \wots{}, the order of the incoming  {\informReaderTag} updates,  as seen by the reader.

\textit{\textbf{Server steps:}} The server responds to messages containing the tags 
{\writeValueTag} and \readValueTag.  The first procedure is used if a server $s_i$ receives a 
message  $(${\writeValueTag}$, ({\kappa}, v_{i}))$  from a writer $w$,  it  adds $({\kappa}, v_i)$ to its set variable   $Vals$ and sends {\ackTag} back to $w$.
The second procedure is used  if  $s_i$ receives a message, i.e., $(${\readValueTag}$, {\kappa}_{i})$, from $r$, then it responds with $v_i$ such that $({\kappa}_{i}, v_i)$ is in its $Vals$.

\begin{algorithm}[!h]
	\begin{algorithmic}[2]
		\begin{multicols}{2}{\footnotesize
				\Statex {\bf At writer $w$}
				\Part{{\it State Variables at $w$}}{ 	
					\Statex $z \in \mathbb N$, initially   $0$
				}\EndPart
				\Statex\Statex
				{\bf  \Writetr{$(o_{i_1}, v_{i_1}), \cdots, (o_{o_p},  v_{i_p})$}}
				\Part{ \underline{\writeValue}} {
					\State ${\kappa} \leftarrow (z +1,  w)$
					\State $z \leftarrow z +1 $
					\State $I\triangleq \{i_1, i_2, \cdots, i_p \}$
					\For{$i \in I$} 
					\State Send (\writeValueTag, $({\kappa}, v_{s_i})$) to  $s_i$
					\EndFor 
					\State  Await {\ackTag}  from  $s_i$ $\forall$ $i \in I$.
				}\EndPart
				\Statex
				\Part{ \underline{\informReader}} {
					\For{$i \in [k]$} 
					\If{$i \in I$}
					\State $b_i \leftarrow 1$
					\Else
					\State $b_i \leftarrow 0$
					\EndIf
					\EndFor 
					\State  Send  (\informReaderTag,
					 \\~~~~~~$({\kappa}, (b_{1}, \cdots, b_{k}))$) to   $r$
					\State  Receive ({\ackTag}, $t_w$) from  $r$
				}\EndPart
		}\end{multicols}	
		\vspace{-1.5em} 
		\\\hrulefill 	
		\vspace{-1.5em}
		\begin{multicols}{2}{\footnotesize	
				\Statex {\bf At reader $r$}
				\Part{{\it State Variables at $r$}}{ 	
					\Statex $List$, a list  of elements in  $\mathcal{K} \times \{ 0, 1 \}^k $,\\ ~~~~~initially  $[({\kappa}^0, 1, \cdots 1)]$
				}\EndPart
				
				\Statex
				\Statex  {\bf \Readtr{$ o_{i_1},  o_{i_2}, \cdots, o_{i_p}$}}	
				\Part{{\underline{{\readValue}}}}{ 
					\State $I\triangleq \{i_1, i_2, \cdots, i_p \}$
					\For{$i \in I$} 
					\State $j^* \leftarrow \max_{1 \leq j \leq |List|} \{j:$
					\\~~~~~~~~~~~$List[j].b_i = 1\}$
					\State ${\kappa}_i \leftarrow List[j^*].{\kappa}$
					\State  Send (\readValueTag, ${\kappa}_i$) to $s_i$
					\EndFor
					\Statex 
					\State  Await responses  $v_{i}$ from  $s_i$ $\forall$ $i\in I$
					\State Return  $(v_{i_1}, v_{i_2}, \cdots, v_{i_p})$
				}\EndPart
				\\\hrulefill
				\Statex
				\Statex {\bf Response routines}
				\Part{{\underline{On recv  (\informReaderTag,}
						\\\underline{$({\kappa}, (b_{1}, \cdots b_{k}))$) from  $w$}}}{ 
					\Statex 
					\State $List  \leftarrow List \bigoplus~ ({\kappa}, (b_{1}, \cdots b_{k}))$ 
					\\~~~/* $\bigoplus$ for append */
					\State $tag \leftarrow |List|$ /* $| \cdot |$ list size */
					\State Send  ({\ackTag}, $tag$) to  $w$
				}\EndPart
		}\end{multicols}
		
		\vspace{-1.5em}
		\\\hrulefill %
		\vspace{-1.5em}
		\begin{multicols}{2}{\footnotesize
				\Statex {\bf At server $s_i$ for any $i \in [k]$}
				\Part{{\it State Variables}}{ 
					
				\Statex $Vals\subset \mathcal{K} \times \mathcal{V}_i$, initially   $\{(t^0_{key}, v_i^0)\}$
				}\EndPart
				\Statex
				
				\Part {\underline{On recv (\writeValueTag, $({\kappa}, v)$) 
						from $w$}} {
					\State $Vals \gets   Vals \cup \{({\kappa}, v)\}$ 
					\State Send {\ackTag} to $w$.
				}\EndPart
				\Statex
				\Part{ \underline{On recv (\readValueTag, ${\kappa}$) from  $r$ }} {
					\State   Send $v$ s.t. $({\kappa}, v) \in Vals$  to $r$
				}\EndPart	
		}\end{multicols}
	\end{algorithmic}	
	\caption{Steps at writer $w$, reader $r$ and server $s_i$ in $A$.}\label{fig:algo_a}
\end{algorithm}	
$A$ respects the SNOW properties as stated below.


\begin{theorem} Any well-formed  and fair execution of $A$ 
		 guarantees all of the SNOW properties.
	\end{theorem}

\begin{proof} Below we show that $A$ satisfies the  SNOW properties. 
	
	\noindent{\emph{\underline{S property:}}} 
	Let $\beta$ be any fair execution  of  $A$ and 
 suppose all clients in $\beta$ behave in a well-formed
manner. Suppose $\beta$ contains no incomplete transactions and let  $\Pi$ be the set of transactions in $\beta$.  We define an irreflexive partial ordering ($\prec$) among the transactions in $\Pi$ as follows:  if $\phi$ and $\pi$ are any two distinct transactions in $\Pi$ then we say 
	$\phi \prec \pi$ if either $(i)$ $tag(\phi) < tag(\pi)$ or $(ii)$ $tag(\phi) = tag(\pi)$ and $\phi$ is a {\sc write} and $\pi$ is a {\sc read}. We will prove the $S$  (strict-serializability) property of $A$ by proving that the properties $P1$, $P2$, $P3$ and $P4$ of Lemma~\ref{lem:equivalence_app} hold for $\beta$. 
	
	\emph{P1:}   If $\pi$ is a {\sc read} then since all {\sc read}s are invoked by a single reader $r$ and in a well-formed manner, 
	therefore, there cannot be an infinite number of {\sc read}s such that they all 
	precede $\pi$ (w.r.t $\prec$).
	 Now, suppose $\pi$ is a {\sc write}. Clearly, from an inspection of the algorithm, 
	 $tag(\pi) \in \mathbb{N}$. From inspection of the algorithm, each {\sc write} increases the size of 
	 $List$, and the value of the tags are  defined by the size of $List$. Therefore, there can be at 
	 most a finite number of {\sc write}s such that can precede $\pi$ (w.r.t. $\prec$) in $\beta$.
	  
	\emph{P2:}  Suppose $\phi$ and $ \pi$ are any two transactions in $\Pi$, such that, $\pi$ begins after $\phi$ completes. 
	Then we show that we cannot have $\pi \prec \phi$. Now, we consider four cases, depending on whether $\phi$ and $\pi$ are {\sc read}s or {\sc write}s.	
	\begin{enumerate}
	    \item [$(a)$] $\phi$ and $\pi$ are {\sc write}s invoked by writers $w_{\phi}$ and $w_{\pi}$, respectively. Since the size of $List$, in $r$,  grows monotonically with each {\sc write}  hence  $w_{\pi}$ receives the  tag at least as high as $tag(\phi)$, so $\pi\not \prec \phi$.
	       \item [$(b)$] $\phi$ is a {\sc write}, $\pi$ is a {\sc read} transactions invoked by writer $w_{\phi}$ and $r$, respectively.  
	        Since the size of $List$, in $r$,  grows monotonically, and because  $w_{\pi}$ invokes $\pi$ after $\phi$ completes hence  $tag(\pi)$ is at least as high as $tag(\phi)$, so $\pi\not \prec \phi$.
	        \item[$(c)$] $\phi$ and $\pi$ are {\sc read}s  invoked by reader $r$. 
	           Since the size of $List$, in $r$,  grows monotonically,  hence  $w_{\pi}$ invoked $\pi$ after $\phi$ completes hence $tag(\pi)$ is at least as high as $tag(\phi)$, so $\pi\not \prec \phi$.
	         \item [$(d)$] $\phi$ is a {\sc read}, $\pi$ is a {\sc write}  invoked by reader $r$ and $w_{\pi}$, respectively.
	         This case is simple because new values are added to $List$  only  by writers, and $tag(\pi)$ 
	         is larger than the tag of $\phi$ and hence   $\pi\not \prec \phi$. 
	\end{enumerate}
	
	\emph{P3:} This is clear by the fact that any {\sc write} transaction always creates a unique tag and all tags are totally ordered since they all belong to $\mathbb{N}$
	
	\emph{P4:} Consider a {\sc read} $\rho$ as $READ(o_{i_1}, o_{i_2}, \cdots, o_{i_q})$, in $\beta$. 
Let the returned value from $\rho$ be $\mathbf{v} \equiv $$(v_{i_1}, v_{i_2}, \cdots, v_{i_q})$ such that 
$1 \leq {i_1} <  {i_2} <  \cdots <  {i_q} \leq k$, where value  $v_{i_j}$ corresponds to $o_{i_j}$. 
	Suppose $tag(\rho) \in \mathbb{N}$ was created during some {\sc write} transaction, say $\phi$, i.e., $\phi$ is the {\sc write} that 
	added the elements in index $(tag(\rho)-1)$ of $List$. Note that element in index $0$ contains the initial value.
	 Now we consider two cases:
	 
	\emph{Case $tag(\rho) = 1$.} We  know that it corresponds the initial default value $v_i^0$ at each sub-object $o_i$, and this equates to $\rho$ returning the default initial value for each sub-object.
	 %
	 
	 \emph{Case $tag(\rho) > 1$.} Then we argue that there exists no {\sc write} transaction, say $\pi$, that updated object $o_{i_j}$,   in $\beta$, such that,  $\pi \neq \phi$ and $\rho$ returns values written by $\pi$ and $\phi \prec \pi \prec \rho$. Suppose we assume the 	contrary, which means $tag(\phi) < tag(\pi) < tag(\rho)$. The latter implies $tag(\phi)  = tag(\pi)$ which is not possible because 
	this contradicts the fact that for any two distinct {\sc write}s $tag(\phi) \neq tag(\pi)$  in any execution of   $A$.
	
	\noindent{\emph{\underline{N property:}}}  By inspection of algorithm $A$ for the  response steps  of the servers to the reader.
	
	\noindent{\emph{\underline{O property:}}} By inspection of the  {\readValue} phase: it consists of one round of communication between the reader and the servers, where the servers send only one version of the value of the object it maintains.
	
	\noindent{\emph{\underline{W property:}}}  By inspection of the {\sc write} transaction steps, and  and  that writers always get to complete the transactions they invoke.
	\end{proof}

\section{No Prior Bounded Latency for SW}\label{sec:eiger}
The SNOW work~\cite{SNOW2016} claimed, after examining existing, work there existed only one system, Eiger~\cite{Lloyd:nsdi2013}, whose \rots{} had bounded latency---i.e., non-blocking and finish in three rounds---while providing the strongest guarantees--i.e., having properties W and S---because Eiger claimed that its \rots{} provide strict serializability within a datacenter.
%
In this section, we correct this claim and show there were no existing algorithms that had bounded latency while providing the strongest guarantees by proving Eiger's \rots{} are not strictly serializable. 

\begin{figure}[t]
\centering
\includegraphics[width=0.6\columnwidth]{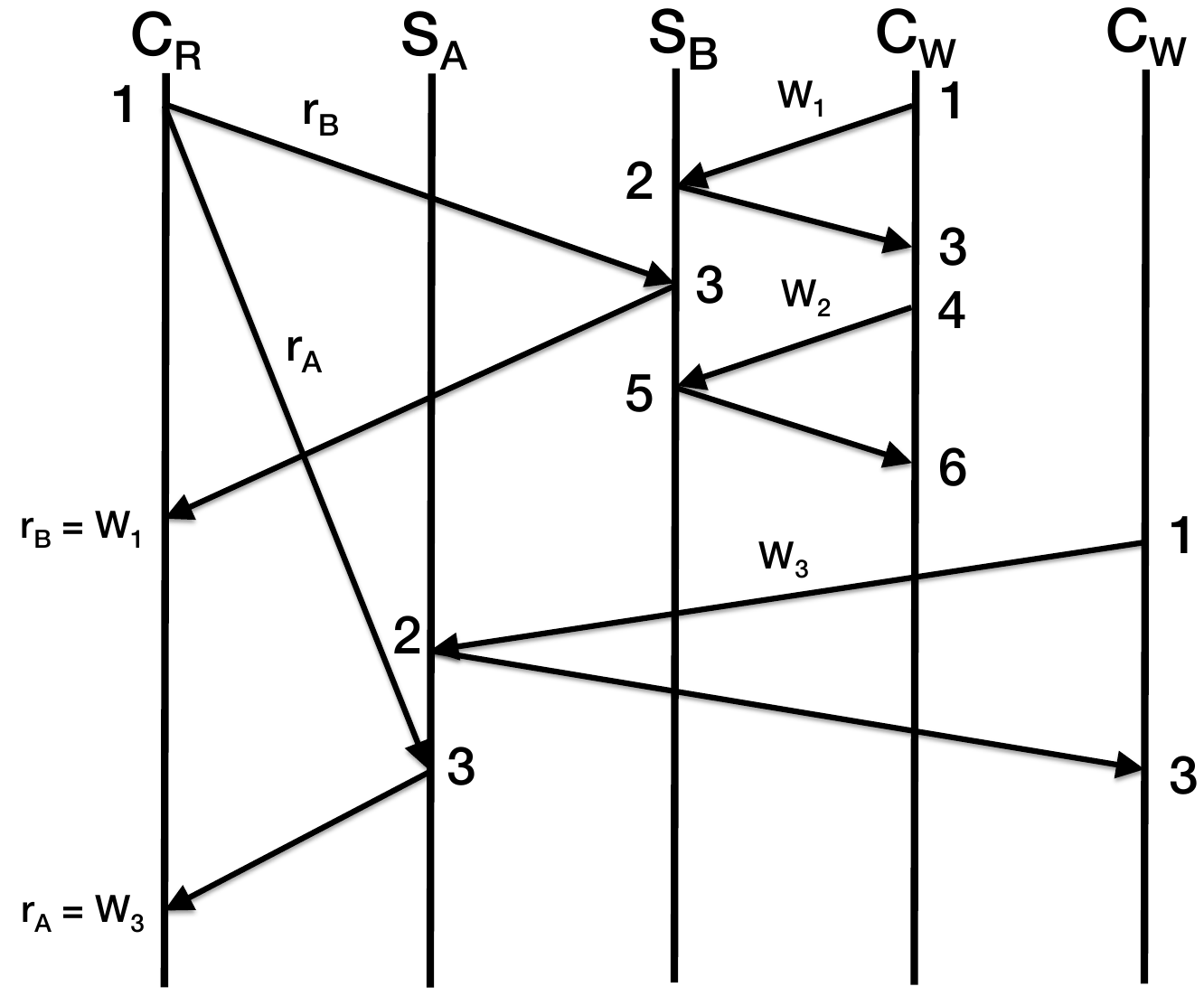}

\caption{
	\small{An example execution that shows Eiger's \textsc{READ} transaction is not strictly serializable. $w_1$, $w_2$, and $w_3$ are three writes, where $w_3$ is issued after $w_2$ finishes. $r_A$ and $r_B$ are read operations from the same \textsc{READ} transaction $R$, which is concurrent with all three writes. Each number is a value of the logical clock on the machine (process) as a result of message exchange.}}
\label{fig:eigerrot}

\end{figure}
\remove{
\subsection{Eiger's \rots{} are not Strictly Serializable}
\label{app:eiger_not_s}
Eiger~\cite{Lloyd:nsdi2013} targeted a geo-replicated setting, which provided causal consistency across wide areas with fast \textsc{READ} and \textsc{WRITE} transactions handled by each datacenter locally. While not focusing on making transactions strictly serializable, Eiger made a claim that its \textsc{READ} transactions within the local datacenter were strictly serializable. Thus, in SNOW, there was a claim that Eiger was the only system, which had bounded latency while providing the strongest guarantees, and the practical lower bound for property $O$ in existing SNW algorithms was three because Eiger's \textsc{READ} transactions required at most three rounds, i.e., the fewest in all examined existing work, were non-blocking, compatible with \textsc{WRITE} transactions, and were claimed to be strictly serializable. In this section, we correct this claim to be that there was no algorithm, which had bounded latency while providing the strongest guarantees by proving that Eiger's \textsc{READ} transactions are not strictly serializable. 
}

The insight on why Eiger's \textsc{READ} transactions are not strictly serializable is that Eiger uses logical timestamps, i.e., Lamport clocks, to track the ordering of operations, and logical clocks are not able to identify the real-time ordering between operations that do not have causal relationship, i.e., operations from different processes. Strict serializability, however, requires the real-time ordering to be respected.

Figure~\ref{fig:eigerrot} is an example execution, which is allowed by Eiger but violates strict serializability. 
Real time goes downwards in the diagram. Each number is a value of the Lamport clock on the machine (process) as a result of message exchange. Initially all processes have Lamport clock value $0$, and no messages happened before the execution in Figure~\ref{fig:eigerrot}. A \textsc{READ} transaction $R=\{r_A, r_B\}$ reads values on $S_A$ and $S_B$ respectively. Due to the asynchronous nature of the network, $r_B$ arrives on $S_B$ before $w_2$ and $r_A$ arrives after $w_3$. Following the \textsc{READ} transaction algorithm of Eiger, $r_A$ returns the value of $w_3$ and its valid logical duration, i.e., $[2,3]$. Similarly, $r_B$ returns the value of $w_1$ and its valid logical duration, i.e., $[2,3]$. Because the two logical durations overlap, Eiger claims the combined values of $r_A$ and $r_B$ are consistent and accept them. However, because $w_3$ starts after $w_2$ finishes, $w_3$ is in real time after $w_2$. By strict serializability, if a \textsc{READ} transaction sees the value of $w_3$, then it must observe the effect of $w_2$. Hence, $R$, which returns the values of $w_1$ and $w_3$, violates strict serializability.

\section{Condition for proving strict serializability}\label{sc_condition_app}
In this section, we describe the data types used in our algorithms. Also, we derive useful properties in proving \emph{strict serializability} property (S property) of executions of the algorithms that implement objects of the data type.

\subsection{Data type} \label{datatype}
In this section, we formally describe the  data type, which we denote as $\mathcal{O}_T$,   for the transaction processing systems considered in this paper. 
We assume there is a set of $k$ objects,  where $k$ is some positive integer, and  $o_k$ denotes the $k^{th}$ object. 
Object $o_k$ stores a value from some non-empty domain $V_k$ and supports two types of operations: \emph{read} and \emph{write} operations.
 A read operation on object $o_k$, denoted by $\readop{o_k}$, on completion returns the value stored in $o_k$.
  A write operation on $o_k$ with some value $v_k$ from $V_k$, denoted as $\writeop{o_k}{v_k}$,  on completion,  updates the value of  $o_k$.

\sloppy A {\sc write} transaction  consists  of a subset of $p$ distinct write operations for  a  subset of $p$ distinct objects, where $1 \leq p \leq k$.   For example, a {\sc write} transaction with  the set of operations   $\{ \writeop{o_{i_1}}{v_{i_1}}$, $\writeop{o_{i_2}}{v_{i_2}} \cdots \writeop{o_{i_p}}{v_{i_p}} \}$, means  value $v_{i_1}$  is to be written to object $o_{i_1}$,   value $v_{i_2}$ to object $o_{i_2}$, and so on.
 We denote such a {\sc write} transaction as $WRITE( (o_{i_1}, v_{i_1}), (o_{i_2}, v_{i_2}), \cdots, $
 $(o_{i_p}, v_{i_p})  )$. 
 
 A {\sc read} transaction consisting of a set of read operations is denoted as 
  $READ(o_{i_1}, o_{i_2}, \cdots, o_{i_q})$,  where $o_{i_1}, o_{i_2}, \cdots, o_{i_q}$ is a set of distinct objects,  $q$ is any positive integer, $1 \leq q \leq k$, which upon completion  returns the  values 
 $(v_{i_1}, v_{i_2}, \cdots, v_{i_q})$, where  $v_{i_j} \in V_{i_j}$  is the value returned by $\readop{o_{i_j}}$, for any $i_j  \in  \{i_1, i_2, \cdots i_q\}$ and  $1 \leq {i_1} < {i_2}< \cdots <  {i_q} \leq k$. 
%
 
Formally, we define the value type used in this paper for a  $k$ object data-type  as follows: 
\begin{enumerate}
\item [$(i)$] a tuple $(v_1, v_2, \cdots, v_k) \in \Pi_{i=1}^{k} V_i$, where $V_i$ is a domain of values, for each $i \in \{1, \cdots, k\}$;
\item [$(ii)$] an initial value $v^0_i$, $v^0_i \in V_i$  for each object $o_i$; 
\item[$(iii)$] \emph{invocations}: $READ(
 o_{i_1},  o_{i_2}, \cdots, o_{i_p})$  and  $WRITE((o_{i_1}, v_{i_1}), $
 $(o_{i_2}, v_{i_2}), \cdots, $ $(o_{i_p}, v_{i_p}))$,
  such that $v_{i_j}  \in V_{i_j}$ for any  $i_j  \in  \{i_1, i_2, \cdots i_p\}$  where  $1 \leq {i_1} < {i_2}< \cdots <  {i_p} \leq k$ and $p$ is   some integer with $1 \leq p \leq k$; 
\item[$(iv)$] \emph{responses}:    a tuple $(v_1, v_2, \cdots, v_k) \in \Pi_{i=1}^{k} V_i$ and $ok$; and 
\item [$(v)$] for any subset of  $p$ objects we  define  $f: invocations \times V \rightarrow responses \times V$, such that:
\begin{enumerate} 
\item $f( READ(o_{i_1}, o_{i_2}, \cdots, o_{i_p}),$ $ (v_1, v_2, \cdots, v_k) ) = (  (v_{i_1}, v_{i_2}, \cdots, v_{i_p}), (v_1, v_2, \cdots, v_k) )$; and 
\item $f( WRITE((o_{i_1}, u_{i_1}), (o_{i_2}, u_{i_2}), \cdots, $ $(o_{i_p}, u_{i_p})),  $ $ (v_1, v_2,$ $ \cdots, v_k)   ) = (ack,  (w_1, w_2,  \cdots, w_k) )$, where  for any $i$,  $w_i =  u_i$ if  
$i \in \{ i_1, i_2,  \cdots, i_p \}$, and for all other values of $i$,  $w_i=v_i$.
\end{enumerate} 
\end{enumerate}

\subsection{Strict-serializability in $\mathcal{O}_T$}
Although the strict serializability property in transaction-processing systems is a well-studied topic,  
the specific setting considered  in this paper is much simpler. Therefore, this allows  us to derive
 simpler conditions to prove the safety of these algorithms.  A wide range of transaction types and 
 transaction processing systems are considered in the literature. For example, in ~\cite{Papadimitriou79}, Papadimitriou   
defined the strict serializability conditions as a part of developing a  theory for analyzing transaction 
processing systems. In this work, each transaction $T$ consists of a set of write operations $W$, at 
individual objects, and a set of read operations $R$  from individual objects, where  the operations in 
$W$  must complete before the operations in $R$ execute.  
Other types of transaction processing systems allow nested transactions
\cite{Gary1993},
 where the transactions may contain sub-transactions~\cite{Bernstein:1987} which may 
further contain a mix of read or write operations, or even child-transactions. 
In most transaction processing systems considered in the literature,  transactions can be 
\emph{aborted} so as  to handle  failed  transactions. As a result, the serializability theories are 
developed  while  considering the presence of  aborts. 
However, in our system, we do not consider any abort, nor any  client or  server failures.  A transaction 
in our system is either a set of independent writes or a set of reads 
with 
all the reads or writes in a transaction operating on different  objects. Such simplifications allow us to 
formulate an equivalent condition for the execution of an algorithm  to prove the S  
property of such algorithms while  implementing an object of data type $\mathcal{O}_T$.

\sloppy We note that an execution of a variable of type $\mathcal{O}_T$ is a finite sequence 
$\mathbf{v}_0, INV_1, RESP_1,$ $\mathbf{v}_1, INV_2,$ $ RESP_2, \mathbf{v}_2, \cdots, \mathbf{v}_r$ or an infinite sequence 
$\mathbf{v}_0, INV_1, RESP_1, \mathbf{v}_1, INV_2, RESP_2, \mathbf{v}_2, \ldots$, where $INV$'s and 
$RESP$'s are invocations and 
responses, respectively. The $\mathbf{v}_i$s are tuples  of the the form $(v_{1}, v_{2}, \cdots, v_{k}) \in \Pi_{i=1}^{q} V_{i}$, that 
corresponds to the latest values stored across the objects $o_1, o_2\cdots o_k$, and the values in  $\mathbf{v}_0$ are  the initial 
values of the objects. Any adjacent quadruple such as  $\mathbf{v}_i, INV_{i+1}, RESP_{i+1},$ $\mathbf{v}_{i+1}$  is consistent with the $f$ 
function for an object of type $\mathcal{O}_T$. Now, the safety property of such an 
object is a trace that describes the 
correct response to a sequence of $INV$s when all the transactions are executed sequentially. The \emph{strict serializability} of 
$\mathcal{O}_T$ says that each trace produced by an execution of $\mathcal{O}_T$ with concurrent transactions appears as some trace of $\mathcal{O}_T$. We describe this below in more detail. 

\begin{definition}[Strict-serializability] Let us consider an execution $\beta$ of an object of type $\mathcal{O}_T$, such that the invocations of any transaction at any 
client respects the well-formedness property.  Let $\Pi$ denote the set of complete transactions in $\beta$ then we say $\beta$ satisfies the strict-serializability property for 
$\mathcal{O}_T$ if the following are possible:
\begin{enumerate}
\item[$(i)$] For every complete {\sc read} or {\sc write} transaction $\pi$ we insert a point (serialization point) $\pi_*$ between the actions   $\INV{\pi}$ and $\RESP{\pi}$.
\item[$(ii)$] We select a set $\Phi$ of incomplete transactions in $\beta$ such that for 
each $\pi \in \Phi$ we   select a response $\RESP{\pi}$.
\item[ $(iii)$] For each $\pi \in \Phi$ we insert $\pi_*$ somewhere after $\INV{\pi}$ in $\beta$, and remove the $INV$ for the rest of the incomplete transactions in $\beta$.
\item[$(iv)$] If we assume for each $\pi \in \Pi \cup \Phi$ both $\INV{\pi}$ and $\RESP{\pi}$ to occur consecutively at $\pi_*$, with the interval of the transaction shrunk to $\pi_*$,   then the sequence of transactions in this  new trace  is a trace of an object of data type  $\mathcal{O}_T$.
\end{enumerate}
\end{definition}
Now, we consider any automaton $\mathcal{B}$ that implements an object of  type $\mathcal{O}_T$, and prove a result that serves us an  equivalent condition for proving  the strict serializability property of $\mathcal{B}$. Any trace property $P$ of an automaton  is a \emph{safety property}  if the set of executions in $P$ is non-empty; \emph{prefix-closed}, meaning any prefix of an execution in $P$ is also in $P$; and \emph{limit-closed}, i.e., if $\beta_1$, $\beta_2, \cdots$ is any infinite sequence of executions in $P$ is such that $\beta_i$ is prefix of $\beta_{i+1}$ for any $i$, then the limit $\beta$ of the sequence of executions $\{\beta_i\}_{i=0}^{\infty}$  is also in $P$. From Theorem~13.1 in~\cite{Lynch1996}, we know that the trace property, which we denote by $P_{SC}$, of 
any well-formed execution of $\mathcal{B}$ that satisfies the strict-serializability property is a safety property. Moreover, from 
Lemma 13.10 in~\cite{Lynch1996} we can deduce that  if  every execution of $\mathcal{B}$ that is well-formed and failure-free, and also contains no incomplete transactions,  satisfies $P_{SC}$, then any well-formed  execution of $\mathcal{B}$ that can possibly have incomplete transactions is also in $P_{SC}$. 
Therefore, in the following lemma,  which gives us an equivalent condition for the strict serializability property of an execution $\beta$, we consider only executions without any incomplete transactions.
The lemma is proved in a manner similar to Lemma 13.16 in~\cite{Lynch1996}, for atomicity guarantee of a single multi-reader multi-writer object.

\begin{lemma}\label{lem:equivalence_app}
Let $\beta$ be an execution (finite or infinite) of  an  automaton $\mathcal{B}$ that implements an object of type $\mathcal{O}_T$, which consists of  a set of $k$ sub-objects.
 Suppose all clients in $\beta$ behave in an well-formed
manner. Suppose $\beta$ contains no incomplete transactions and let  $\Pi$ be the set of transactions in $\beta$. Suppose there exists an irreflexive partial ordering ($\prec$)  among  the transactions
 in $\Pi$, such that,
 \begin{enumerate}
  \item [$P1$] For any transaction $\pi \in \Pi$ there are only a finite number of transactions $\phi \in \Pi$ such that 
    $\phi \prec \pi$;
 \item [$P2$] If the response event for $\pi$ precedes the invocation event for $\phi$ in $\beta$, then it cannot be that $\phi \prec \pi$;
 \item [$P3$]If $\pi$ is a {\sc write} transaction  in $\Pi$ and $\phi$ is any transaction in $\Pi$, then either $\pi \prec \phi$ or $\phi \prec \pi$; and 
 \item [$P4$] \sloppy A tuple 
 $ \mathbf{v} \equiv $$(v_{i_1}, v_{i_2}, \cdots, v_{i_q})$  
returned by a
$READ(o_{i_1}, o_{i_2}, \cdots, o_{i_q})$,  where $q$ is any positive integer, $1 \leq q \leq k$, 
 is such that  $v_{i_j}$  $j \in \{1, \cdots, q\}$ is written in $\beta$ by the last preceding (w.r.t. $\prec$)  {\sc write} transaction that contains a $\writeop{o_{i_j}}{*}$,
 or the initial value $v_{i_j}^0$ if no such {\sc write} exists in $\beta$. 
 \end{enumerate}
Then execution $\beta$ is strictly serializable.
\end{lemma}

\begin{proof}We discuss how to insert a serialization point $*_{\pi}$ in $\beta$ for every transaction $\pi \in \Pi$. 
First, we add  $*_{\pi}$ immediately after the latest of the invocations of $\pi$ or $\phi \in \Pi$ such that $\phi \prec \pi$.
We want to stress that any complete operation $\pi$ in $\Pi$  refers to an operation, with the invocation and response events, whereas
$\pi_*$ refers to a point in the executions.
Note that according to condition $P1$ for $\pi$ there are only finite number of such  invocations in
 $\beta$, therefore, $\pi_*$ is well-defined for any $\pi \in \Pi$. Now, since the order of the invocation events of the 
transactions in $\Pi$ are already defined, therefore,  the order of the corresponding set of serialization points are 
well-defined, except for the case when more than one serialization points are placed immediately after an 
invocation.
 In the case such multiple  serialization points corresponding to an invocation we order  these 
serialization points  in accordance with the $\prec$  relation of  the underlying transactions.

Next, we show that for any pair of transactions $\phi$, $\pi \in \Pi$ if $\phi \prec \pi$ then in $\beta$ we have $*_{\phi}$ precedes $*_{\pi}$.
Suppose $\phi \prec \pi$. By construction,  each of $\pi_*$ and $\phi_*$ appear immediately after 
some invocation, in $\beta$,  of some transaction in $\Pi$. If both $\pi_*$ and $\phi_*$ appear immediately  after 
the same invocation, then since $\phi \prec \pi$,  by construction of $\pi_*$,  the serialization point $\pi_*$ appears in $\beta$ after 
$\phi_*$.  Also, if the invocations after which $\pi_*$ and $\phi_*$ appear are distinct,  then by 
construction of $\pi_*$ the serialization point  $\pi_*$ appear after $\phi_*$ in $\beta$ since $\phi \prec \pi$.


Next we argue that each $*_{\pi}$  serialization point  for any $\pi \in \Pi$ is placed between the 
invocation $\INV{\pi}$ and responses $\RESP{\pi}$. By construction, $*_{\pi}$ is after $\INV{\pi}$ in $\beta$. To 
show that $*_{\pi}$ is before $\RESP{\pi}$ for the sake   of contradiction assume that  $*_{\pi}$  
appears after $\RESP{\pi}$. By construction, $*_{\pi}$ must be after $\INV{\psi}$ for some 
$\psi \in \Pi$ and $\psi \neq \pi$, then by the condition of construction of $\pi_*$ we have $\psi \prec 
\pi$. But from above 
$\INV{\psi}$ occurs after $\RESP{\pi}$, i.e., $\pi$ completes before $\psi$ is invoked which means, by 
property $P2$,  we cannot have $\psi \prec \pi$, which is a contradiction.

Next, we show that 
if we were to shrink  the transactions intervals to  their corresponding serialization points, the 
 resulting  trace would be a trace of the underlying data type $\mathcal{O}_T$. 
In other words, we show any {\sc read}
$READ(o_{i_1}, o_{i_2}, \cdots, o_{i_q})$
returns the  values   $(v_{i_1}, v_{i_2}, \cdots, v_{i_q})$, such that each value
$v_{i_j}$, $j \in [q]$,  was written by the immediately preceding  (w.r.t. the serialization points)  {\sc 
write} that contained  $\writeop{o_{i_j}}{v_{i_j}}$ or the initial values if no such previous {\sc write} exists.
Let us denote the set of {\sc write}s that precedes (w.r.t. $\prec$) $\pi$ by $\Pi_W^{\prec\pi}$, i.e., 
$\phi \in \Pi^{\prec\pi}_W$  $\phi$ is a write and $\phi \prec \pi$.  By property $P3$, all transactions in 
$\Pi_{W}^{\prec\pi}$ are totally-ordered.  By property $P4$,  $v_{i_j}$ must be the value updated by
 the most recent {\sc write} in $\Pi_W^{\prec\pi}$.  Since the total order of serialization points are  consistent with 
 $\prec$ and hence the $v_{i_j}$ corresponds to the write operation of a {\sc write} transaction with 
  the most recent serialization point and  contains a operation of  type $\writeop{o_{i_j}}{*}$.
\end{proof}

\section{SNW +  One Version,  MWMR setting}\label{mwmr_snow_one_version}
Here present  algorithm  $B$, which satisfies  SNW and "one-version" properties, in MWMR setting where a \rot{} must consist of one version of the data but, possibly, multiple communication trips between the reader and the servers.
 In B, the steps for the writers are shown in Pseudocode~\ref{fig:algo_bc} and for readers and the servers are presented in Pseudocode~\ref{fig:algo_b}.
We assume  a set of writers $\mathcal{W}$,  a set of readers $\mathcal{R}$ and a set of $k \geq 1$ servers,  $\mathcal{S}$, with ids $s_1, s_2\cdots s_k$ that stores the objects $o_1, o_2, \cdots, o_k$, respectively.  A key ${\kappa}$ is defined as a pair $(z, w)$, 
where $z \in \mathbb{N}$ and $w \in \mathcal{W}$ the  id of a writer. We use $\mathcal{K}$ to denote the set of all possible keys. 
In $B$, a key uniquely identifies some transaction. Also, with each transaction we associate a tag $t \in \mathbb{N}$. 
	
In  $B$, we designate one of the servers as coordinator, denote as $s^*$,
 for the transactions. The $s^*$  maintains the order of the \wots{} and the objects that are updated during the \wot{} in the variable $List$.  

\textit{\textbf{State variables:}} Each of the writers and servers maintain a set of state variables as follows: $(i)$ At  any  \emph{writer $w$}, there is  a counter $z$ to keep track of the
 number of \wot{}  the writer has  invoked, initially $0$. $(ii)$ At any  \emph{server}, $s_i$, 
 for $i \in [k]$, there is  a set variable $Vals$ 
 with elements 
 that are  key-value pairs $({\kappa}, v_i) \in \mathcal{K} \times \mathcal{V}_i$. Initially,
  $Vals= \{ ({\kappa}^0, v_i^0)\}$. 
A server also contains an
 ordered list variable $List$  of elements  as $({\kappa}, (b_1, \cdots, b_k))$,  where 
 ${\kappa}  \in \mathcal{K}$  and 
 $(b_1, \cdots b_k) \in  \{0, 1\}^k$. Initially,  
 $List= [ ({\kappa}^0, (1, \cdots 1) ]$, where ${\kappa}^0  \equiv (0, w_0)$, 
 where $w_0$ is any
  place holder identifier string for writer id. The elements in $List$ can be identified with an index, e.g., 
  $List[0] =({\kappa}^0, (1, \cdots, 1))$.  Essentially, a $(k+1)$-tuple  $(\kappa, (b_1, \cdots, b_k))$ in $List$ corresponds to a \wot{} and 
   identifies the set of objects that are updated during the \wot{}, i.e., if $b_i=1$ then object 
   $o_i$ was updated  during  the  \wot{}, otherwise $b_i=0$.
					
\textit{\textbf{Writer steps:}} A \wot{} updates
			a  list of $p$ objects $o_{i_1}, o_{i_2}, \cdots o_{i_p}$  with values
			 $v_{i_1}, v_{i_2}, \cdots v_{i_p}$, respectively, is invoked at $w$ via the procedure
			$\Writetr{ (o_{i_1}, v_{i_1}), \cdots, (o_{i_p}, v_{i_p}) }$.
We use the notations:  $I \triangleq \{i_1, i_2, \cdots, i_p\}$  and $S_I\triangleq \{s_{i_1}, s_{i_2}, \cdots, s_{i_p}\}$.
This procedure  consists of two phases: {\writeValue} and {\updateCoord}. 
During the {\writeValue} phase,  $w$ creates a new key ${\kappa}$ as 
 $ {\kappa}  \equiv (z + 1, w)$, where $w$ identifies the writer; and also increments the local counter $z$ by one.  Then $w$ sends $(${\writeValueTag}$, ({\kappa}, v_{i}))$ to each server in $S_I$, and awaits {\ackTag}  
from all servers in $S_I$.
After receiving {\ackTag} from all servers in $S_I$,  $w$
initiates the {\updateCoord} phase where it sends 
(\updateCoordTag, $({\kappa}, (b_{1}, \cdots b_{k})$) to $s^*$, where for any $i \in [k]$,  $b_i=1$ if $s_i \in S_I$, 
otherwise $b_i=0$,   and completes then \wot{} after it receive a   ({\ackTag}, $t_w$) from $s^*$.  


\textit{\textbf{Reader steps:}}
We use the same notations for $I$ and $S_I$ as above but the indices can vary across  transactions.
The procedure  \Readtr{$ o_{i_1},  o_{i_2}, \cdots, o_{i_p}$} can be 
initiated by  some reader  $r$,   as a \rot{}, intending to read the values of 
subset $o_{i_1},  o_{i_2}, \cdots, o_{i_p}$ of the objects. The procedure 
consists of two consecutively executed phases of communication rounds
between the $r$ and the  servers, viz.,  {\getTagArray} and {\readValue}. 
During  the  phase {\getTagArray},  $r$ sends $s^*$ the message  {\getTagArrayTag}  
 requesting the  list of the latest added keys for each object. 
 Once $r$ receives a list of tags, such as, $(t_r, ({\kappa}_1, {\kappa}_2, \cdots,  {\kappa}_k))$ from $s^*$   the phase completes.
In the subsequence phase, {\readValue},   $r$ requests each server $s_i$ in $S_I$ by sending the message 
$(\text{\readValueTag}, \kappa_i)$. 
%
After  receiving the values $v_{i_1}$, $v_{i_2}, \cdots v_{i_p}$ from the servers in $\mathcal{S_I}$, 
 $r$ completes the transaction  by 
 returning the tuple of values $(v_{i_1}, \cdots v_{i_p})$.

\textit{\textbf{Server steps:}}  
When a server $s_i$ receives a message of type $(${\writeValueTag}$, ({\kappa}, v_{i}))$ from a writer $w$ then 
it  adds $({\kappa}, v_i)$ to its set variable  
$Vals$ and sends {\ackTag} back to $w$.

If the coordinator $s^*$ receives  (\updateCoordTag, $({\kappa}, (b_{1}, \cdots, b_{k})$) from writer $w$, then it appends  
			 $({\kappa}, (b_{1}, \cdots, b_{k}))$ to its  $List$,  and responds with  
			 {\ackTag} and $t_{w}$ (set to be the number of elements in the local list $List$)  to $w$.
			 The order of the  elements in  $List$ corresponds to  the order  
the \wots{}, the order of the incoming  {\updateCoordTag} updates,  as seen by $s^*$.
	
When $s^*$  receives  the message  {\getTagArrayTag} from $r$  it responds with 
$(\kappa_1, \cdots, \kappa_k)$ such that for each $i \in [k]$, $\kappa_i$ is the key  part of the $(k+1)$-tuple that was modified
last, i.e., 	${\kappa}_i = List[j^*].{\kappa}$ such that 	 $j^* \triangleq\max \{ j : List[j].b_i =1 \}$, and 
$t_r$, $t_r \triangleq \max_{1 \leq j \leq |List|} \{ j : List[j].b_i = 1 \wedge i \in I\}$.
  	If any server $s_i$ receives a message  $(${\readValueTag}$, {\kappa})$ from a reader $r$ then it responds to $r$ with 
   the value $v_i$ corresponding to key with value  $\kappa$ in  $Vals$. 


\begin{theorem} Any well-formed  and fair execution of algorithm $B$  satisfies the SNW and "one-version"  properties. 
\end{theorem}


\begin{proof} Below we show that algorithm $B$ satisfies the  SNoW properties. 
	
	\noindent{\emph{\underline{S property:}}} 
	Let $\beta$ be any fair execution  of  $B$ and 
 suppose all clients in $\beta$ behave in a well-formed
manner. Suppose $\beta$ contains no incomplete transactions and let  $\Pi$ be the set of transactions in $\beta$.  We define an irreflexive partial ordering ($\prec$) in $\Pi$ as follows:  if $\phi$ and $\pi$ are any two distinct transactions in $\Pi$ then we say 
	$\phi \prec \pi$ if either $(i)$ $tag(\phi) < tag(\pi)$ or $(ii)$ $tag(\phi) = tag(\pi)$ and $\phi$ is a \wot{} and $\pi$ is a \rot{}. Below we prove the $S$  property of $B$ by showing that  properties $P1$, $P2$, $P3$ and $P4$ of Lemma~\ref{lem:equivalence_app} hold for $\beta$. 
	
	\emph{P1:}   Clearly, from an inspection of the algorithm, 
	 $tag(\pi) \in \mathbb{N}$. From inspection of the algorithm, each \wot{} increases the size of 
	 $List$, and the value of the tags are  defined by the size of $List$. Therefore, there can be at 
	 most a finite number of \wots{} such that can precede $\pi$ (w.r.t. $\prec$) in $\beta$.
	  On the other hand, if $\pi$ is a \rot{} then since all \rot{}s are invoked by readers  in a well-formed manner,  and there are only finite number of readers 
	therefore, there cannot be an infinite number of \rot{}s such that they all 
	precede $\pi$ (w.r.t $\prec$).

	\emph{P2:}  Suppose $\phi$ and $ \pi$ are any two transactions in $\Pi$, such that, $\pi$ begins after $\phi$ completes. Then we show that we have  we cannot have $\pi \prec \phi$. Now, we consider four cases, depending on whether $\phi$ and $\pi$ are \rot{}s or \wots{}.	
	\begin{enumerate}
	    \item [$(a)$] $\pi$ and $\phi$ are \wots{} invoked by writers $w_{\pi}$ and $w_{\phi}$, respectively. Since the size of $List$, in $s^*$ grows monotonically due to each \wot{}  hence  $w_{\pi}$ receives the  tag  from $s^*$ at least as high as $tag(\phi)$, so $\pi\not \prec \phi$.
	       \item [$(b)$] $\pi$ is a \rot{}, $\phi$ is a \wot{} invoked by reader  $r_{\pi}$ and writer $w_{\phi}$, respectively.  
	        Since the size of $List$, in $s^*$,  grows monotonically,  because  $r_{\pi}$  invokes $\pi$ after $\phi$ completes hence $tag(\pi)$ is at least as high as $tag(\phi)$, so $\pi\not \prec \phi$.
	        \item[$(c)$]$\pi$ and $\phi$ are both \rot{}s  invoked by readers $r_{\pi}$ and $r_{\phi}$, respectively. 
	           Since the size of $List$, in $s^*$, grows monotonically,  because   $w_{\pi}$ invokes $\pi$ after $\phi$ completes hence $tag(\pi)$ is at least as high as $tag(\phi)$, so $\pi\not \prec \phi$.
	         \item [$(d)$] $\pi$ is a \wot{}, $\phi$ is a \rot{}  invoked by writer $w_{\pi}$ and reader $r_{\phi}$, respectively.
	         This case is simple because new values are added to $List$, in $s^*$,  only  by writers, and $tag(\pi)$  has to be larger than the tag of $\phi$ and hence   $\pi\not \prec \phi$. 
	\end{enumerate}

	\emph{P3:} This is from  the fact that any \wot{} always creates a unique tag and all tags are totally ordered since they all belong to $\mathbb{N}$
	
		\emph{P4:} Consider a \rot{} $\rho$ as $READ(o_{i_1}, o_{i_2}, \cdots, o_{i_q})$, in $\beta$. 
Let the returned value from $\rho$ be $\mathbf{v} \equiv $$(v_{i_1}, v_{i_2}, \cdots, v_{i_q})$ such that 
$1 \leq {i_1} <  {i_2} <  \cdots <  {i_q} \leq k$, where value  $v_{i_j}$ corresponds to $o_{i_j}$. 
	Suppose $tag(\rho) \in \mathbb{N}$ was created during some \wot{}, say $\phi$, i.e., $\phi$ is the \wot{} that 
	added the elements in index $(tag(\rho)-1)$ of $List$ at the coordinator $s^*$. Note that element in index $0$ contains the initial value.
	 Now we consider two cases:
	 
	\emph{Case $tag(\rho) = 1$.} We  know that it corresponds the initial default value $v_i^0$ at each sub-object $o_i$, and this equates to $\rho$ returning the default initial value for each sub-object.
	 %
	 
	 \emph{Case $tag(\rho) > 1$.} Then we argue that there exists no \wot{}, say $\pi$, that updated object $o_{i_j}$,   in $\beta$, such that,  $\pi \neq \phi$ and $\rho$ returns values written by $\pi$ and $\phi \prec \pi \prec \rho$. Suppose we assume the 	contrary, which means $tag(\phi) < tag(\pi) < tag(\rho)$. The latter implies $tag(\phi)  = tag(\pi)$ which is not possible because 
	this contradicts the fact that for any two distinct \wots{} $tag(\phi) \neq tag(\pi)$  in any execution of   $B$.

	\noindent{\emph{\underline{N, o and W properties:}}}  Evident from an inspection of the algorithm.
	\end{proof}

\begin{algorithm}[!ht]
	\label{app:algrithmBC}
	\begin{algorithmic}[3]
		
		\begin{multicols}{1}{\footnotesize
				\Statex  {\bf At writer $w$}
				\Part{{\it State Variables}}{   
					\Statex $z \in \mathbb N$, initially   $0$
				}\EndPart
				
				\Statex  {\bf  \Writetr{$(i_1, v_{i_1}), \cdots, (i_p,  v_{i_p})$}   }
				\State $I\triangleq \{i_1, i_2, \cdots, i_p \}$
				\Part{ \underline{\writeValue}} {
					\State ${\kappa} \leftarrow (z +1,  w)$
					\State $z \leftarrow z +1 $
					\For{$i \in I$} 
					\State Send (\writeValueTag, $({\kappa}, v_{s_i})$) \\ to $s_i$
					\EndFor 
					\State  Await {\ackTag} from servers in  $S_I$.
				}\EndPart
				\Statex
				\Part{ \underline{\informSerializer}} {
					
					\For{$i \in [k]$} 
					\If{$i \in I$}
					\State $b_i \leftarrow 1$
					\Else
					\State $b_i \leftarrow 0$
					\EndIf
					\EndFor 
					\State  Send  (\informSerializerTag, $({\kappa},$
					\\ $(b_{1}, \cdots, b_{k}))$) to   $s^*$
					
					\State Receive ({\ackTag}, 
					\\$t_w$) from $s^*$
				}\EndPart
		}\end{multicols}
	\end{algorithmic} 
	\caption{Protocol for writer $w$ in algorithms $B$ and $C$.}\label{fig:algo_bc}
\vspace{-1.2em}
\end{algorithm}

\begin{algorithm}[!ht]
\label{app:algrithmB}
  \begin{algorithmic}[3]
      \begin{multicols}{0}{\footnotesize    
          \Statex  {\bf At reader $r$}
          \Statex  {\bf \Readtr{$ o_{i_1},  o_{i_2}, \cdots, o_{i_p}$}} 
          \State $I\triangleq \{i_1, i_2, \cdots, i_p \}$
            \Part{ \underline{\getTagArray}} {
          \State  Send (\getTagArrayTag) to $s^*$
          \State  Receive  $(t_r, ({\kappa}_1, {\kappa}_2, \cdots, {\kappa}_k))$ from  $s^*$
        }\EndPart 
            \Statex
          \Part{{\underline{{\readValue}}}}{ 
                     \For{$i \in  I$} 
            \State  Send (\readValueTag, ${\kappa}_i$) to $s_i$
           \EndFor 
            \State  Await responses as $v_{i}$ $\forall$ $s_i\in S$
            \State Return  $(v_{i_1}, v_{i_2}, \cdots, v_{i_p})$
          }\EndPart
        }\end{multicols}
  \vspace{-1.5em}
\\\hrulefill
\vspace{-1.5em}
      \begin{multicols}{1}{\footnotesize
          \Statex {\bf At server $s_i$ for any $i \in [k]$}
            \Part{{\it State Variables}}{ 
              \Statex $Vals\subset \mathcal{K} \times \mathcal{V}_i$, initially   $\{({\kappa}^0, v_i^0)\}$
\Statex $List$, list  of  $\mathcal{K} \times \{ 0, 1 \}^k $, initially  $[({\kappa}^0, (1, \cdots 1))]$
            }\EndPart
            \Statex
            \Part {\underline{On recv (\writeValueTag,}
            	\\\underline{$({\kappa}, v)$) from  $w$}} {
                       \State   $Vals \gets   Vals \cup \{({\kappa}, v)\}$ 
              \State  Send {\ackTag} to $w$.
            }\EndPart
            \Statex
            \Part{{\underline{On recv  (\informSerializerTag, }} \\\underline{
            		$({\kappa}, (b_{1}, \cdots b_{k}))$)
            		 from  $w$} }{ 
               \State $List  \leftarrow $ 
                 \\~~~~~~ $List \bigoplus~ ({\kappa}, (b_{1}, \cdots b_{k}))$ 
                 \\~~~//$\bigoplus$ for append
                  \State $tag \leftarrow |List|$ // $| \cdot |$ list size
            \State Send  ({\ackTag}, $tag$) to  $w$
          }\EndPart
                   
                        \Statex
            \Part{ \underline{On recv (\readValueTag, ${\kappa}$) from  $r$ }} {
              \State   Send $v_i$ s.t. $({\kappa}, v) \in Vals$  to  $r$
            }\EndPart 
          \Statex\Statex   /* used only by $s^*$ */
          
          \Part{ \underline{On recv {\getTagArrayTag} from  $r$ }} {
          \For{$i \in [k]$} 
                     \State $j^* \leftarrow \max \{ j :$ 
                      \\~~~~~~~~~~~~$ List[j].b_i =1 \}$
                     \State ${\kappa}_i \leftarrow List[j^*].{\kappa}$
           \EndFor 
            \Statex
             \State  $t_r \triangleq \max_{1 \leq j \leq |List|} \{j:$
             \\~~~~~~~~$List[j].b_i = 1 \wedge i \in I\}$ 
            \State  Send  ($t_r, ({\kappa}_1,{\kappa}_2, \cdots, {\kappa}_k)$) to $r$
            }\EndPart       
          }\end{multicols}
        \end{algorithmic} 
        \caption{Protocols reader $r$ and server $s_i$ in alg. $B$.}\label{fig:algo_b}
                  \vspace{-1.2em}
\end{algorithm}

\label{app:algrithmC}
\begin{algorithm}[!ht]
  \begin{algorithmic}[3]
      \begin{multicols}{0}{\footnotesize    
          \Statex  {\bf At reader $r$}
          \Statex  {\bf \Readtr{$ o_{i_1},  o_{i_2}, \cdots, o_{i_p}$}} 
          \State $I\triangleq \{i_1, i_2, \cdots, i_p \}$
            \Part{ \underline{\readValuesAndTags}} {
          \State  Send  (\getTagArrayTag) to $s^*$
           \For{$i \in  I$} 
            \State  Send (\readValuesTag) to $s_i$
             \EndFor 
             
          \State  Recv $(t_r, ({\kappa}_1, {\kappa}_2, \cdots, {\kappa}_k))$ from  $s^*$
            \State Recv. $Vals_{i}$ from $\forall$ $s_i\in S_I$
            \State Return  $(v_{i_1}, v_{i_2}, \cdots, v_{i_p})$ 
            \\~~~~~~s.t.  $({\kappa}_{j}, v_{j}) \in Vals_{j}$, $j\in I$
          }\EndPart
        }\end{multicols}

  \vspace{-1.5em}
\\\hrulefill
\vspace{-1.5em}
      \begin{multicols}{2}{\footnotesize
          \Statex {\bf At server $s_i$ for any $i \in [k]$}
            \Part{{\it State Variables}}{ 
              \Statex $Vals\subset \mathcal{K} \times \mathcal{V}_i$, initially   $\{({\kappa}^0, v_i^0)\}$
\Statex $List$, a list  of  $\mathcal{K} \times \{ 0, 1 \}^k $, initially  $[({\kappa}^0, (1, \cdots 1))]$
            }\EndPart
            \Statex
            \Part {\underline{On recv (\writeValueTag, }
            	\\ \underline{$({\kappa}, v))$ from  $w$}} {
                       \State   $Vals \gets   Vals \cup \{({\kappa}, v)\}$ 
              \State  Send {\ackTag} to writer $w$.
            }\EndPart
            \Statex
            \Part{{\underline{On recv (\informSerializerTag,}
            	\\ \underline{$({\kappa}, (b_{1}, \cdots b_{k}))$) from  $w$}}}{ 
                 \State $List  \leftarrow$ 
                 \\~~~~~~$List \bigoplus~ ({\kappa}, (b_{1}, \cdots b_{k}))$ 
                  \\~~~~~~//$\bigoplus$ for append
                  \State $tag \leftarrow |List|$ /* $| \cdot |$ list size */
            \State Send  ({\ackTag}, $tag$) to  $w$
          }\EndPart
                   
                        \Statex
            \Part{ \underline{On recv ({\readValuesTag}) from  $r$ }} {
              \State   Send $Vals$  to  $r$
            }\EndPart 
          \Statex\Statex   /* used only by $s^*$ */
          
          \Part{ \underline{On recv {\getTagArrayTag} from  $r$ }} {
          \For{$i \in [k]$} 
          \State $j^* \leftarrow \max \{j:$ 
          \\~~~~~~~~$ List[j].b_i = 1\}$
          \State ${\kappa}_i \leftarrow List[j^*].{\kappa}$
           \EndFor 
           \Statex
             \State  $t_r \triangleq \max_{1 \leq j \leq |List|} \{j:$ 
              \\~~~~~~$List[j].b_i = 1 \wedge i \in I\}$ 
            \State  Send  ($t_r, ({\kappa}_1,{\kappa}_2, \cdots, {\kappa}_k)$) to $r$
            }\EndPart       
          }\end{multicols}
        \end{algorithmic} 
        \caption{Protocols for reader $r$ \& server $s_i$ in alg. $C$.}\label{fig:algo_c}
          \vspace{-1.2em}
      \end{algorithm}

\section{SNW + One Round,  MWMR setting}\label{mwmr_snow_one_round}
Here, we present  algorithm  $C$ which satisfies  SNW and "one-round"  properties in the MWMR setting, where a \textit{READ} consists one round of communications between the reader and the servers but servers may respond with multiple versions of the data. 
The notation for the  writers, servers and tag are similar to algorithm $B$. 
Pseudocodes~\ref{fig:algo_bc} and ~\ref{fig:algo_c} show the steps for the  writers, and the readers and the servers, respectively.
%
We designate a server as the coordinator,  denote as $s^*$. 

\textit{\textbf{State variables:}} The state variables are similar to  $B$.
			
\textit{\textbf{Writer steps:}} \emph{WRITE} transaction is similar to algorithm $B$..


\textit{\textbf{Reader steps:}}  
The step  \Readtr{$ o_{i_1},  o_{i_2}, \cdots, o_{i_p}$} can be 
initiated by  some reader  $r$ intending to read the values of 
subset $o_{i_1},  o_{i_2}, \cdots, o_{i_p}$ of the objects. 
Denote  $I \triangleq \{i_1, i_2, \cdots, i_p\}$  and $S_I\triangleq \{s_{i_1}, s_{i_2}, \cdots, s_{i_p}\}$.
The procedure 
consists of only one  phase  of communication round
between the $r$ and the  servers, called  {\readValuesAndTags}. 
During  {\readValuesAndTags},  $r$ sends $s^*$ the message  {\getTagArrayTag}  
 requesting the  list of the latest added keys for each object, and also sends
    requests $(\text{\readValuesTag})$  each server $s_i$ in $S_I$. Note that if $s^*$ is also one of the 
    servers in $S_I$ then the {\getTagArrayTag} and  {\readValuesTag} messages to $s^*$ can be combined to create one message; however, we keep them separate for clarity of presentation.
 Once $r$ receives a list of tags, such as, $(t_r, ({\kappa}_1, {\kappa}_2, \cdots,  {\kappa}_k))$ from $s^*$  and 
 the set of $Vals_{i}$ from each $s_i \in S_I$ then $r$ returns  the values 
 $v_{i_1}$, $v_{i_2}, \cdots v_{i_p}$ such that   $({\kappa}_{j}, v_{j}) \in Vals_{j}$, $j\in \{1, \cdots p\}$, and completes the 
\textit{READ}.

\textit{\textbf{Server steps:}}  
When a server $s_i$ receives a message $(${\writeValueTag}$, ({\kappa}, v_{i}))$ from a writer $w$ or 
 $s^*$, receives  (\updateCoordTag, $({\kappa}, (b_{1}, \cdots, b_{k})$) from writer $w$ or 
 receives  a message as  {\getTagArrayTag} $r$ the steps are similar to those in $B$. 
On the other hand, if any server $s_i$ receives a message  $(${\readValuesTag}$)$ from a reader $r$ then it responds to $r$ with  $Vals$. 
The following result states that $C$ respects SNW and ``one-round"  properties.

\begin{theorem} Any well-formed  and fair execution of  $C$, in the MWMR setting
  satisfies the SNW and "one-round" properties. 
\end{theorem}

\begin{proof} Below we show that algorithm $C$ satisfies the  SN$\bar{o}$W properties. 
	
	\noindent{\emph{\underline{S property:}}} 
	Let $\beta$ be any fair execution  of  $B$ and 
 suppose all clients in $\beta$ behave in a well-formed
manner. Suppose $\beta$ contains no incomplete transactions and let  $\Pi$ be the set of transactions in $\beta$.  We define an irreflexive partial ordering ($\prec$) in $\Pi$ as follows:  if $\phi$ and $\pi$ are any two distinct transactions in $\Pi$ then we say 
	$\phi \prec \pi$ if either $(i)$ $tag(\phi) < tag(\pi)$ or $(ii)$ $tag(\phi) = tag(\pi)$ and $\phi$ is a \emph{WRITE} and $\pi$ is a \textit{READ}. Below we prove the $S$  property of $B$ by showing that  properties $P1$, $P2$, $P3$ and $P4$ of Lemma~\ref{lem:equivalence_app} hold for $\beta$.  The properties $P1$-$P4$ can be proved to hold in a manner very similar to algorithm $B$ (Section~\ref{mwmr_snow_one_version}). Therefore, we avoid repeating them. 

	\noindent{\emph{\underline{N, $\bar{o}$ and W properties:}}}  Evident from an inspection of the algorithm.
	\end{proof}

%

\section{Conclusion}
We revisited the SNOW Theorem and when it is possible for READ transactions to have the same latency as simple reads.
We provided a new and more rigorous proof of the original result.
We also closed several open questions that were either explicitly posed by the original work or that emerged from our careful analysis.
We found that READ transactions can match the latency of simple reads when client-to-client communication is allowed in MWSR setting.
We found that they cannot and must have higher worst-case latency when client-to-client communication is disallowed or there are at least two readers.
We also presented the first algorithms that provide bounded worst-case latency for read-only transactions in strictly serializable systems with 
WRITE transactions.

\remove{
There are, however, several remaining important questions.
The first three questions focus on bounds for read-only transactions in strictly serializable systems with WRITE transactions.
First, is it possible to return exactly one version, always complete in two rounds, but that also can optimistically complete after one round?
Second, is it possible to always complete in one round and return a bound on number of versions that is less than the number of concurrent 
WRITEs?
Third, is there a way to enhance our model such that blocking could be bounded, and, if so, what bounds on blocking are possible?
Another open question is whether our new algorithms for READ transactions could be adapted to work in systems with READ-WRITE 
transactions that are more powerful than WRITE transactions? We conjecture they can be.

Our final open questions focus on ever more practical settings. 
coordinator, which in a real system would eventually become bottle-necked and restrict scaling to arbitrary sizes.
Is it possible to design algorithms that match our bounds on 
latency without relying on a centralized component and thus could scale arbitrarily?
Can one design a coordinator that can scale to arbitrary throughput?
}

\section*{Acknowledgment}
This work was supported by NSF awards CNS-1824130, CCF-2003830,  CCF-0939370 and NSF CCF-1461559.

\bibliographystyle{plain}
\bibliography{biblio}

\appendix




	\section{Algorithm Specification with I/O Automata}
	\label{subsec:model}
	We model a distributed algorithm using   the I/O Automata~\cite{Lynch1996}.
	Here we limit our discussion of I/O Automata to the relevant concepts, but for a detailed account the reader should refer 
	to ~\cite{Lynch1996}.
	An algorithm is a composition $\mathcal{A}$ of a set of \emph{automata} where each automaton  $A_i$ corresponds to a process (e.g., client or server) or a communication channel in the system.  $A_i$ is defined in terms of
	a set of deterministic transition functions  $\trans{A_i}$ (also called \textit{actions}), which can be thought of as the algorithmic steps of $A_i$;  and  a set of states $\states{A_i}$. 
	An execution of $\mathcal{A}$ is a sequence of alternating states and actions of $A$,  $\sigma_0, a_1, \sigma_1, a_2, \sigma_2, \ldots, \sigma_n$. 
	A state change, called a \textit{step}, is a 3-tuple $(\sigma_i, a_i, \sigma_{i+1})$, with $\sigma_i, \sigma_{i+1} \in \states{A_i}$ and $a_i \in \trans{A_i}$. 
	The set of input actions is denoted by $\inactions{A_i}$, e.g.,  $a_i \in \inactions{A_i}$ is an input action if it receives a message. The set of output actions is denoted by $\outactions{A_i}$. The input and output actions are also called external actions, denoted by $\extactions{A_i}$ and $\inactions{A_i} \cup \outactions{A_i}=\extactions{A_i}$. If an action $a_i \notin \extactions{A_i}$, then $a_i$ is an internal action.
	Communications between any two automata $A_i$ and $A_j$ is modeled by using channel automata $Channel_{i, j}$ for sending messages from $A_i$ to $A_j$; and $Channel_{j, i}$ for sending message from $A_j$ to $A_i$. When $A_i$ sends some
	message $m$ to $A_j$ the following sequence of actions occur: $send(m)_{i,j}$  occurs at $A_i$, then 
	$send(m)_{i,j}$ followed by $recv(m)_{i, j}$ occur at $Channel_{i,j}$; then finally, $A_j$ receives $m$ via the action $recv_{i,j}(m)$.
	In our model,  the communication channels are simple because we assume reliable communication between each pair of processes. Therefore, we ignore the actions in the $Channel_{i, j}$ and instead say   $send(m)_{i,j}$  occurs at $A_i$ and then $A_j$ receives $m$ via the action $recv_{i,j}(m)$.
	%
	An execution fragment $\alpha$ 
	can be either finite, i.e., having finite states, or infinite.
	If $\alpha$ is a finite execution and $\beta$ is an execution fragment, such that $\beta$ starts with the final state of $\alpha$ then we use $\alpha \circ \beta$ to denote the concatenation of $\alpha$ and $\beta$.
	If $\epsilon$ and $\epsilon'$ are two execution fragments, such that they have the same sequence of states at automaton $A_i$, i.e., $\epsilon|A_i = \epsilon'|A_i$, then $\epsilon$ and $\epsilon'$ are indistinguishable at $A_i$, denoted by $\epsilon \stackrel{A_i}{\sim} \epsilon'$. When the context is clear, we simply use $\epsilon \sim \epsilon'$.

For any execution of $\mathcal{A}$, 
$\finiteprefix{k-1}{k}\ldots$, where $\sigma$'s and $a$'s are states and actions,
we use the notation $\finiteprefixt{k-1}{k}\ldots$ that shows only the actions while leaving out the states to simplify notation. 


\section{Some Useful I/OA results}
Below we add some useful  theorems are useful related to  executions of a composed I/O Automata~\cite{Lynch1996}. 

\begin{theorem} \label{thm:paste} 
	Let $\{A_i\}_{i \in  I}$ be a compatible collection of automata and let $A = \Pi_{i \in I}A_i$. Suppose $\alpha_i$ is an execution of $A_i$ for every $i \in I$, and suppose $\beta$ is a sequence of actions in $ext(A)$ such that 
	$\beta | A_i = trace(\alpha_i)$ for every $i \in I$. Then there is an execution $\alpha$ of $A$ such that $\beta = trace(\alpha)$ and $\alpha_i = \alpha | A$, for every $i \in I$.
\end{theorem}

\begin{theorem} \label{thm:extension}
	Let $A$ be any I/O automaton.
	\begin{enumerate}
		\item [1.] If $\alpha$ is a finite execution of $A$, then there is a fair execution of $A$ that starts with $\alpha$.
		\item [2.]  If $\beta$ is a finite trace of $A$, then there is a fair trace of $A$ that starts with $\beta$.  
		\item [3.] If $\alpha$ is a finite execution of $A$  and $\beta$ is any finite or infinite sequence of input actions of $A$, then there is a fair execution  $\alpha \circ \alpha'$ of $A$ such that the sequence of input actions in $\alpha'$ is exactly $\beta$.
		\item [4.]  If $\beta$ is a finite trace of $A$ and $\beta'$ is any finite or infinite sequence of input actions of $A$, then there is a fair execution $\alpha \circ \alpha'$ of $A$ such that $trace(\alpha) = \beta$ and such that the sequence of input actions in $\alpha'$ is exactly $\beta'$.
	\end{enumerate}
\end{theorem}

The following useful  claim is adopted from  Chapter 16 of \cite{Lynch1996}.
\begin{claim}\label{claim:reorder}
	Suppose we have an automaton $A = \Pi_{i=1}^k A_i$ where $A$ is composed of the compatible collection of automata $A_i$, where $i \in \{1, \cdots, k\}$.
	Let $\beta$ be a fair trace of $A$ then we define an irreflexive partial order $\rightarrow_{\beta}$ on the actions of $\beta$
	as follows. If $\pi$ and $\phi$ are events in $\beta$, with $\pi$ preceding $\phi$, then we say $\phi$ depends on $\pi$, which we denote as $\pi \rightarrow_{\beta}\phi$, if one of the following holds:  
	\begin{enumerate}
		\item[1.] $\pi$ and $\phi$ are actions at the same automaton;
		\item[2.] $\pi$ is some $send(\cdot)_{j, i}$ at some $A_j$ and $\phi$ is some $recv(\cdot)_{j,i}$ at $A_i$; and
		\item[3.] $\pi$ and $\phi$ are related by a chain of the relations of items 1. and 2.
	\end{enumerate}
	Then if $\gamma$ is a sequence obtained by reordering the events in $\beta$ while preserving the $\rightarrow_{\beta}$, then $\gamma$ is also a fair trace of $A$.
\end{claim}
\begin{theorem}[\cite{Lynch1996}] \label{thm:fairtrace}
	Let $\{A_i\}_{i \in I}$ be a compatible collection of automata and let $A = \Pi_{i=1}^k A_i$. Suppose $\alpha_i$ is a fair
	execution of $A_i$ for every $i \in I$, and suppose $\beta$ is a sequence of actions  in $ext(A)$ such that $\beta|A_i=trace(\alpha_i)$
	for every $i \in I$. Then there is a fair execution $\alpha$ of $A$ such that $\beta=trace(\alpha)$ and $\alpha_i = \alpha|A$ for 
	every $i\in I$.
\end{theorem}

\end{document}